\documentclass[onecolumn, compsoc,10pt]{IEEEtran}

\usepackage{enumitem}
\usepackage{amssymb,amsfonts,amsmath,amsthm}
\usepackage{graphicx}
\usepackage{booktabs}
\usepackage[usenames,x11names, dvipsnames, svgnames]{xcolor}
\usepackage{amsmath,amssymb}
\usepackage{dsfont}
\usepackage{amsfonts}
\usepackage{mathrsfs}
\usepackage{float}
\usepackage{texshade}
\usepackage{hyperref}
\hypersetup{
  colorlinks=true,
  linkcolor=black,
  citecolor=black,
  filecolor=black,
  urlcolor=DodgerBlue4,
  breaklinks=false,
}
\usepackage{array}
\usepackage{xr}
\usepackage{verbatim}
\usepackage{multirow}
\usepackage{longtable}
\usepackage{tikz-network}
\usepackage{tikz}
\usepackage{pgfplots}
\usetikzlibrary{shapes,calc,shadows,fadings,arrows,decorations.pathreplacing,automata,positioning}
\usetikzlibrary{external}
\usetikzlibrary{decorations.text}
\usepgfplotslibrary{colorbrewer} 
\usepgfplotslibrary{statistics}

\tikzexternaldisable
\usepackage{geometry}
\usepackage{rotating}
\definecolor{nodecol}{RGB}{240,240,220}
\definecolor{nodeedge}{RGB}{240,240,225}
\definecolor{edgecol}{RGB}{130,130,130}
\tikzset{%
  fshadow/.style={      preaction={
      fill=black,opacity=.3,
      path fading=circle with fuzzy edge 20 percent,
      transform canvas={xshift=1mm,yshift=-1mm}
    }} 
}
\usetikzlibrary{pgfplots.dateplot}
\usetikzlibrary{patterns}
\usetikzlibrary{decorations.markings}
\usepackage{fancyhdr}
\usepackage{mathtools}
\usepackage{datetime}
\usepackage{comment}

\newcommand{\cgather}[2][\EQSP]{\begingroup\setlength\abovedisplayskip{#1}\setlength\belowdisplayskip{#1}\begin{gather} #2 \end{gather}\endgroup\noindent}

\newcommand{\calign}[2][\EQSP]{\begingroup\setlength\abovedisplayskip{#1}\setlength\belowdisplayskip{#1}\begin{align} #2 \end{align}\endgroup\noindent}

\newcommand{\mnp}[2]{\begin{minipage}{#1}#2\end{minipage}} 
\newtheorem{thm}{Theorem}

\newtheorem{defn}{Definition}

\newtheorem{rem}{Remark}

\newif\ifproof
\prooffalse %

\newcommand{\etal}{\textit{et} \mspace{3mu} \textit{al.}}

\newcommand{\HA}{\boldsymbol{\mathds{H}}}

\newcommand{\abs}[1]{\left\lvert #1 \right\rvert}%

\DeclareMathOperator*{\argmin}{arg\,min}

\usepackage[linesnumbered,ruled,vlined,noend]{algorithm2e}
\newcommand{\captionN}[1]{\caption{\color{darkgray} \sffamily \fontsize{9}{10}\selectfont #1  }}

\usepackage{txfonts}
\renewcommand{\rmdefault}{phv} %
\renewcommand{\sfdefault}{phv} %
\edef\keptrmdefault{\rmdefault}
\edef\keptsfdefault{\sfdefault}
\edef\keptttdefault{\ttdefault}

\usepackage[OT1]{fontenc}
\usepackage{concmath}
\edef\rmdefault{\keptrmdefault}
\edef\sfdefault{\keptsfdefault}
\edef\ttdefault{\keptttdefault}
\tikzexternalenable
\tikzfading[name=fade out,
inner color=transparent!0,
outer color=transparent!100]

\tikzset{wiggle/.style={decorate, decoration={random steps, amplitude=10pt}}}
\usetikzlibrary{decorations.pathmorphing}
\pgfdeclaredecoration{Snake}{initial}
{
  \state{initial}[switch if less than=+.625\pgfdecorationsegmentlength to final,
  width=+.3125\pgfdecorationsegmentlength,
  next state=down]{
    \pgfpathmoveto{\pgfqpoint{0pt}{\pgfdecorationsegmentamplitude}}
  }
  \state{down}[switch if less than=+.8125\pgfdecorationsegmentlength to end down,
  width=+.5\pgfdecorationsegmentlength,
  next state=up]{
    \pgfpathcosine{\pgfqpoint{.25\pgfdecorationsegmentlength}{-1\pgfdecorationsegmentamplitude}}
    \pgfpathsine{\pgfqpoint{.25\pgfdecorationsegmentlength}{-1\pgfdecorationsegmentamplitude}}
  }
  \state{up}[switch if less than=+.8125\pgfdecorationsegmentlength to end up,
  width=+.5\pgfdecorationsegmentlength,
  next state=down]{
    \pgfpathcosine{\pgfqpoint{.25\pgfdecorationsegmentlength}{\pgfdecorationsegmentamplitude}}
    \pgfpathsine{\pgfqpoint{.25\pgfdecorationsegmentlength}{\pgfdecorationsegmentamplitude}}
  }
  \state{end down}[width=+.3125\pgfdecorationsegmentlength,
  next state=final]{
    \pgfpathcosine{\pgfqpoint{.15625\pgfdecorationsegmentlength}{-.5\pgfdecorationsegmentamplitude}}
    \pgfpathsine{\pgfqpoint{.15625\pgfdecorationsegmentlength}{-.5\pgfdecorationsegmentamplitude}}
  }
  \state{end up}[width=+.3125\pgfdecorationsegmentlength,
  next state=final]{
    \pgfpathcosine{\pgfqpoint{.15625\pgfdecorationsegmentlength}{.5\pgfdecorationsegmentamplitude}}
    \pgfpathsine{\pgfqpoint{.15625\pgfdecorationsegmentlength}{.5\pgfdecorationsegmentamplitude}}
  }
  \state{final}{\pgfpathlineto{\pgfpointdecoratedpathlast}}
}
\newcolumntype{L}[1]{>{\rule{0pt}{2ex}\raggedright\let\newline\\\arraybackslash\hspace{0pt}}m{#1}}
\newcolumntype{C}[1]{>{\rule{0pt}{2ex}\centering\let\newline\\\arraybackslash\hspace{0pt}}m{#1}}
\newcolumntype{R}[1]{>{\rule{0pt}{2ex}\raggedleft\let\newline\\\arraybackslash\hspace{0pt}}m{#1}}

\def\DISCLOSURE#1{\def\disclosure{#1}}
\DISCLOSURE{\raisebox{15pt}{$\phantom{XxxX}$This sheet contains proprietary information 
    not to be released to third parties except for the explicit purpose of evaluation.}
}

\makeatletter
\pgfdeclarepatternformonly[\hatchdistance,\hatchthickness]{flexible hatch}
{\pgfqpoint{0pt}{0pt}}
{\pgfqpoint{\hatchdistance}{\hatchdistance}}
{\pgfpoint{\hatchdistance-1pt}{\hatchdistance-1pt}}%
{
  \pgfsetcolor{\tikz@pattern@color}
  \pgfsetlinewidth{\hatchthickness}
  \pgfpathmoveto{\pgfqpoint{0pt}{0pt}}
  \pgfpathlineto{\pgfqpoint{\hatchdistance}{\hatchdistance}}
  \pgfusepath{stroke}
}
\makeatother

\pgfdeclarepatternformonly{north east lines wide}%
{\pgfqpoint{-1pt}{-1pt}}%
{\pgfqpoint{10pt}{10pt}}%
{\pgfqpoint{9pt}{9pt}}%
{
  \pgfsetlinewidth{0.7pt}
  \pgfpathmoveto{\pgfqpoint{0pt}{0pt}}
  \pgfpathlineto{\pgfqpoint{9.1pt}{9.1pt}}
  \pgfusepath{stroke}
}

\pgfdeclarepatternformonly{north west lines wide}%
{\pgfqpoint{-1pt}{-1pt}}%
{\pgfqpoint{10pt}{10pt}}%
{\pgfqpoint{9pt}{9pt}}%
{
  \pgfsetlinewidth{0.7pt}
  \pgfpathmoveto{\pgfqpoint{0pt}{9pt}}
  \pgfpathlineto{\pgfqpoint{9.1pt}{-0.1pt}}
  \pgfusepath{stroke}
}
\makeatletter

\pgfdeclarepatternformonly[\hatchdistance,\hatchthickness]{flexible hatchB}
{\pgfqpoint{0pt}{\hatchdistance}}
{\pgfqpoint{\hatchdistance}{0pt}}
{\pgfpoint{1pt}{\hatchdistance-1pt}}%
{
  \pgfsetcolor{\tikz@pattern@color}
  \pgfsetlinewidth{\hatchthickness}
  \pgfpathmoveto{\pgfqpoint{0pt}{\hatchdistance}}
  \pgfpathlineto{\pgfqpoint{\hatchdistance}{0pt}}
  \pgfusepath{stroke}
}    \makeatother

\usetikzlibrary{arrows.meta}
\usetikzlibrary{decorations.pathreplacing,shapes.misc}
\usepgfplotslibrary{fillbetween}
\usetikzlibrary{shapes.geometric}
\usetikzlibrary{math}
\usepgfplotslibrary{colorbrewer} 

\usepackage{textcomp}
\usepackage{colortbl}
\usepackage{array}
\usepackage{courier} 
\usepackage{wrapfig}
\usepackage{pifont}
\usetikzlibrary{chains,backgrounds}
\usetikzlibrary{intersections}
\usetikzlibrary{pgfplots.groupplots}
\usepgfplotslibrary{fillbetween} 
\usetikzlibrary{arrows.meta}
\usepackage{pgfplotstable}
\usepackage[super,compress,sort,comma]{natbib}
\usepackage{setspace}
\usetikzlibrary{math}
\usetikzlibrary{matrix}
\usepackage{xstring}
\usepackage{xspace}
\usepackage{flushend}
\makeatletter
\renewcommand\section{\@startsection {section}{1}{\z@}%
  {-2ex \@plus -1ex \@minus -.2ex}%
  {1ex \@plus.1ex}%
  {\Large\bfseries\scshape}}
\renewcommand\subsection{\@startsection {subsection}{1}{\z@}%
  {-2ex \@plus -.25ex \@minus -.2ex}%
  {0.1ex \@plus.0ex}%
  {\fontsize{11}{10}\selectfont\bfseries\sffamily\color{black}}}
\renewcommand\subsubsection{\@startsection {subsubsection}{1}{\z@}%
  {0ex \@plus -.5ex \@minus -.2ex}%
  {0.0ex \@plus.5ex}%
  {\fontsize{9}{9}\selectfont\bfseries\itshape\sffamily\color{darkgray}}}
\renewcommand\paragraph{\@startsection {paragraph}{1}{\z@}%
  {-.2ex \@plus -.5ex \@minus -.2ex}%
  {0.0ex \@plus.5ex}%
  {\fontsize{9}{9}\selectfont\itshape\sffamily\color{darkgray}}}

\makeatother
\makeatletter
\pgfdeclareradialshading[tikz@ball]{ball}{\pgfqpoint{-10bp}{10bp}}{%
  color(0bp)=(tikz@ball!30!white);
  color(9bp)=(tikz@ball!75!white);
  color(18bp)=(tikz@ball!90!black);
  color(25bp)=(tikz@ball!70!black);
  color(50bp)=(black)}
\makeatother

\tikzexternaldisable 
\newcounter{Dcounter}
\newcommand{\qn}[1][i]{\Phi_{#1}}
\newcommand{\D}[1][i]{\mathscr{D}\left ( {\Sigma_#1} \right ) }
\newcommand{\Dx}{\mathscr{D}}
\def\J{\mathds{J}}
\def\M{\omega}

\newcommand{\mem}[1]{\M_{#1}}

\makeatletter
\newcommand\transformxdimension[1]{
    \pgfmathparse{((#1/\pgfplots@x@veclength)+\pgfplots@data@scale@trafo@SHIFT@x)/10^\pgfplots@data@scale@trafo@EXPONENT@x}
}
\newcommand\transformydimension[1]{
    \pgfmathparse{((#1/\pgfplots@y@veclength)+\pgfplots@data@scale@trafo@SHIFT@y)/10^\pgfplots@data@scale@trafo@EXPONENT@y}
}
\makeatother

\pgfplotsset{
    discard if/.style 2 args={
        x filter/.code={
            \edef\tempa{\thisrow{#1}}
            \edef\tempb{#2}
            \ifx\tempa\tempb
                \def\pgfmathresult{inf}
            \fi
        }
    },
    discard if not/.style 2 args={
        x filter/.code={
            \edef\tempa{\thisrow{#1}}
            \edef\tempb{#2}
            \ifx\tempa\tempb
            \else
                \def\pgfmathresult{inf}
            \fi
        }
    }
  }

\def\commatononei#1,{#1}
\def\commatononej#1,#2,{#1#2}
\def\commatonone#1{\expandafter\commatononei#1}
\def\commatononeT#1{\expandafter\commatononej#1}
\newcommand{\Sum} [2] {#1 + #2 = \the\numexpr #1 + #2 \relax \\}

\usepackage{sistyle}
\SIthousandsep{,}

\newcounter{Ccounter}
\def\HA{\mathcal{H}}
\def\NA{\mathcal{N}}
\def\allha{\mathds{E}^\mathcal{H}}
\def\allna{\mathds{E}^\mathcal{N}}

\newif\iftikzX
\tikzXtrue
\tikzXfalse

\newif\ifFIGS
\FIGSfalse 
\FIGStrue
\newif\ifdraftQ
\draftQtrue
\def\TITLE{\enet: A Digital Twin of Sequence Evolution for Scalable
Emergence Risk Assessment of Animal \infl Strains}

\def\authore{Kevin Yuanbo Wu}
\def\authora{ Jin Li}
\def\authorc{Aaron Esser-Kahn}
\def\authord{Ishanu Chattopadhyay}

\def\addressa{Department of Medicine, University of Chicago, IL, USA}

\def\addressc{Committee on Quantitative Methods in Social, Behavioral, and Health Sciences, University of Chicago, IL, USA}
\def\addressd{Pritzker School of Molecular Engineering, University of Chicago, Chicago, IL, USA}
\def\addresse{Committee on Immunology, University of Chicago, Chicago, IL, USA}
\newif\ifdraftQ
\draftQtrue
\draftQfalse

\title{\LARGE \TITLE}
\author{\sffamily  \fontsize{10}{12}\selectfont   \authore$^{1}$,\authora$^{1}$,  \authorc$^{2,3}$, and \authord$^{1,4,5\bigstar}$\\                                                                
\vspace{10pt}                                                                   

\sffamily  \fontsize{10}{12}\selectfont                                         
$^{1}$\addressa\\   
$^{2}$\addressd\\
$^{3}$\addresse\\
$^{5}$\addressc                                                                 
\vskip 1em                                                                      
$^\bigstar$To whom correspondence should be addressed: e-mail: \texttt{ishanu@uchicago.edu}.}

\def\hcov{SARS-CoV-2\xspace}
\def\RATG13{RaTG13\xspace}

\def\qnet{Enet\xspace}
\def\enet{Emergenet\xspace}
\def\qnet{\enet}%
\def\erisk{E-risk\xspace}
\def\qdist{E-distance\xspace}

\def\infl{Influenza A\xspace}

\def\dst{x_\star^{t+\delta}}
\def\dsta{x^{t+\delta}}

\definecolor{colh1n1x}{HTML}{551177}
\definecolor{colh3n2x}{HTML}{0077FF}
\definecolor{colh5n1x}{HTML}{991133}
\definecolor{colh7n9x}{HTML}{aa6622}
\definecolor{colh9n2x}{HTML}{8877FF}
\definecolor{colh3n3x}{HTML}{FF0000}
\definecolor{colh1n2x}{HTML}{337733}
 
\usepackage{flushend}
\newif\iftikzX
\tikzXtrue
\tikzXfalse
\def\EXTENDED{Extended Data\xspace}
\def\SUPPLEMENTARY{Supplementary\xspace}
\newif\ifFIGS
\FIGSfalse  
\FIGStrue 

\def\METHODS{Methods \& Materials\xspace}

\pgfplotsset{compat=1.16}

\begin{document}  
\maketitle 

{\bf \sffamily \fontsize{10}{12}\selectfont \noindent   
  {\normalfont \itshape Abstract:} Despite having triggered devastating pandemics in the past, our ability to quantitatively assess the emergence potential of individual strains of animal influenza viruses remains limited. This study introduces Emergenet, a tool to infer a digital twin of sequence evolution to chart how new variants might emerge in the wild. Our predictions based on Emergenets built only using 220,151 Hemagglutinnin (HA) sequences consistently outperform WHO seasonal vaccine recommendations for H1N1/H3N2 subtypes over two decades (average match-improvement: 3.73 AAs, 28.40\%), and are at par with state-of-the-art approaches that use more detailed phenotypic annotations. Finally, our generative models are used to scalably calculate the current odds of emergence of animal strains not yet in human circulation, which strongly correlates with CDC’s expert-assessed Influenza Risk Assessment Tool (IRAT) scores (Pearson’s $r = 0.721, p = 10^{-4}$). A minimum five orders of magnitude speedup over CDC’s assessment (seconds vs months) then enabled us to analyze 6,354 animal strains collected post-2020 to identify 35 strains with high emergence scores ($> 7.7$). The Emergenet framework opens the door to preemptive pandemic mitigation through targeted inoculation of animal hosts before the first human infection.}

\vspace{10pt} 
 
\subsection*{Introduction} 
Influenza viruses constantly evolve~\cite{dos2016influenza}, sufficiently altering surface protein structures to evade the prevailing host immunity, and cause the recurring seasonal  epidemic. These periodic  infection peaks claim a quarter to half a million lives~\cite{huddleston2020integrating} globally,  and currently our response hinges on annually  inoculating  the  human population with a  reformulated  vaccine~\cite{boni2008vaccination,dos2016influenza}.  Among numerous factors that hinder optimal design of the flu shot, failing to correctly predict the future frequency-dominant strains  dramatically reduces vaccine effectiveness~\cite{tricco2013comparing}. Despite  recent advances~\cite{neher2014predicting,huddleston2020integrating} such predictions remain imperfect. In addition to  the seasonal  epidemic, influenza strains spilling over into humans from animal reservoirs have triggered  pandemics  at least four times (1918 Spanish flu/H1N1, 1957 Asian flu/H2N2, 1968 Hong Kong flu/H3N2, 2009 swine flu/H1N1) in the past 100 years~\cite{shao2017evolution}. With the memory of the  sudden \hcov emergence
fresh in our minds, a looming question  is whether we can  preempt and mitigate such events in the future. \infl, partly on account of its segmented genome and its wide prevalence in common animal hosts, can easily incorporate genes from multiple strains and (re)emerge as novel human pathogens~\cite{reid2003origin,vergara2014ns}.  

\def\MXCOL{black}
\def\FXCOL{Orchid3}
\def\MNCOL{SeaGreen4}
\def\FNCOL{SeaGreen4}
\def\NCOL{SeaGreen4}
\def\XCOL{Tomato}
\def\WCOL{Tomato}
\def\YCOL{DodgerBlue4}
\def\TEXTCOL{gray}
\def\AXISCOL{white}
\ifFIGS
\begin{figure*}[!ht]
  \tikzexternalenable
  \tikzsetnextfilename{scheme}
  \centering
  \iftikzX  
  \def\EPATH{e2}
  \input{Figures/scheme}
 \else
  \includegraphics[width=\textwidth]{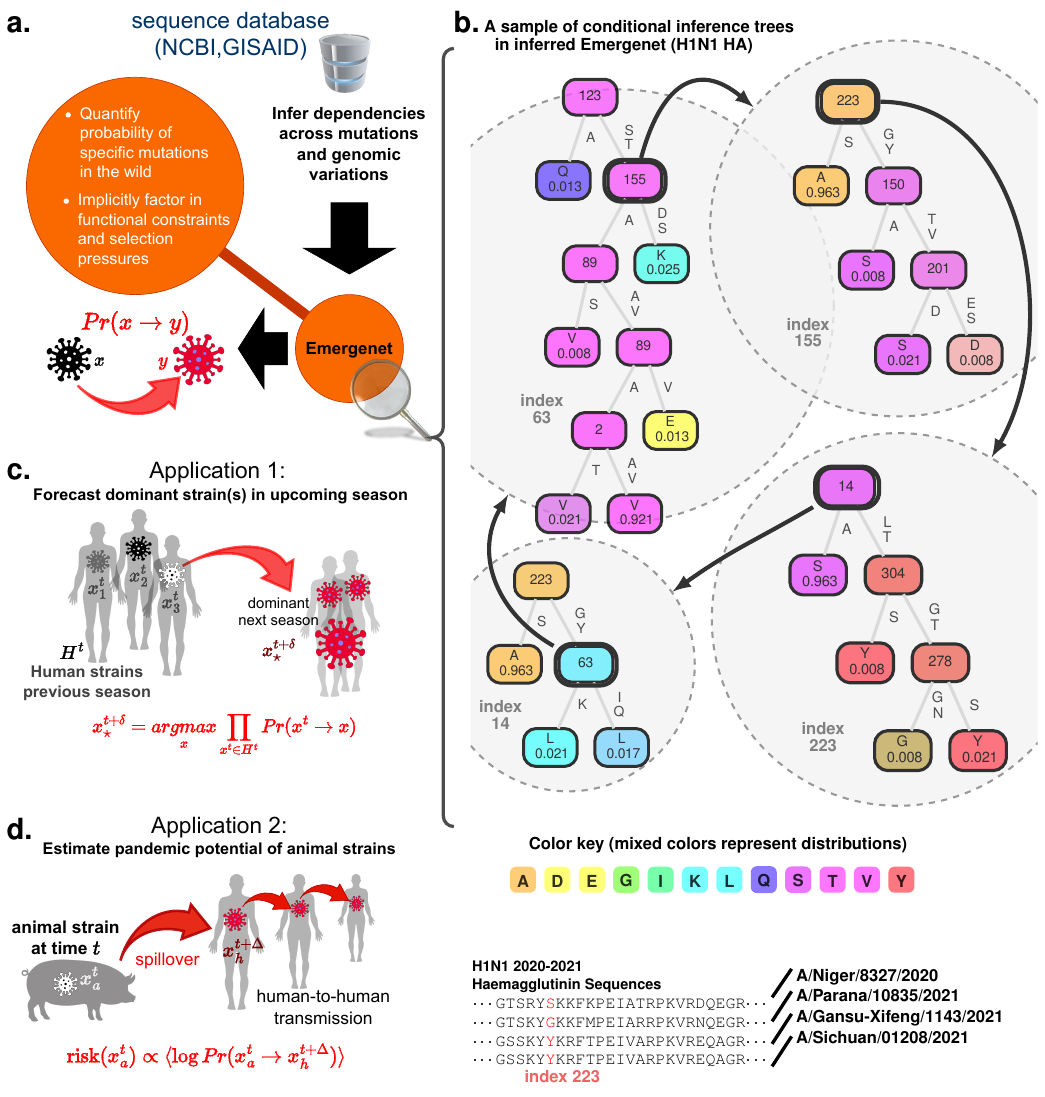}
  \fi
  \vspace{-20pt}
  
 \captionN{\textbf{\enet inference and applications}. \textbf{Panel a} Variations of genomes for identical subtypes of \infl are analyzed to infer a recursive forest of conditional inference trees~\cite{Hothorn06unbiasedrecursive} -- the \enet --   which maximally captures the emergent  dependencies between an a priori unspecified number of   mutations. With these inferred dependencies we can    estimate the numerical odds of specific mutations, and by extension, the numerical value of the probability of one strain giving rise to another in the wild, under  complex selection pressures from the background. \textbf{Panel b} Snapshot of decision trees  from the \enet inferred for H1N1 HA sequences collected in 2020-2021, which reveals a cyclic dependency. In general, every internal node of a component tree can be ``expanded'' into its own tree, underscoring the recursive structure of the \enet. \textbf{Panel c} First application: forecast  dominant strain(s) for the next flu season, using only  sequences collected up to six months prior and the inferred \enet, using data from the past year. \textbf{Panel d} Second application: estimation of the pandemic risk posed by individual animal strains that are still not known to circulate in humans.}\label{figscheme}.
\end{figure*}
\else
\refstepcounter{figure}\label{figscheme}
\fi

\ifFIGS
\begin{figure*}[!ht]
  \centering
  \tikzexternalenable
   \tikzsetnextfilename{seasonalpred_both}

  \iftikzX
  \begin{tikzpicture}
  \def\HGT{.350in}
  \def\WDT{2.75in}
  \def\YST{-.3in}

   \node[,label={[font=\bf\sffamily,yshift=-.60in]90:\underline{Single Recommendation}}] (AAA) at (0,0) {\input{Figures/seasonalpred_single.tex}};
 \node[anchor=north,label={[font=\bf\sffamily]90:\underline{Two Recommendations}}] (BBB) at ([yshift=-.25in]AAA.south) {\input{Figures/seasonalpred_two.tex}};
     \node[anchor=center,rotate=90,align=center] (Lh) at ([xshift=.35in]$(AAA.south west)!.5!(BBB.north west)$) 
   {\large Improvement in edit distance from future population};

\end{tikzpicture}
   \else 
   \includegraphics[width=\textwidth]{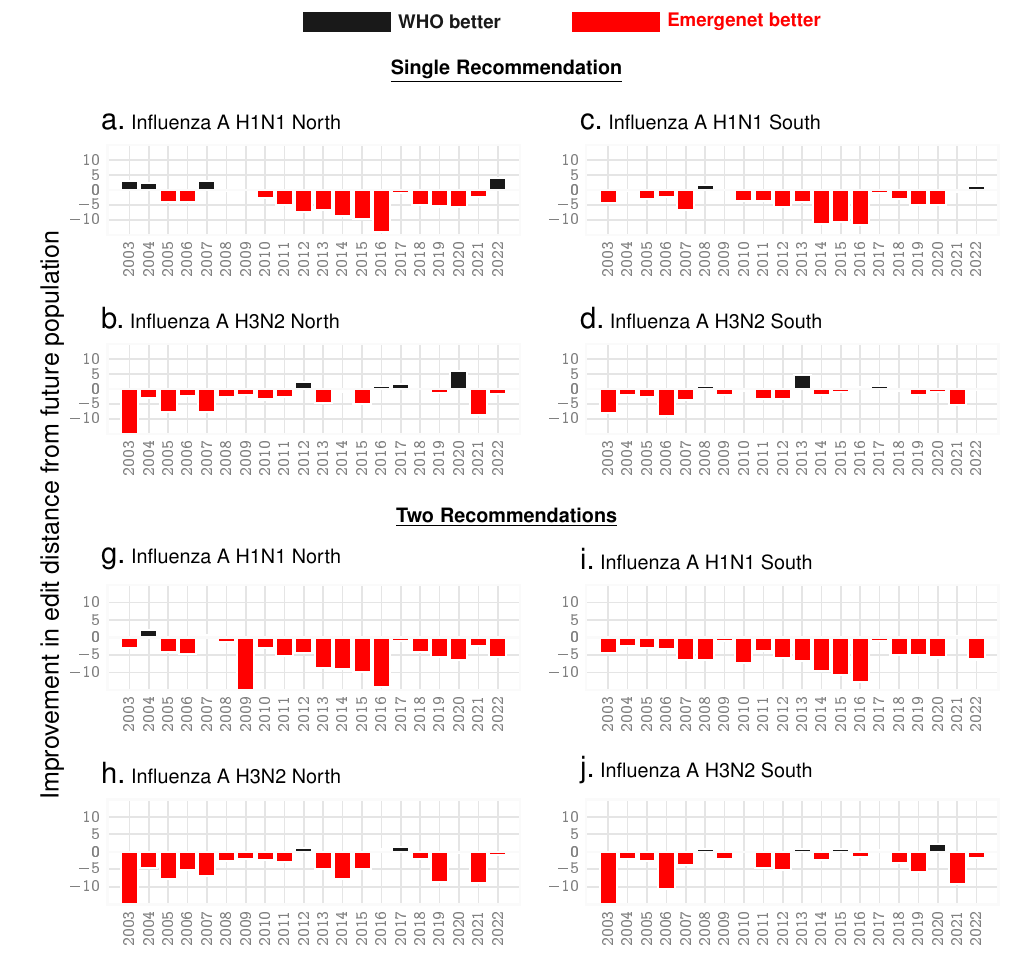}
   \fi
   \captionN{\textcolor{black}{\textbf{Seasonal predictions for Influenza A.} Relative out-performance of \enet predictions against WHO recommendations for H1N1 and H3N2 subtypes for Hemagglutinin (HA) over the both hemispheres. The negative bars (red) indicate the reduced average Hamming distance between the predicted sequence and the sequence population that season. Providing two recommendations shows a significant improvement over providing a single recommendation. Note that the recommendations for the north are given in February, while that for the south are given in September, keeping in mind that the southern flu season begins a few months earlier (e.g. for the 2022-2023 flu season, northern data is labelled `2022').}}\label{figseasonal}
\end{figure*}
\else
\refstepcounter{figure}\label{figseasonal}
\fi

\ifFIGS
\begin{figure*}[!ht]
  \tikzexternalenable
  \tikzsetnextfilename{figpred}
  \centering
  \iftikzX 
   \input{Figures/figpred_v3}
   \else 
   \includegraphics[width=\textwidth]{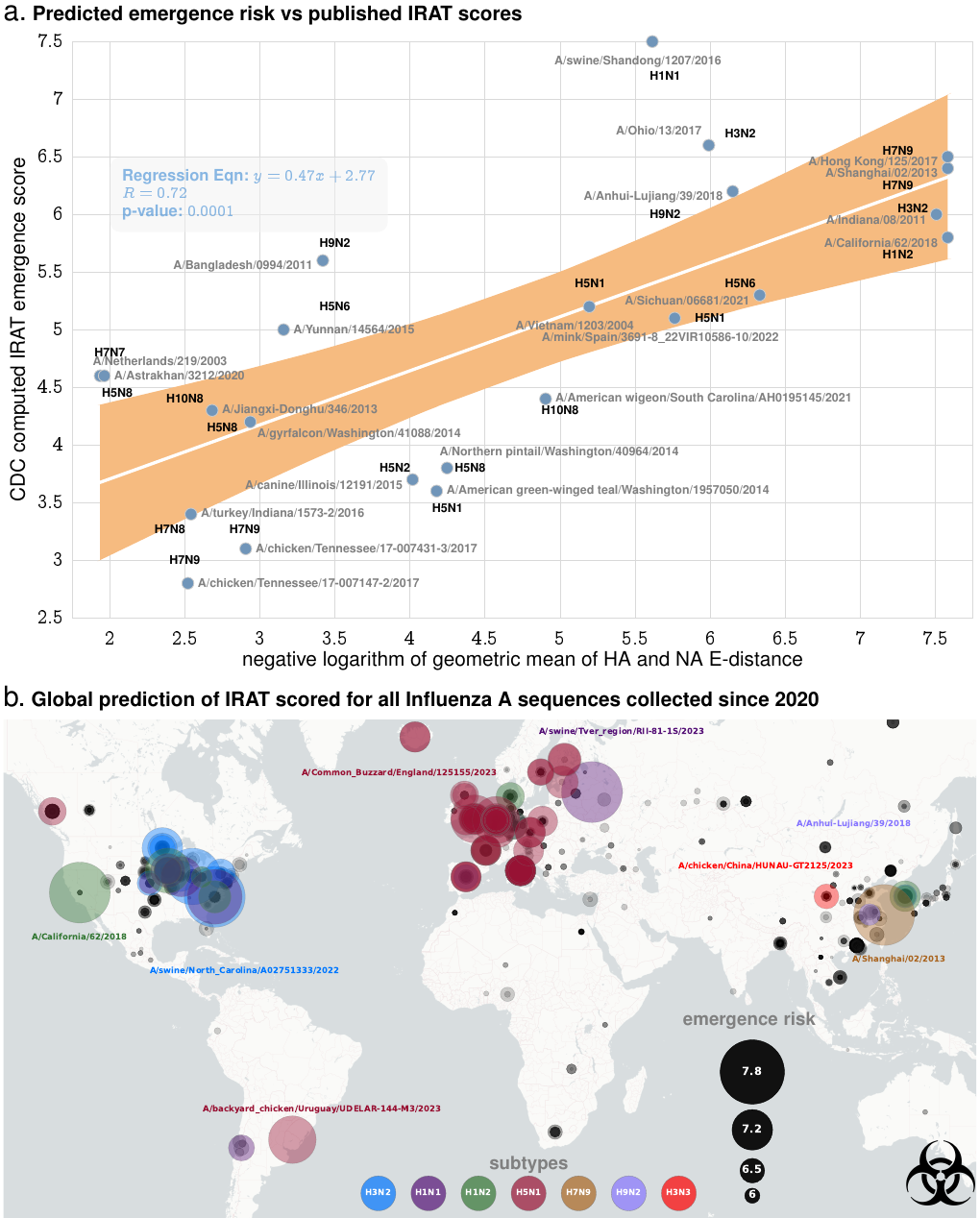}
   \fi
   \vspace{-7pt}
   
  \captionN{\textcolor{black}{\textbf{\enet based estimation of IRAT scores for animal strains. Panel a} We find  an approximate linear relationship between the negative logarithm of the geometric mean of the smallest \qdist for  HA and NA sequence of a target strain from circulating human strains in the year of estimation and the CDC published IRAT emergence scores.  \textbf{Panel b} Identifying risky \infl strains amongst those collected 2020-2023  using our approach.
  }}\label{figirat}
\end{figure*}
\else
\refstepcounter{figure}\label{figirat}
\fi
\ifFIGS

\begin{figure}[!ht]\centering
\centering
 \tikzexternalenable
  \tikzsetnextfilename{figphylo}
  \centering
  \iftikzX 
   \input{Figures/figphylo}  
   \else 
   \includegraphics[width=.78\textwidth]{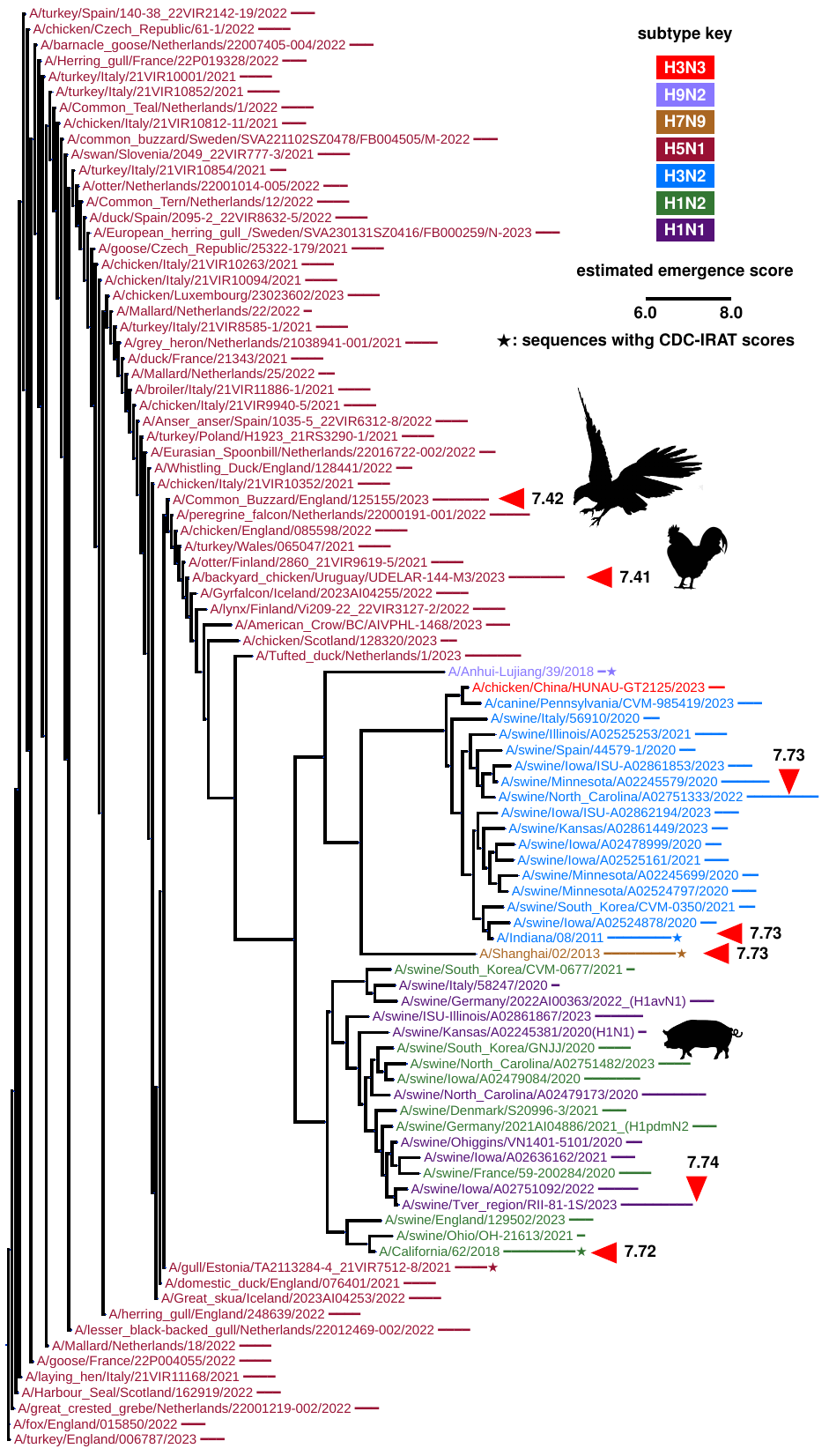}
   \fi
   \vspace{-10pt}

\captionN{\textcolor{black}{\textbf{Phylogeny constructed with edit distances,  with all  \infl strains collected post 2020, with estimated IRAT emergence risk $> 6.0$.} Strains with CDC-computed IRAT scores are also included (shown with $\bigstar$). Leaves have been collapsed which differ by less than 20 edits in the HA, displaying the most risky strains in the collapsed group,  which comes from diverse animal hosts, and geographic regions. Top strains are indicated by arrowheads.}}\label{figphylo}
\end{figure}
\else
\refstepcounter{table}\label{figphylo}
\fi

\ifFIGS

\begin{figure}[!ht]\centering
\centering
 \tikzexternalenable
  \tikzsetnextfilename{figshap}
  \centering
  \iftikzX 
   \input{Figures/figshap}  
   \else 
   \includegraphics[width=\textwidth]{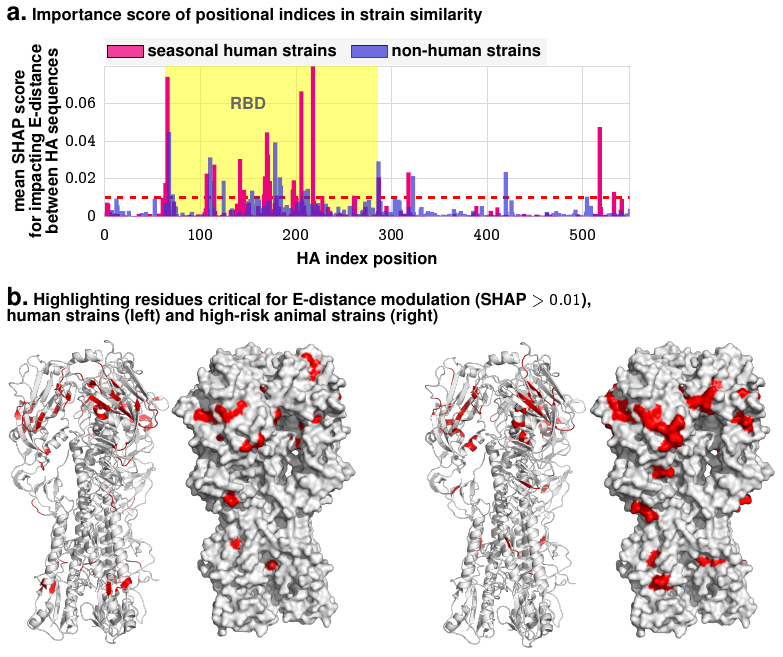}
   \fi
   \vspace{-10pt}

\captionN{\textbf{SHAP analysis of RBD.} Using a SHAP analysis we quantify the relative importance of individual residues in modulating the \qdist between two HA sequences, on average. \textbf{Panel a} The residues within the RBD is clearly having more impact for both seasonal \infl strains and high-risk animal strains (high-risk as determined by IRAT). \textbf{Panel b} Visualizing the residues on the molecular structure of HA (approximate, uses the same atom positions for all renderings (1ruz.pdb)), we find that the high SHAP residues (high importance in modulating \qdist between strains) are largely localized near the globular head domain, and have surface exposure. We also note that the seasonal and animal strains have similar but distinct high-value residues. }\label{figshap}
\end{figure}
\else
\refstepcounter{table}\label{figshap}
\fi

One possible approach to mitigating such risk is to identify  animal strains  that do not yet circulate in humans, but is likely to spill-over and quickly achieve human-to-human (HH) transmission capability. While global surveillance efforts collect wild specimens from diverse hosts and geo-locations annually, our  ability to objectively, reliably and scalably  risk-rank individual strains remains limited~\cite{wille2021accurately}, despite some recent progress~\cite{pulliam2009ability,grewelle2020larger,grange2021ranking}.
 
The Center for Disease Control's (CDC) current solution to this problem is the Influenza Risk Assessment Tool (IRAT)~\cite{Influenz24:online}.  Subject matter experts (SME) score strains based on  the number of  human infections, infection and transmission in laboratory animals, receptor binding characteristics, population immunity, genomic analysis, antigenic relatedness, global prevalence,  pathogenesis, and  treatment options, which are averaged to obtain two scores (between 1 and 10) that  estimate 1) the emergence  risk and 2) the potential public health impact on sustained transmission. IRAT scores  are potentially subjective, and  depend on multiple experimental assays, possibly taking  weeks to compile for a single strain. This results in  a scalability bottleneck, particularly with    thousands of strains being sequenced annually.

Here we introduce a pattern recognition algorithm to automatically parse out emergent evolutionary constraints operating on \infl viruses in the wild, to provide a less-heuristic, theory-backed scalable solution to emergence prediction. Our approach is centered around numerically estimating the probability $Pr(x \rightarrow y)$ of a strain $x$  giving rise to  $y$ which is well-adapted to humans. We show that this capability is key to preempting  strains which are expected to be in future circulation, and  %
approximate IRAT scores of non-human strains without  experimental assays or SME scoring.

\subsection*{\enet: Digital Twin of Sequence Evolution in the Wild}
To uncover relevant evolutionary constraints, we analyze  variations (point substitutions and indels) of the  residue  sequences  of key proteins implicated  in cellular entry and exit~\cite{gamblin2010influenza,shao2017evolution}, namely HA and NA respectively. By capturing the emergent cross-dependencies between these sequence variations our model, the \enet (Enet), estimates the  odds of a specific mutation to arise in future, and consequently the probability of a specific strain   evolving into another (Fig.~\ref{figscheme}a).  Our calculations are are based on first inferring the variation of mutational probabilities, and the potential residue replacements which vary along a sequence. The many well-known classical  DNA  substitution models~\cite{posada1998modeltest} or standard phylogeny inference tools which assume a constant species-wise mutational characteristics,  are not applicable here. Similarly, newer algorithms such  as FluLeap~\cite{eng2014predicting}  which identifies host tropism from sequence data, or estimation of species-level risk~\cite{grange2021ranking} do not allow for strain-specific assessment.

The dependencies we uncover are shaped by  a  functional necessity of conserving/augmenting  fitness. Strains must be sufficiently common  to be recorded, implying that the sequences from public databases that we train  with have  high replicative fitness. Lacking kinetic proofreading, \infl integrates  faulty nucleotides   at a relatively high rate ($10^{-3}-10^{-4}$) during  replication~\cite{ahlquist2002rna,chen2006avian}. However, few variations are actually viable, manifesting discernible  dependencies between such mutations. These fitness constraints also vary  over time. The background strain distribution, and selection pressure from the evolution of cytotoxic T lymphocyte  epitopes~\cite{woolthuis2016long,fan2012role,van2016differential,berkhoff2007assessment,van2012evasion} in humans can change quickly. However, with a sufficient number of unique samples to train on for each flu season, the \enet (recomputed for each time-period) is expected to automatically factor in the evolving host immunity, and the current background environment.  

Structurally, an \enet comprises an interdependent collection of  local predictors, each aiming to predict the  residue at a particular index  using as features  the residues   at other  indices  (Fig.~\ref{figscheme}b). Thus,  an \enet comprises almost as many such  position-specific predictors as the length of the sequence. These individual predictors are implemented as conditional inference trees~\cite{Hothorn06unbiasedrecursive}, in which  nodal splits  have  a minimum pre-specified significance in differentiating the  child nodes. Thus, each predictor yields an estimated conditional residue distribution  at each index. The set of residues acting as features in each predictor are automatically identified, \textit{e.g.}, in the fragment of the  H1N1 HA \enet (2020-2021, Fig.~\ref{figscheme}b), the predictor for residue 63 is dependent on   residue  155, and the predictor for  155 is dependent on  223, the predictor for  223 is dependent on  14, and the residue at  14 is again dependent on  63, revealing a cyclic dependency. The complete \enet harbors a vast number of such  relationships, wherein each internal node of a tree may be  ``expanded'' to its own tree. Owing to this recursive expansion,  a complete \enet substantially captures the complexity of the rules guiding evolutionary change as evidenced by our out-of-sample validation.

Before we apply our approach to quantifying risk from animal strains, we demonstrate identification of vaccine
strains for the seasonal flu epidemics, with only sequence information, with performance at par with approaches
using deep mutational scanning (DMS) assays~\cite{huddleston2020integrating}.

For that purpose,  we used  H1N1 and H3N2 HA sequences from \infl strains in the public NCBI and GISAID databases recorded between 2001-2023 ($220,151$ in total, \SUPPLEMENTARY Table~S-\ref{tabseq}). We  construct \enet{s} separately for both subtypes, yielding $84$ models for predicting seasonal dominance. 
The \enet predictors induce an intrinsic distance metric (\qdist) between strains (Eq.~\eqref{q-distance} in \METHODS), defined as the square-root of the Jensen-Shannon (JS) divergence~\cite{cover} of the conditional residue distributions, averaged over the sequence. Unlike the classical approach of measuring the number of edits between sequences, the \qdist is informed by the \enet-inferred  dependencies, and adapts to the specific subtype, allele frequencies, and environmental variations.  We show  (Theorem~\ref{thmbnd} in \METHODS)   that the \qdist  scales as  the log-likelihood of spontaneous change $i.e.$ $\log Pr(x \rightarrow y )$.

Phenotypic differences are often driven by specific loci, and antigenic differences in influenza are typically strongly affected by changes within the receptor binding domain (RBD)~\cite{caton1982antigenic, koelle2006epochal}. With the complex cross-talk between loci uncovered by the \enet, evaluating the impact of individual loci in modulating phenotypic characteristics becomes  non-trivial.   We investigated the relative importance of specific  loci in modulating the \qdist  between strains,  and consequently  phenotypic changes, via  a standard SHAP analysis~\cite{lou2012intelligible,lundberg2017unified} (Fig.~\ref{figshap}). Such analyses, rooted in cooperative game theory,  has been successfully used in the literature to uncover average impact of individual inputs in complex machine learning models~\cite{shapley1953value}.   We first estimate the impact of each  residue in determining the \qdist between a given strain, and a ``consensus'' strain, and then calculate the  average impact of the individual residues  over a representative set of strains. Our results indicate  that in both seasonal analysis and emergence risk of animal strains, we have 1) a complex distribution of important residues  along the HA sequence, 2) residues within the RBD are clearly playing dominant roles  (Fig.~\ref{figshap}a), 3) relatively few residues among the 566 possible locations are driving most of the effect (Fig.~\ref{figshap}b).

As a baseline comparison without using SHAP analysis, we also estimated  the correlation of \qdist between animal  and  common  human strains (\EXTENDED Table~\ref{baselinetab}), which shows that edit distance  on the RBD is strongly correlated to \qdist, both for the CDC-analyzed strains and animal strains collected within 2020-2022 period, but were not informative or trended incorrectly when we looked at selected individual residues of known functional relevance, suggesting complicated impact sharing by individual loci.

Note that despite general correlation between \qdist and edit-distance, the \qdist between fixed strains can change if only the background environment changes (\SUPPLEMENTARY Table~S-\ref{tabex}, S-\ref{tabcor}) with the strains remaining unchanged. Thus, we re-learn models using recent historical circulation, and only predict for the near future. In in-silico experiments, we find that while random mutations to genomic sequences produce rapidly diverging sets, \enet-constrained replacements produce verifiably meaningful sequences (See ``In-silico Corroboration of Emergenet’s Capability To Capture Biologically Meaningful Structure" in \METHODS and \SUPPLEMENTARY Fig.~S-\ref{figsoa}).

We then investigated how sampling of the observed strains, which naturally arises by how strains are collected,  affects model training (\SUPPLEMENTARY Fig.~S-\ref{randomsampling}). Fixing a particular time-period (2018-19 northern hemisphere shown), we trained 100 \enet models with different random samples of 3000 strains, picked two random strains from two of the largest clusters, and computed their \qdist under each of the 100 models. We see that the \qdist{s} are relatively stable, with distances between two strains in the same cluster  substantially smaller  than the inter-cluster distances, demonstrating model stability (and that cluster identity is largely preserved under such variations). Later we also investigate that the uncertainty from the  limited sampling of animal strains leads to small standard error of mean (SEM) in the estimated emergence scores.

Finally, we note that while   learning a precise model of jump from one strain to another is  infeasible from the available data, our approach can feasibly uncover a subset of the transition rules  that are best supported by statistical evidence, allowing us to probabilistically constrain the set of future possibilities. Viewing each locus in the  sequence individually as a target variable, and using all other loci as features to obtain a classification model, results in $\leqq 566$ conditional inference trees (for HA), as described above. Using conditional inference trees for the individual classification models implies that  only statistically significant splits are allowed in the inferred decision trees. When no such splits are found, we do not return a classification model for that index. Finally, this forest of trees collectively comprise a \enet model. Each such model  induces a distance metric between strains, referred to as the \qdist metric, which is  shown (Theorem~\ref{thmbnd} in \METHODS) to scale with the log-likelihood of the  jump probability. Estimating the numerical odds of a  jump $ Pr(x \rightarrow y)$ (Fig.~\ref{figscheme}) allows us to derive quantitative solutions to  the problem of forecasting  vaccine strain(s), and that of estimating the  emergence potential of an animal strain  (Fig.~\ref{figscheme}c-d, Eq.~\eqref{dompred0}-\eqref{eqrho}).

\subsection*{Predicting Seasonal Vaccine Strains}
A vaccine strain for an upcoming  season may be identified as one which maximizes the joint probability of simultaneously arising from each (or most)  of the currently circulating strains (Fig.~\ref{figscheme}c). This might not be a correct mechanistic picture per se (in which clades and their relative fitness might be more important~\cite{huddleston2020integrating}), but can still serve as an useful conceptualization if we are interested in estimating  the vaccine strain, and not  the full strain distribution. The vaccine strain identified in this manner does not deterministically specify a frequency-dominant strain, but a strain satisfying this criterion  has  high odds of being close to the distribution of circulating strains that emerge in the targeted future time-frame. Simplifying, we ultimately obtain  the following search criteria (See derivation in ``Predicting  Seasonal Vaccine Strains" in \METHODS) to identify  historical strain(s) that are  expected to be maximally representative of the future  distribution of circulation:
\calign{\label{dompred0}
&\dst = \argmin_{y \in \cup_{\tau \leqq t} H^\tau}  \left ( \sum_{x\in H^t}  \theta_{m_t}(x,y) - \abs{H^t}A \ln \mem{y}  \right )
}%
where $\dst$ is a predicted vaccine strain  at time $t+\delta$, $H^t$ is the set of currently circulating human strains at time $t$  observed over the past year, $\theta_{m_t}$ is the \qdist induced  by  \enet $m_t$ inferred with sequences in $H^t$, $\mem{y}$ is the persistence probability (See Def.~\ref{defmem} in \METHODS) of strain $y$ in the inferred \enet which can be shown to be proportional to its replicative fitness (See ``Relating Persistence Probability to Replicative Fitness" in \METHODS), and $A$ is a constant dependent on the sequence length and significance threshold. Thus, we have a straightforward interpretation for Eq.~\eqref{dompred0}: the first term requires the solution to be as close as possible to  the centroid of the current strain distribution (in the \qdist metric), while  the second term penalizes for low replicative fitness (See ``Predicting Seasonal Vaccine Strains'' in \METHODS). Eq.~\eqref{dompred0} does not attempt to replicate the future strain distribution, and thus is not claimed to be a mechanistically correct picture of the evolutionary process; but  a computational criteria for identifying a vaccine strain candidate.

For validation of this scheme, we define the distance to the future circulation at  time $t+\delta$ for sequence $x$ as the average Hamming  distance, $h$, between $x$ and each sequence $y$ in the future population $H^{t+\delta}$, weighted by its frequency $f_{t+\delta}(y)$, following recent evaluation approaches in the literature~\cite{huddleston2020integrating}. We predict one year into the future, $\delta=1$, and thus:
\calign{\label{dist}
&d_{\delta}(x) = \sum_{y \in H^{t+\delta}} f_{t+\delta}(y)h(x, y)
}%
In Eq.~\eqref{dompred0}, we are making the assumption that $H^t$ is a single cluster of strains. In practice, this might not be the case; we often see several distinct clusters arise in each season, which are functionally similar to the notion of clades~\cite{huddleston2020integrating}. Identifying clusters  on the \qdist matrix $H_1^t, H_2^t,\dots H_n^t$ with $\bigcup_{i=1}^nH_i^t = H^t$, we  compute cluster-specific predictions $x_{i*}^{t+\delta}$ using Eq.~\eqref{dompred0}, and obtain the unique vaccine recommendation as the weighted centroid of these cluster-specific recommendations, where the weights are the cluster sizes, $i.e.$:
\cgather{
\dst = \argmin_{x \in \{x_{i*}^{t+\delta}\}} \sum_{i=1}^n h(x_{i*}^{t+\delta},x)|H_i^t|
  }%
Using instantaneous clusters instead of  standard clades allows for simplified calculations with no phylogenetic tree construction. This approach,  which while possibly  losing some information related to common ancestors,  actually improves predictive performance. Also, we can make multiple predictions per season, $e.g.$, reporting the recommendations of the two largest clusters, as is shown in Fig.~\ref{figseasonal}g-j, which dramatically improves predictive performance as expected.

Prediction of the seasonal vaccine strain as  a close match to a historical strain  allows out-of-sample validation against past World Health Organization (WHO) recommendations for the flu shot, which  is  reformulated February and September of each year for the northern and southern hemispheres, respectively, based on a  cocktail of historical strains determined via global surveillance~\cite{agor2018models}. For each year of the past two decades, we constructed a separate \enet for the southern and northern flu seasons using HA strains available from the previous season. For seasons with $>3000$ strains available, we randomly sampled $3000$ strains which provide an accurate representation of the population. We show in \SUPPLEMENTARY Fig.~S-\ref{randomsampling} that sampling has little effect on our results.

Our \enet-informed forecasts outperform  WHO/CDC recommended flu vaccine compositions consistently over the past two decades, for both H1N1 and H3N2 subtypes, across both hemispheres (which have distinct recommendations~\cite{boni2008vaccination}). Over the last two decades, the  \enet H1N1 recommendations $\dst$ outperform WHO's vaccine recommendations by, on average, $3.73$ amino acids (AAs) in the north and $4.19$ AAs in the south. For H3N2, this out-performance is $3.32$ AAs in the north and $2.01$ AAs in the south. Detailed predictions are tabulated in \EXTENDED Tables~\ref{tabrec0} through~\ref{enetimprovement}, including averages over the last decade. Visually, Fig.~\ref{figseasonal} illustrates the relative \enet improvement.

We also compare our performance with recently reported vaccine-strain prediction results~\cite{huddleston2020integrating} at identical time points with better or comparable performance. Huddleston $\etal$'s model  using mutational load and local branching index (LBI) outperforms the WHO vaccine by $1.93$ AAs on average, while  over the same time points, we outperform the WHO vaccine by $2.05$ AAs. Our ability to obtain comparable performance with only sequence features  demonstrates the added predictive information that the \enet is able to extract. The full comparison can be found in \EXTENDED table~\ref{huddlestoncomparison}. While using  Hemagglutinnin inhibition (HI) antigenic novelty data Huddleston $\etal$ outperforms the WHO vaccine by $2.33$ AAs on average, this added information requires experimental evaluation of individual strains precluding scale-up and presenting challenges in generalizing to other viruses.

To evaluate  a strategy of two-strain predictions, we  evaluated two-cluster predictions $x_{1*}^{t+\delta}$ and $x_{2*}^{t+\delta}$ by taking the average minimum Hamming distance between $x_{1*}^{t+\delta}$ and $x_{2*}^{t+\delta}$ and the sequence population, we see a performance improvement. Over the last two decades, the two-cluster \enet H1N1 recommendations $\dst$ outperform WHO's vaccine recommendations by $4.69$ AAs in the north and $5.52$ AAs in the south. For H3N2, this out-performance is $4.64$ AAs in the north and $3.43$ AAs in the south.

To confirm the robustness of our method, we also ask how our recommendations perform against selecting random strains. We select a random strains from each season to be our random ``predictions'', choosing from strains circulating in the past one year leading up to vaccine selection, and compare them using our distance to the future metric (Eq.~\eqref{dist}). We repeat this for 20 random strains per season and report the mean.
Over the last two decades, the \enet H1N1 recommendations $\dst$ outperform the random recommendations by $3.98$ AAs (29.69\%) in the north and $7.31$ AAs (46.03\%) in the south. For H3N2, this out-performance is $2.94$ AAs (37.37\%) in the north and $3.23$ AAs (35.54\%) in the south (\EXTENDED Table~\ref{enetimprovementrandom}). Over one decade the \enet out-performance is higher for H1N1 ($>70\%$), while for H3N2 it remains comparable ($>30\%$).

Comparing the \enet inferred strain (ENT) against the one recommended by the WHO, we find that the residues that only the  \enet recommendation matches correctly with a common strain denoted as DOM (near dominant or strain observed with maximal frequency), while the WHO recommendation fails,  are largely localized within the RBD, with $>57\%$ occurring within  the RBD on average (\EXTENDED Fig.~\ref{figseq}a), and when the WHO strain deviates from  the ENT/DOM   matched residue, the ``correct'' residue is often replaced  in the WHO recommendation with one that has very different side chain, hydropathy  and/or chemical properties (\EXTENDED Fig. \ref{figseq}b-f), suggesting deviations in recognition characteristics~\cite{carugo2001normalized,righetto2014comparative}. Combined with the fact that we find circulating strains are almost always within a few edits of the DOM (\SUPPLEMENTARY Fig.~S-\ref{figdom}), these observations suggest that  hosts vaccinated with the ENT recommendation is can have season-specific antibodies that recognize a larger cross-section of the circulating strains.

\subsection*{Estimating Emergence Risk of Non-human Strains}
Our primary claim is the ability to estimate the emergence potential of  animal strains that do not yet circulate in humans.  While modeling approaches for seasonal vaccine strain choice has been explored before, the problem of modeling  emergence risk for animal strains is largely open.   We propose a  time-dependent  Emergence Potential score \erisk $\rho_t(x)$ for an animal strain $x$ in Eq.~\eqref{eqrho}. Our risk measure is defined as the negative logarithm of the geometric mean of minimum \qdist to recent  human HA and NA sequences, where the \qdist correspond to \enet{s} constructed in a subtype specific manner from human sequences.

Before  going into details,  we need some more notation, which we keep consistent as much as possible with Eq.~\eqref{dompred0}. For any strain $y$, we denote the HA and NA sequences as $y^{\HA},y^{\NA}$ respectively, and for a \qdist function induced for  \enet $r$ is denoted as $\theta_r(\cdot,\cdot)$. To validate our score against CDC-estimated IRAT emergence scores using \erisk, we constructed \enet models for HA and NA sequences using subtype-specific human strains, collected within the  year prior to the assessment date, \textit{e.g.},  the  assessment date for A/swine/Shandong/1207/2016 is 06/2020, and  we  use all human strains collected  between 01/07/2019 and 06/30/2020 for the \enet inference. More specifically, for each HA subtype for which sufficient number of sequences are collected in the span of one year leading up to that point  ($e.g.$ H1Nx, H3Nx, etc.), we construct a separate  \enet model. The set of all \enet models so inferred for HA strains are denoted as  $\allha$, and  all models for NA sequences are denoted as  $\allna$.  With these notations, we define:
\cgather{\label{eqrho}
\rho_t(x) \triangleq -\log \min_{\begin{subarray}{c}y,z \in H^t \\ r \in \allha,s \in \allna\end{subarray}} \sqrt{\theta_r\left (x^{\HA},y^{\HA}\right )  \theta_{s}\left (x^{\NA},z^{\NA}\right )}, \textrm{ where } H^t \textrm{ defined as before in Eq.~\eqref{dompred0}}
}%
We claim that larger the  value of $\rho_t(x)$  larger the  probability of jumping to a  well-adapted human strain (See ``Measure of Emergence Potential" in \METHODS). Strains observed to circulate in humans ( $i.e.$ those in $H^t$) are obviously well-adapted to humans, since we are only likely to sample common strains. Using Theorem~\ref{thmbnd} (See \METHODS):
\cgather{
  \left \lvert \ln  \frac{Pr(x \rightarrow x_h)}{Pr(x_h \rightarrow x_h)} \right \rvert \leqq \beta \theta(x,x_h),  
}%
where $x, x_h$ are respectively protein sequences (HA or NA) corresponding to an animal and a well-adapted human strain, and $\beta$ is a constant. Importantly, when $\theta(x,x_h)$ is close to zero and $Pr(x_h \rightarrow x_h)$ is large, it follows that  $Pr(x \rightarrow x_h)$ is large as well, justifying the claim (See ``Measure of Emergence Potential" in \METHODS). Thus, in both vaccine strain prediction and emergence risk estimation, our approach is predicated on our inferred models to correctly relate \qdist $\theta(x,y)$  to the probability $Pr(x \rightarrow y)$, and hence good performance on the first problem gives confidence about the results in the second. The geometric mean ensures equal weighting  of  both HA and NA (this is the simplest assumption), and the   logarithm in Eq.~\eqref{eqrho} ensures that scale of our measure matches with that of CDC-estimated IRAT scores, and this comparison provides quantitative validation of our framework.
 
  Our approach ensures that, we can  obtain a \erisk score for animal sequences that do not have any human sequences of the same subtype. For example, $\allha$ comprises \enet models for H1,H3, H5, H7, and $\allna$ comprises models for N1,N2,N9 (for some years H5 \enet is not constructed due to not enough sequences found in the previous one year).
Now suppose we observe a (yet non existent)  animal strain H14N12. We will carry out the minimizations stated in Eq.~\eqref{eqrho} with the HA sequence of this strain considering all human HA sequences observed in the past year using each of the HA \enet{s} in $\allha$, and with the NA sequence of this strain again against all human NA sequences of strains observed within the past year using each of the NA \enet{s} in $\allna$: it does not matter that we do not have a \enet for H14N12, or that this strain is yet to be observed in humans.

 Considering IRAT emergence scores of $23$ strains published by the CDC, we find strong out-of-sample support  (total least squares correlation: $0.721$, p-value: $0.00010$, Fig.~\ref{figirat}) for this claim. We also fit total least squares models against IRAT's ``Mean Low Acceptable Emergence", a lower bound on the estimate of the emergence score (correlation: $0.685$, p-value: $0.00061$) and ``Mean High Acceptable Emergence", an upper bound (correlation: $0.733$, p-value: $0.00016$). We take approximately $30$ seconds (on a 28-core Intel processor) as opposed to potentially weeks taken by IRAT experimental assays and SME evaluation. Full results can be found in \EXTENDED Table~\ref{irattab}.

Note that the nature of risk implies that a
``high-risk" strain is not guaranteed to emerge, and neither IRAT nor our assessment should be dismissed on the basis that a high-risk strain did not cause a pandemic. However, to interpret our high correlative association with  the CDC-computed IRAT scores as validation of our claim, there needs to be some evidence that the IRAT scores themselves are good predictors of emergence. In addition to  the CDC's expert-consensus driven methodology being the best alternative to direct gold-standard checks (which are impossible for ethical and legal reasons), we also have evidence that some  high-risk strains did corroborate with emergence events (\SUPPLEMENTARY Table~S-\ref{iratoutbreak}), including  A/swine/Shandong/1207/2016, A/Ohio/13/2017, A/Hong Kong/125/2017, and A/California/62/2018; the last one was indeed  isolated in human hosts~\cite{cdph_immunization_2017}.

Since biosurveillance is expected to be only sparsely sampling the wild reservoirs, we investigated how limited and incomplete sampling of the animal strain distribution impacts our estimated emergence scores (\EXTENDED Table~\ref{irattabsample}), to quantify uncertainty in the estimates. Instead of using all available human strains, we chose a random sample of 75\% of the available human strains of the appropriate sub-types to construct the \enet models, and used the same sample to compute the emergence scores. As shown in the \EXTENDED Table~\ref{irattabsample}, the SEM of the score estimates are at least an  order of magnitude smaller than the estimates themselves, suggesting that the scores stabilize for a reasonable sampling of the wild distributions.

A second demonstration of our approach is obtained by computing the \erisk scores for ``variants'' (animal \infl strains isolated in human hosts), which might be  expected to  pose high risk due to their successful replication in human cells, even if the possibility of HH transmission is not yet observed or guaranteed.  We analyze 15 such variants isolated post-2015 listed as antigenic prototypes for canidate vaccine viruses by the CDC~\cite{who_influenza_2023} (\SUPPLEMENTARY Table~S-\ref{tabvariant}). Our \erisk estimates puts approximately 30\% of these variants in the high risk category (the $5\%$ tail of the distribution of emergence scores amongst all animal strains collected post 2020, \SUPPLEMENTARY Fig.~S-\ref{animalfrequency}-a). More specifically (\SUPPLEMENTARY Table~S-\ref{tabvariantpc}), we note that $26.7\%$ of the variants have emergence score $>7.0$ compared to $6.4\%$ of the strains in the animal reservoirs, suggesting, as expected, the variants pose much higher risk. Indeed $4$ out of the 15 variants have scores greater than $7.0$. Additionally, we see in \SUPPLEMENTARY Fig.~S-\ref{figvariantrisk}, that the risk posed by the high-risk variants can be seen to be generally increasing over time, up to the point of collection in human hosts.

Finally, we estimate the IRAT scores of all $6,354$ wild \infl animal viruses collected globally  between January 2020 and January 2024. We trained HA and NA \enet models for each subtype using recent sequence data, computed \erisk using the same method as described above, and identified the ones posing maximal risk (Fig.~\ref{figirat}c). $413$ strains turn out to have a predicted emergence score $>7$ (top 15$^{th}$ percentile, and the $95\%$ significance threshold for a normal distribution fit, \SUPPLEMENTARY Fig.~S-\ref{animalfrequency}-a)), and $190$ strains have score $>7.3$ (top $10$$^{th}$ percentile). However, many of these strains are highly similar, differing by only a few edits. To identify the sufficiently distinct risky strains, we constructed the standard phylogeny from  HA sequences with score $>6.25$ (Fig.~\ref{figphylo}), and collapsed all leaves within $15$ edits, showing only the most risky strain within a collapsed group. This leaves $80$ strains (Fig.~\ref{figphylo}), with $12$ having emergence risk $>7$, and $10$ with  risk above $7.3$ (top 10 rows in \EXTENDED Table~\ref{highrisktab}, see also the distribution of the scores of this set in \SUPPLEMENTARY Fig.~S-\ref{animalfrequency}-b).
Subtypes of the risky strains are H7N9(1), H1N2(2), H3N2(2), H1N1(2) and H5N1(3) with  the top 5 most risky strains as follows: A/Hong Kong/125/2017(H7N9), A/California/62/2018(H1N2), A/swine/Missouri/A02524408/2023(H3N2), A/swine/Tver\_region/RII-81-1S/2023(H1N1), A/swine/ North\_Carolina/ A02479173/ 2020(H1N1),  A/Common\_Buzzard /England/125155/2023(H5N1).

To visually check if high risk strains are edit-wise closer to human strains, we compare the HA sequences along with two frequency dominant human strains in 2021-2022 season (\EXTENDED Fig.~\ref{figriskyseq},\ref{figriskyseqlr},\ref{figriskyseq4} and \ref{figriskyseq4lr}). High-risk strains are somewhat closer in number of edits to these high frequency human strains, but both cases (high and low risk) show substantial residue replacements, both in and out of the receptor binding domain (RBD), to make reliable assessments unreliable  without the use of the \erisk score. %

Broadly the risk assessments are not surprising: swines are known to be efficient mixing vessels~\cite{ma2009pig,nelson2018origins,reid2003origin,Baumann}, and hence unsurprisingly host a large fraction of the risky strains ($>63\%$ over $7.0$). A large diversity of avian hosts for the strains ($>30\%$ over $7.0$, $>70\%$ of the reduced set with score $>6.25$) explains our finding of large number of high risk H5N1 strains (avian flu) among the hig-risk set, which correlates with current increasing concerns about H5N1~\cite{CDC_H5N1_Human_Case}, with at least two reported human cases in May 2024. Thus,  qualitatively our results  are well aligned with the current expectations; nevertheless the ability to quantitatively rank  specific strains which pose maximal risk is a crucial new capability enabling proactive pandemic mitigation efforts.

\subsection*{Conclusion}
While numerous tools exist for ad hoc quantification of genomic similarity~\cite{posada1998modeltest,goldberger2005genomic,huelsenbeck1997phylogeny,neher2014predicting,VanderMeer2010,Smith2009}, higher similarity between strains in  these frameworks is not sufficient to imply a high likelihood of a jump. To the best of our knowledge, the \enet algorithm is  the first of its kind to learn an appropriate biologically meaningful comparison metric from data, without assuming any model of DNA or amino acid substitution, or a genealogical tree a priori. While the effect of the environment and selection cannot be inferred from a single sequence, an entire database of observed strains, processed through the right lens, can parse out useful predictive models of these complex interactions. Our results are  aligned with recent studies demonstrating effective  predictability of  future mutations  for different organisms~\cite{mollentze2021identifying,maher2021predicting}.

The \qdist calculation is currently limited to analogous sequences (such as point variations of the same protein from different viral subtypes), and the \enet inference requires a  sufficient diversity of observed strains. A multi-variate regression analysis indicates  that the most important factor for our approach to succeed is  the diversity of the sequence dataset (See ``Multivariate Regression to Understand Data Characteristics Necessary For Emergenet Modeling" in \METHODS and \SUPPLEMENTARY  Table~S-\ref{tabreg}), which would exclude applicability to completely novel pathogens with no related human variants, and ones that evolve very slowly. Nevertheless, the tools reported here can potentially improve effectiveness of the annual flu shot, and, more importantly,   allow for the development of preemptive vaccines to  target risky animal strains  before the first human infection in the next pandemic.
Apart from outlining new precision public health measures to avert pandemics, such strategies might also help to non-controversially counter the impact of vaccine hesitancy which has interfered with optimal pandemic response in recent times.

\subsection*{Data Source}
We use two public sequence databases: 1) National Center for Biotechnology Information (NCBI) virus~\cite{hatcher2017virus} and 2) GISAID~\cite{bogner2006global} databases. The former is a community portal for viral sequence data, aiming to increase the usability of data archived in various NCBI repositories. GISAID has a  more restricted user agreement, and use of GISAID data in an analysis requires acknowledgment of the contributions of both the submitting and the originating laboratories (Corresponding acknowledgment tables are included as supplementary information). We collected a total of $463,266$ sequences in our analysis (see \SUPPLEMENTARY Table~S-\ref{tabseq}).

\subsection*{Data and Software Sharing} 
Working open-source software (requiring Python 3.x) is publicly available at \url{https://pypi.org/project/emergenet/}. All inferred \enet models inferred is  available at \url{https://doi.org/10.5281/zenodo.7387861}.

\bibliographystyle{naturemag}
\bibliography{allbib}

\section*{Acknowledgments}
This work has been partly funded by the PREEMPT program from Defense Advanced Research Projects Agency (D19AC00004/N/A), and intramural grants from the University of Chicago, Biological Science Division.

\clearpage   

\vspace*{0.1cm}
{\LARGE\textbf{Supplementary Materials}}

\section*{\METHODS}
\allowdisplaybreaks{
We briefly describe the proposed computational framework.

\subsection*{\enet Framework}
We do not assume that the mutational  variations at the individual indices of a genomic sequence are independent (See Fig~\ref{figscheme}a). Irrespective of whether mutations are truly random~\cite{hernandez2018algorithmically}, since only certain combinations of individual mutations are viable, individual mutations across a genomic sequence replicating in the wild  appear  constrained, which is what is explicitly  modeled in our approach.

Consider a set of random variables $X=\{X_i\}$, with $i \in \{1, \cdots, N\}$, each taking value from the respective sets $\Sigma_i$. Here each $X_i$ is the random variable modeling the ``outcome'' $i.e.$ the AA residue at the $i^{th}$ index of the protein sequence. A sample $x \in \prod_1^N \Sigma_i$ is an ordered $N$-tuple, which is a specific strain in this context,  consisting of a realization of each of the variables $X_i$ with the $i^{th}$ entry $x_i$ being the realization of random variable $X_i$.

We use the notation $x_{-i}$ and $x^{i,\sigma}$ to denote:
\begin{subequations}\cgather{
x_{-i} \triangleq x_1, \cdots, x_{i-1},x_{i+1},\cdots,x_N\\
x^{i,\sigma} \triangleq x_1, \cdots, x_{i-1},\sigma,x_{i+1},\cdots,x_N, \sigma \in \Sigma_i
}\end{subequations} Also, $\Dx(S)$ denotes the set of probability measures on  a set $S$, $e.g.$,  $\D$ is the set of  distributions on  $\Sigma_i$.

We note that $X$ defines a random field~\cite{vanmarcke2010random} over the index set $\{1, \cdots, N\}$. 

\begin{defn}[\enet]
For a random field $X=\{X_i\}$ indexed by $i \in \{1, \cdots, N\}$, the \enet is defined to be the set of predictors $\Phi=\{\qn\}$, $i.e.$, we have:
\cgather{
\qn : \prod_{j \neq i} \Sigma_j \rightarrow \D,
}  where for a sequence $x$, $\Phi_i(x_{-i}) $ estimates the distribution of $X_i$ on the set $\Sigma_i$.
\end{defn}
We use conditional inference trees as models for predictors~\cite{Hothorn06unbiasedrecursive}, although more general models are possible.

\subsection*{Biology-Aware Distance Between Sequences}
The mathematical form of our metric is not arbitrary; JS divergence is a symmetricised version of the more common KL divergence~\cite{cover} between distributions, and among  different possibilities, the \qdist  is the simplest metric such that the likelihood of a spontaneous jump (See Eq.~\eqref{fundeq} in Methods) is provably bounded above and below  by simple exponential functions of the \qdist.

\begin{defn}[\qdist: adaptive biologically meaningful dissimilarity between sequences]\label{defqdistance}
Given two sequences $x,y \in \prod_1^N\Sigma_i$, such that $x,y$ are drawn from the  populations $P,Q$  inducing the \enet $\Phi^P,\Phi^Q$, respectively,  we define a pseudo-metric $\theta(x,y) $, as follows:
\cgather{\label{q-distance}
\theta(x,y) \triangleq \mathbf{E}_i \left (  \J^{\frac{1}{2}} \left (\qn^P(x_{-i}) , \qn^Q(y_{-i})\right ) \right )
} 
where $ \J(\cdot,\cdot)$ is the Jensen-Shannon divergence~\cite{manning1999foundations} and $\mathbf{E}_i$ indicates expectation over the indices.
\end{defn}
The square-root in the definition arises naturally from the bounds we are able to prove, and is dictated by the form of Pinsker's inequality~\cite{cover}, ensuring that   the sum of the length of successive path fragments equates the length of the path.%

\subsection*{Persistence Probability}
We can formulate a computable characteristic for a sequence $x$ that is related to the more popular notion of replicative fitness. Referred to as \textit{persistence probability} of the sequence $x$, we note that the probability $Pr(x \rightarrow x)$ is the probability that a sequence does not ``jump away'', and higher these odds, lower are the odds that there are  incentives to change, $i.e.$, higher the replicative fitness.
\begin{defn}[Persistence probability of a sequence]\label{defmem}
Given a population $P$ inducing the \enet $\Phi^P$ and a sequence $x$, we can estimate the   probability $Pr(x \rightarrow x)$, denoted as the persistence probability of $x$ as:
\cgather{
\mem{x}^P \triangleq Pr(x \rightarrow x) = \prod_{j=1}^N \left ( \Phi^P_j(x_{-j}) \vert_{x_j} \right )
}
\end{defn}
$x_j$ is the $j^{th}$ entry in $x$, and is thus an element in the set $\Sigma_j$. Since we are mostly concerned with the case where $\Sigma_j$ is a finite set, $\Phi^P_j(x_{-j}) \vert_{x_j}$ is the entry in the probability mass function corresponding to the element of $\Sigma_j$ which appears at the  $j^{th}$ index in sequence $x$. 
\subsection*{Relating Persistence Probability to Replicative Fitness}
The concept of replicative fitness, denoted as \( f(x) \), is fundamentally connected to the persistence probability of a strain, \(\mem{x}\), which as per the definition above, is the likelihood that a strain \( x \) reproduces into itself rather than mutating into a different strain. Mathematically, this relationship can be derived by considering the population dynamics of strain \( x \). The change in the population \( N_x(t) \) of strain \( x \) over time can be modeled as:
\[
\frac{dN_x(t)}{dt} = f(x) N_x(t) - \mu_x N_x(t),
\]
where \( f(x) \) is the replicative fitness of strain \( x \) and \( \mu_x \) is the mutation rate to other strains. This equation indicates that the net growth rate of the population is determined by the difference between the replication rate and the mutation rate:
\[
\frac{dN_x(t)}{dt} = (f(x) - \mu_x) N_x(t).
\]
To express the persistence probability \(\mem{x}\) of strain \( x \), we consider the probability that the strain reproduces into itself rather than mutating. This can be formulated as the ratio of the replication rate to the total rate of change, which includes both replication and mutation:
\[
\mem{x} = \frac{f(x)}{f(x) + \mu_x}.
\]
thus that the persistence probability is directly proportional to the replicative fitness and inversely proportional to the mutation rate. Higher replicative fitness \( f(x) \) leads to a higher persistence probability \(\mem{x}\), while higher mutation rates \( \mu_x \) reduce the persistence probability. Therefore, strains with higher fitness are more likely to maintain their genetic identity over time, ensuring their continued presence in the population. This relationship underscores the critical role of replicative fitness in the evolutionary success and stability of strains within a dynamic environment.

\subsection*{Theoretical Probability Bounds}

The \enet framework  allows us to rigorously compute bounds on the probability of a spontaneous change of one strain to another, brought about by chance mutations. While any sequence of mutations is equally likely, the ``fitness'' of the resultant strain, or the probability that it will even result in a viable strain, or not. Thus the necessity of preserving  function  dictates that not all random changes  are viable, and the probability of observing some trajectories through the sequence space  are far greater  than others. The \enet framework allows us to explore this constrained dynamics, as revealed by a sufficiently large set of genomic sequences.

The mathematical intuition  relating  \qdist  to the log-likelihood of spontaneous change  is similar to quantifying the  odds of  a rare biased outcome when we  toss a fair coin.
While for an unbiased coin, the odds of roughly 50\% heads is overwhelmingly likely, large deviations do happen rarely, and it turns out that the probability of such rare deviations can be explicitly quantified with existing statistical theory~\cite{varadhan2010large}.
 Generalizing to non-uniform conditional distributions inferred by the \enet, the likelihood of a spontaneous transition  by random chance may also be similarly bounded.

We show in Theorem~\ref{thmbnd} in the supplementary text that at a significance level $\alpha$, with a sequence length $N$, the probability of spontaneous jump of sequence $x$ from population $P$ to sequence $y$ in population $Q$, $Pr(x \rightarrow y)$, is bounded by:
\cgather{\label{fundeq}
\mem{y}^Q e^{ \frac{\sqrt{8}N^2}{1-\alpha}\theta(x,y)} \geqq Pr(x \rightarrow y) \geqq \mem{y}^Q e^{-\frac{\sqrt{8}N^2}{1-\alpha}\theta(x,y)}}
where $\mem{y}^Q$ is the persistence probability of strain $y$ in the target population, $N$ is the sequence length, and $\alpha$ is the statistical signifacnce level.

\subsection*{Predicting Seasonal Vaccine Strains} 
 
Analyzing the distribution of sequences observed to circulate in the human population at the present time allows us to forecast effective vaccine strain(s) in the next flu season as follows:

Let $\dst$ be a  strain that will be in circulation  in the upcoming flu season at time $t+\delta$,
where $H^t$ is the set of observed strains presently in circulation in the human population (at time $t$). We will assume that the \enet is constructed using the sequences in the set $H^t$, and remains unchanged upto $t+\delta$. Since this set is a function of time, the inferred \enet also changes with time, and the induced \qdist is denoted as $\theta^{[t]}(\cdot,\cdot)$.

Let us estimate  $\dst$  as the strain that has the maximum average probability of arising from the current strain distribution. From the RHS bound established in Theorem~\ref{thmbnd} (See Eq.~\eqref{fundeq} above) in the supplementary text, we have for a strain $x \in H^t$:
\calign{
  &\ln  \frac{Pr(x \rightarrow \dsta)}{\mem{\dsta}} \geqq  -\frac{\sqrt{8}N^2}{1-\alpha}\theta^{[t]}(x,\dsta)\\
\Rightarrow &\sum_{x \in H^t} \ln  \frac{Pr(x \rightarrow \dsta)}{\mem{\dsta}}  
\geqq  \sum_{x \in H^t}-\frac{\sqrt{8}N^2}{1-\alpha}\theta^{[t]}(x,\dsta)\\
\Rightarrow  &\sum_{x\in H^t}  \theta^{[t]}(x,\dsta) - \abs{H^t}A \ln \mem{\dsta} \geqq  A \ln \frac{1}{\prod_{x \in H^t} Pr(x \rightarrow \dsta)} \intertext{where $A =\frac{1-\alpha}{\sqrt{8}N^2} $, where $N$ is the sequence length considered, and $\alpha$ is a fixed significance level. Since \textbf{minimizing the LHS maximizes the lower bound on the probability of the observed strains simultaneously giving rise to $\dsta$},  $\dst$ may be estimated as a solution to the optimization problem:}
\label{dompred1}&\dst = \argmin_{y \in \cup_{\tau \leqq t} H^\tau} \left ( \sum_{x\in H^t}  \theta^{[t]}(x,y) - \abs{H^t}A \ln \mem{y} \right ) \tag{Eq. \eqref{dompred0}}
}%
We can explicitly interpret the two terms in the Eqn.~\ref{dompred0}. The first term clearly minimies the distance from the centroid, measured in the \qdist metric. This alone will suffice in a static population, and has been observed before using the edit distance metric~\cite{huddleston2020integrating}. The second term penalizes the solution if it has lower replicative fitness. This justifies the interpretation that Eqn.~\ref{dompred0} prescribes the optimal solution as one that is close to the centroid of the circulating strain in recent past, while penalizing strains which do not have high replicative fitness. While this explanation probably seems obvious after the fact, the above argument derives it precisely, and highlights the role of the \qnet.

For validation of this scheme, we define the distance to the future circulation at  time $t+\delta$ for sequence $x$ as the average Hamming  distance, $h$, between $x$ and each sequence $y$ in the future population $H^{t+\delta}$, weighted by its frequency $f_{t+\delta}(y)$, following recent evaluation approaches in the literature~\cite{huddleston2020integrating}. We predict one year into the future, $\delta=1$, and thus:
\calign{\label{dist1}
&d_{\delta}(x) = \sum_{y \in H^{t+\delta}} f_{t+\delta}(y)h(x, y) \tag{Eq. ~\eqref{dist}} 
}%
In Eq.~\eqref{dompred0}, we are making the assumption that $H^t$ is a single cluster of strains. In practice, this might not be the case; we often see several distinct clusters arise in each season, which are often aligned with the distinct clades. We can cluster the strains using a clustering method (here we use the MeanShift algorithm~\cite{carreira2015review}) on the \qdist matrix computed between strains in $H^t$ such that we have $n$ disjoint clusters $H_1^t, H_2^t,\dots H_n^t$ with $\bigcup_{i=1}^nH_i^t = H^t$. We then compute clutser-specific predictions $x_{i*}^{t+\delta}$ using Eq.~\eqref{dompred0}, and obtain the unique vaccine recommendation as the weighted centroid of these cluster-specific recommendations, where the weights are the clutser sizes, $i.e.$:
\cgather{
\dst = \argmin_{x \in \{x_{i*}^{t+\delta}\}} \sum_i h(x_{i*}^{t+\delta},x)|H_i^t|
  }%
For two-cluster predictions where we aim to make two strain recommendations instead of one, we  take the predictions $x_{1*}^{t+\delta}$ and $x_{2*}^{t+\delta}$ from the two largest clusters, $H_1^t$ and $H_2^t$.%
Thus, we can make multiple predictions per season, by replacing $H^t$ in Eq.~\eqref{dompred0} with $H_i^t$ for $i = 1, 2, \cdot n$. When we provide recommendations for two vaccine strains,  we take the predictions from the largest two clusters, as is shown in Fig.~\ref{figseasonal}g-j in the main text.

\subsection*{Measure of Emergence Potential}
\def\ast{x_a^t}
\def\hst{x_h^{t+\delta}}

Recall, Eq.~\eqref{eqrho} states our proposed measure for emergence potential for a strain $x$, which is negative logarithm of the geometric mean of minimum \qdist to recent  human HA and NA sequences, where the \qdist correspond to \enet{s} constructed in a subtype specific manner from human sequences. The key rationale can be established by considering only one protein (say HA), and only one subtype of human models, as follows.

Theorem~\ref{thmbnd} can be rewritten as such: for a strain $x$ in animals, and a human circulating strain $x_h$, we have:
\cgather{
\left \lvert \ln  \frac{Pr(x \rightarrow x_h)}{Pr(x_h \rightarrow x_h)} \right \rvert \leqq \beta \theta(x,x_h),  \textrm{ with } \beta = \frac{\sqrt{8}N^2}{1-\alpha}}%
Now assume that the human strain $x_h$ is well-adapted to humans, implying that it has high fitness. A high fitness of any strain $y$  may be interpreted as a high value for $Pr(y \rightarrow y)$, and in the extreme case $Pr(y \rightarrow y) \approx 1$. Now if the RHS of the above relationship is close to zero, we conclude:
\cgather{
\theta(x,x_h) \rightarrow 0 \Rightarrow Pr(x \rightarrow x_h) \rightarrow Pr(x_h \rightarrow x_h)
}%
and
\cgather{
\theta(x,x_h) = 0 \Rightarrow Pr(x \rightarrow x_h) =  Pr(x_h \rightarrow x_h)
}%
This implies that for an animal strain $x$,  a small \qdist to some well-adpated circulating human strain $x_h$ implies a  high jump probability to a  similarly well-adapted human strain. Given a sufficiently large sample of well-adapted human strains, we can now quantify a measure of emergence potential as:
\cgather{
\rho_t(x) \triangleq -\log \min_{y \in H^t} \theta^{[t]}(x,y)
}%
where as before $H^t$ is the set of observed strains presently in circulation in the human population (at time $t$), restricted to observations over the past year. We do not use historical human strains from further back in time to prevent introducing errors from changing host immune characteristics over time. %

A large value of the emergence potential defined above thus can be interpreted as a large  probability of jumping to a currently well-adapted human strain wrt to both HA and NA sequences, due to the geometric mean of \qdist wrt to HA and NA in Eq.~\eqref{eqrho}.

\subsection*{Proof of Probability Bounds}\label{sec:proof}

\begin{thm}[Probability bound]\label{thmbnd}
Given a sequence  $x$ of length $N$ that transitions  to a strain $y\in Q$, we have the following bounds at significance level $\alpha$.
\cgather{
\mem{y}^Q e^{ \frac{\sqrt{8}N^2}{1-\alpha}\theta(x,y)} \geqq Pr(x \rightarrow y) \geqq \mem{y}^Q e^{-\frac{\sqrt{8}N^2}{1-\alpha}\theta(x,y)}
  }%
  where $\mem{y}^Q$ is the persistence probability of strain $y$ in the target population $Q$ (See Def.~\ref{defmem}), and $\theta(x,y)$ is the q-distance between $x,y$ (See Def.~\ref{defqdistance}).
\end{thm}
\begin{proof}
Using Sanov's theorem~\cite{cover} on large deviations, we conclude that the probability of spontaneous jump from strain $x\in P$ to strain $y\in Q$, with the possibility $P \neq Q$, is given by:
\cgather{\label{eq29}
  Pr(x\rightarrow y) =\prod_{i=1}^N \left ( \Phi^P_i(x_{-i}) \vert_{y_i} \right )
}
Writing the factors on the right hand side as:
\cgather{
 \Phi^P_i(x_{-i}) \vert_{y_i} =  \Phi^Q_i(y_{-i}) \vert_{y_i} \left (  \frac{\Phi^P_i(x_{-i}) \vert_{y_i}}{\Phi^Q_i(y_{-i}) \vert_{y_i}}  \right )
}%
we note that $\Phi^P_i(x_{-i})$, $\Phi^Q_i(y_{-i})$ are distributions on the same index $i$, and hence:
  \cgather{
\vert  \Phi^P_i(x_{-i})_{y_i} - \Phi^Q_i(y_{-i})_{y_i}\vert \leqq \sum_{y_i \in \Sigma_i} \vert  \Phi^P_i(x_{-i})_{y_i} - \Phi^Q_i(y_{-i})_{y_i}\vert 
}%
Using a standard refinement of Pinsker's inequality~\cite{fedotov2003refinements}, and the relationship of Jensen-Shannon divergence with  total variation, we get:
\cgather{
  \theta_i \geqq \frac{1}{8} \vert  \Phi^P_i(x_{-i})_{y_i} - \Phi^Q_i(y_{-i})_{y_i}\vert^2
\Rightarrow \left   \lvert  1  - \frac{\Phi^Q_i(y_{-i})_{y_i}}{\Phi^P_i(x_{-i})_{y_i}} \right \rvert \leqq \frac{1}{a_0}\sqrt{8 \theta_i}
}%
where $a_0$ is the smallest non-zero probability value of generating the entry at any index. We will see that this parameter is related to statistical significance of our bounds. First, we can formulate a lower bound as follows:
\cgather{\label{eqLB}
 \log \left  ( \prod_{i=1}^N   \frac{\Phi^P_i(x_{-i}) \vert_{y_i}}{\Phi^Q_i(y_{-i}) \vert_{y_i}}  \right )
  = \sum_i \log  \left  (  \frac{\Phi^P_i(x_{-i}) \vert_{y_i}}{\Phi^Q_i(y_{-i}) \vert_{y_i}}  \right )
\geqq \sum_i \left  ( 1- \frac{\Phi^Q_i(y_{-i})_{y_i}}{\Phi^P_i(x_{-i})_{y_i}} \right ) \geqq  \frac{\sqrt{8}}{a_0}\sum_i\theta_i^{1/2} = -\frac{\sqrt{8}N}{a_0}\theta
}%
Similarly,  the upper bound may be derived as:
\cgather{\label{eqUB}
\log \left  ( \prod_{i=1}^N   \frac{\Phi^P_i(x_{-i}) \vert_{y_i}}{\Phi^Q_i(y_{-i}) \vert_{y_i}}  \right )
  = \sum_i \log  \left  (  \frac{\Phi^P_i(x_{-i}) \vert_{y_i}}{\Phi^Q_i(y_{-i}) \vert_{y_i}}  \right ) \leqq \sum_i \left  ( \frac{\Phi^Q_i(y_{-i})_{y_i}}{\Phi^P_i(x_{-i})_{y_i}} - 1 \right ) \leqq \frac{\sqrt{8}N}{a_0}\theta
}%
Combining Eqs.~\ref{eqLB} and \ref{eqUB}, we conclude:
\cgather{
\mem{y}^Q e^{ \frac{\sqrt{8}N}{a_0}\theta} \geqq Pr(x \rightarrow y) \geqq \mem{y}^Q e^{-\frac{\sqrt{8}N}{a_0}\theta}
}%
Now, interpreting $a_0$ as the probability of generating an unlikely event below our desired threshold ($i.e.$ a ``failure''), we note that the probability of generating at least one such event is given by $1-(1-a_0)^N$. Hence if $\alpha$ is the pre-specified significance level, we have for $N >> 1 $:
\cgather{
 a_0 \approx (1 -\alpha)/N
}%
Hence, we conclude, that at significance level $\geqq \alpha$, we have the bounds:
\cgather{
\mem{y}^Q e^{ \frac{\sqrt{8}N^2}{1-\alpha}\theta} \geqq Pr(x \rightarrow y) \geqq \mem{y}^Q e^{-\frac{\sqrt{8}N^2}{1-\alpha}\theta}
  }%
\end{proof}
\begin{rem}
This bound can be rewritten in terms of the log-likelihood of the spontaneous jump and  constants independent of the  initial sequence $x$ as:
\cgather{
\left \lvert \log Pr(x \rightarrow y) -C_0 \right \vert \leqq C_1 \theta
}%
where the constants are given by:
\calign{
C_0 &= \log \mem{y}^Q \\
C_1 &= \frac{\sqrt{8} N^2}{1-\alpha}
}%
\end{rem}

\subsection*{In-silico Corroboration of \enet{'s} Capability To Capture Biologically Meaningful Structure}
We compare the results of simulated mutational perturbations to sequences from our databases (for which we have already constructed \enet{s}), and then use NCBI BLAST (\href{https://blast.ncbi.nlm.nih.gov/Blast.cgi}{https://blast.ncbi.nlm.nih.gov/Blast.cgi}) to identify  if  our perturbed sequences match with existing sequences in the databases (\SUPPLEMENTARY Fig.~S-\ref{figsoa}). We find that in contrast to random variations, which rapidly diverge the trajectories, the \enet constraints tend to produce smaller variance in the trajectories, maintain a high degree of match as we extend our trajectories, and produces matches closer in time to the collection time of the  initial sequence, suggesting that the \enet  does indeed capture realistic constraints.

\subsection*{Multivariate Regression to Understand Data Characteristics Necessary For \enet Modeling}

We investigate the key factors that contribute to modeling a set of strains well within the \enet framework. We carry out a multivariate regression with data diversity, the complexity of inferred \enet and the edit distance of the WHO recommendation from a frequency-dominant strain as independent variables (See \SUPPLEMENTARY Table~S-\ref{tabreg} for definitions). Here we define data diversity as the number of clusters we have in the input set of sequences, such that any two sequences five or less mutations apart are in the same cluster. \enet complexity is measured by the number of decision nodes in the component decision trees of the recursive forest.

We select several plausible structures of the regression equation, and in each case conclude that  data diversity has the most important and statistically significant contribution (\SUPPLEMENTARY Table~S-\ref{tabreg}).

}

\clearpage

\setcounter{figure}{0}
\renewcommand{\figurename}{Extended Data Figure}                               
\setcounter{table}{0}                                     
\renewcommand{\tablename}{Extended Data Table}                                 
\begin{table}[!ht]
\captionN{H1N1 northern hemisphere}\label{tabrec0}
\sffamily\fontsize{6}{8}\selectfont
\begin{tabular}{L{0.35in}|L{1.1in}|L{1.1in}|L{1.1in}|L{1.1in}|C{0.27in}|C{0.27in}|C{0.27in}}\hline
\rowcolor{lightgray!50}Season & WHO  Recommendation & Enet  Recommendation,  Cluster 1 & Enet  Recommendation, Cluster 2 & Enet  Recommendation,  Single  Cluster & WHO  Error & Enet  Error & Enet  Error  Single \\\hline
2003-04& A/New  Caledonia/20/99 & A/Memphis/1/2001 & A/HaNoi/2704/2002 & A/Memphis/1/2001 &8.3&5.3&11.3\\\hline
2004-05& A/New  Caledonia/20/99 & A/Memphis/1/2001 & A/Memphis/1/2001 & A/Memphis/1/2001 &7.4&9.7&9.7\\\hline
2005-06& A/New  Caledonia/20/99 & A/Malaysia/25862/2003 & A/Malaysia/25862/2003 & A/Malaysia/25862/2003 &10.6&6.4&6.4\\\hline
2006-07& A/New  Caledonia/20/99 &A/Yazd/144/2006& A/Malaysia/30025/2004 & A/New  York/230/2003 &10.4&5.7&6.4\\\hline
2007-08& A/Solomon  Islands/3/2006 & A/New  York/1050/2006 & A/Incheon/2647/2007 & A/New  York/1050/2006 &10.5&11.2&13.6\\\hline
2008-09& A/Brisbane/59/2007 & A/England/545/2007 & A/Hong\_Kong/2613/2007 & A/England/545/2007 &28.9&27.8&28.5\\\hline
2009-10& A/Brisbane/59/2007 & A/Hawaii/02/2008 & A/Hong\_Kong/H090-751-V3 & A/Hawaii/02/2008 &426.2&2.4*&426.2\\\hline
2010-11& A/California/7/2009 &A/OKINAWA/283/2009& A/Qingdao/FF86/2009 &A/OKINAWA/283/2009&10.6&7.5&7.8\\\hline
2011-12& A/California/7/2009 & A/Florida/14/2010 & A/Taiwan/66179/2010 & A/Florida/14/2010 &12.8&7.6&7.6\\\hline
2012-13& A/California/7/2009 & A/England/WTSI1769/2010 & A/Mexico/3723/2011 & A/Singapore/GP2892/2010 &12.6&8.3&5.4\\\hline
2013-14& A/California/7/2009 & A/IIV-Vladimir/67/2011 & A/Srinigar/827/2011 & A/Helsinki/207/2013 &14.5&5.8&7.7\\\hline
2014-15& A/California/7/2009 & A/Nicaragua/6184\_13/2013 & A/Nicaragua/6184\_13/2013 & A/Nicaragua/6184\_13/2013 &20.0&11.1&11.1\\\hline
2015-16& A/California/7/2009 &A/SENDAI/42/2014& A/USA/VFFSP\_UNITHER \_00011/2014 &A/SENDAI/42/2014&16.8&6.9&6.9\\\hline
2016-17& A/California/7/2009 & A/Massachusetts/24/2015 & A/India/1399/2015 & A/Massachusetts/24/2015 &17.9&3.7&3.7\\\hline
2017-18& A/Michigan/45/2015 & A/Bolzano/11/2016 & A/Costa\_Rica/4631/2016 & A/Bolzano/11/2016 &7.4&6.4&6.4\\\hline
2018-19& A/Michigan/45/2015 & A/Cambodia/B0629529/2017 & A/Mississippi/26/2017 &A/Qatar/16-VI-17-0046212/2017&10.3&6.2&5.3\\\hline
2019-20& A/Brisbane/02/2018 & A/Michigan/398/2018 & A/Zambia/301/2017 & A/Bhutan/15/2019 &11.3&5.7&5.7\\\hline
2020-21& A/Guangdong-Maonan/SWL1536/2019 & A/Singapore/GP1719/2019 & A/Baltimore/R0675/2019 & A/Linkou/R0160/2018 &17.1&10.7&11.5\\\hline
2021-22& A/Victoria/2570/2019 &A/Togo/64/2021& A/Tianjin/00030/2020 &A/Togo/64/2021&10.8&8.4&8.4\\\hline
2022-23& A/Victoria/2570/2019 & A/Cameroon/7082/2021 & A/England/221340346/2022 & A/Lyon/CHU\_0186413/2020 &11.8&6.3&15.6\\\hline
2023-24& A/Victoria/4897/2022 & A/Bahrain/2446/2021 & A/South\_Africa/R05701/2022 & A/South\_Africa/14627/2021 &-1&-1&-1\\\hline
\end{tabular}

\fontsize{7}{8}\selectfont
\vskip 0.5em
$^\star$ A/Hong Kong/H090-751-V3, which is close to the 2009-10 H1N1 pandemic strain (5 edits), was collected on 2009-02-08, before the WHO recommendation was released, but submitted on 2011-07-08.
\end{table}
\begin{table}[!ht]
\captionN{H1N1 southern hemisphere}\label{tabrec1}
\sffamily\fontsize{6}{8}\selectfont
\begin{tabular}{L{0.35in}|L{1.1in}|L{1.1in}|L{1.1in}|L{1.1in}|C{0.27in}|C{0.27in}|C{0.27in}}\hline
\rowcolor{lightgray!50}Season & WHO  Recommendation & Enet  Recommendation,  Cluster 1 & Enet  Recommendation, Cluster 2 & Enet  Recommendation,  Single  Cluster & WHO  Error & Enet  Error & Enet  Error  Single \\\hline
2003& A/New  Caledonia/20/99 & A/HaNoi/2143/2001 & A/HaNoi/2143/2001 & A/HaNoi/2143/2001 &7.3&3.0&3.0\\\hline
2004& A/New  Caledonia/20/99 & A/HaNoi/2546/2002 & A/HaNoi/2704/2002 & A/HaNoi/2546/2002 &11.2&8.8&10.8\\\hline
2005& A/New  Caledonia/20/99 & A/Michigan/1/2003 & A/Michigan/1/2003 & A/Michigan/1/2003 &8.7&5.8&5.8\\\hline
2006& A/New  Caledonia/20/99 & A/Malaysia/25531/2003 & A/Auckland/619/2005 & A/Malaysia/25531/2003 &12.9&9.7&10.6\\\hline
2007& A/New  Caledonia/20/99 & A/Malaysia/1652509/2006 & A/Malaysia/25531/2003 & A/Malaysia/34291/2006 &21.6&15.1&14.9\\\hline
2008& A/Solomon  Islands/3/2006 & A/Malaysia/1715991/2007 & A/Canada/0379/2007 & A/New\_York/08/2007 &11.1&4.7&12.7\\\hline
2009& A/Brisbane/59/2007 & A/California/28/2007 & A/Vietnam/100/2008 & A/California/28/2007 &402.0&401.2&402.1\\\hline
2010& A/California/7/2009 & A/Hong\_Kong/H090-751-V3 & A/Hawaii/02/2008 & A/Hong\_Kong/H090-751-V3 &11.7&4.3&8.0\\\hline
2011& A/California/7/2009 & A/Guangdong/1250/2009 & A/Kenya/1329/2010 & A/Guangdong/1250/2009 &10.5&6.7&6.7\\\hline
2012& A/California/7/2009 & A/Singapore/GP81/2011 & A/Chongqing-Yuzhong/SWL11307/2009 & A/Scotland/ Inverness\_10/2009 &12.6&6.9&6.9\\\hline
2013& A/California/7/2009 & A/Virginia/07/2012 & A/Kenya/CDC-KNH-028/2010 & A/Singapore/GP4215/2010 &13.3&6.5&9.2\\\hline
2014& A/California/7/2009 & A/England/358/2013 & A/India/P1114854/2011 & A/Cruz  Alta/ LACENRS-129/2012 &13.8&4.3&2.5\\\hline
2015& A/California/7/2009 & A/North\_Carolina/07/2013 & A/shandong-zhifu/SWL1105/2014 & A/North\_Carolina/07/2013 &14.9&4.1&4.1\\\hline
2016& A/California/7/2009 & A/Colombia/1457/2015 & A/India/1399/2015 & A/Cruz  Alta/ LACENRS-129/2012 &16.8&4.2&5.0\\\hline
2017& A/Michigan/45/2015 & A/Ivanovo/CRIE-73/2016 & A/Albania/6485/2016 & A/Ivanovo/CRIE-73/2016 &5.5&4.5&4.5\\\hline
2018& A/Michigan/45/2015 &A/Qatar/16-VI-17-0046212/2017& A/Chongqing-Yuzhong/SWL1453/2017 & A/Texas/28/2016 &8.8&3.9&5.9\\\hline
2019& A/Michigan/45/2015 & A/California/145/2017 & A/California/145/2017 & A/California/145/2017 &11.2&6.3&6.3\\\hline
2020& A/Brisbane/02/2018 & A/Shanghai/192/2018 & A/Germany/9054/2019 & A/Nepal/19FL3860/2019 &28.9&23.2&23.7\\\hline
2021& A/Victoria/2570/2019 & A/Yaroslavl/138/2020 & A/Ontario/RV1058/2019 &A/Qatar/16-VI-19-0068514/2019&11.6&12.0&11.1\\\hline
2022& A/Victoria/2570/2019 & A/Netherlands/10253/2020 & A/India/PUN-NIV323483/2021 & A/Mecklenburg-Vorpommern/2/2020 &9.9&3.7&11.3\\\hline
2023& A/Sydney/5/2021 & A/Pakistan/GIHSN/ ICT/826/2023 &A/NIIGATA/922/2019& A/Pakistan/GIHSN/ ICT/826/2023 &-1&-1&-1\\\hline
\end{tabular}

\end{table}
\begin{table}[!ht]
\captionN{H3N2 northern hemisphere}\label{tabrec2}
\sffamily\fontsize{6}{8}\selectfont
\begin{tabular}{L{0.35in}|L{1.1in}|L{1.1in}|L{1.1in}|L{1.1in}|C{0.27in}|C{0.27in}|C{0.27in}}\hline
\rowcolor{lightgray!50}Season & WHO  Recommendation & Enet  Recommendation,  Cluster 1 & Enet  Recommendation, Cluster 2 & Enet  Recommendation,  Single  Cluster & WHO  Error & Enet  Error & Enet  Error  Single \\\hline
2003-04&A/Moscow/10/99& A/Hong  Kong/ CUHK53327/2002 & A/Netherlands/88/2003 & A/Hong  Kong/ CUHK24044/2002 &26.5&3.8&3.9\\\hline
2004-05& A/Fujian/411/2002 & A/Queensland/40/2003 & A/Singapore/ NHRC0001/2003 & A/New  York/214/2003 &10.4&5.7&7.5\\\hline
2005-06& A/California/7/2004 & A/Canterbury/18/2004 & A/Tairawhiti/369/2004 & A/Canterbury/18/2004 &13.0&5.1&5.2\\\hline
2006-07& A/Wisconsin/67/2005 & A/Mexico/TLA2227/2005 & A/New\_York/1034/2006 & A/Mexico/TLA2227/2005 &10.9&5.7&8.5\\\hline
2007-08& A/Wisconsin/67/2005 & A/Mexico/MEX2640/2005 &A/Madagascar/2694/2006& A/Kentucky/UR06-0044/2007 &11.2&4.3&3.3\\\hline
2008-09& A/Brisbane/10/2007 & A/LangSon/LS218/2007 & A/Malaysia/1767091/2007 & A/LangSon/LS218/2007 &5.0&2.3&2.4\\\hline
2009-10& A/Brisbane/10/2007 & A/Pennsylvania/PIT43/2008 & A/Managua/5/2007 & A/Pennsylvania/PIT43/2008 &9.2&7.2&7.2\\\hline
2010-11& A/Perth/16/2009 & A/Stockholm/89/2009 & A/Bordeaux/1942/2009 & A/California/VRDL384/2009 &10.4&8.0&6.9\\\hline
2011-12& A/Perth/16/2009 & A/Quebec/RV2804/2010 & A/Mumbai/2516/2009 & A/Singapore/N1604/2009 &10.6&7.8&7.9\\\hline
2012-13& A/Victoria/361/2011 & A/Singapore/C2011.573/2011 & A/Idaho/01/2010 & A/Cote\_d'Ivoire/ GR1596/2010 &6.9&7.9&9.1\\\hline
2013-14& A/Victoria/361/2011 & A/Boston/DOA2-162/2012 & A/Belem/119244/2012 & A/Boston/DOA2-162/2012 &8.8&4.0&4.0\\\hline
2014-15& A/Texas/50/2012 & A/Schleswig\_Holstein/7/2012 & A/Maryland/03/2014 & A/Schleswig\_Holstein/7/2012 &10.7&2.9&9.8\\\hline
2015-16& A/Switzerland/9715293/2013 & A/New  Hampshire/08/2014 & A/Ireland/14M02879/2014 & A/Thailand/CU-B11417/2014 &10.6&5.7&5.5\\\hline
2016-17& A/Hong  Kong/4801/2014 & A/Alaska/123/2015 & A/Beijing-Xicheng/ 13100/2014 & A/New\_Hampshire/09/2014 &6.2&6.9&7.1\\\hline
2017-18& A/Hong  Kong/4801/2014 & A/South\_Africa/R3703/2016 & A/South\_Africa/3731/2016 & A/Tanzania/2220/2016 &6.2&7.5&7.7\\\hline
2018-19& A/Singapore/INFIMH-16-0019/2016 & A/Hong\_Kong/3554/2017 & A/Taiwan/473/2017 & A/Connecticut/21/2014 &12.7&10.5&12.7\\\hline
2019-20& A/Kansas/14/2017 & A/Maldives/338/2018 & A/Tennessee/64/2018 & A/Arkansas/15/2016 &13.7&5.2&12.5\\\hline
2020-21& A/Hong  Kong/2671/2019 & A/Alaska/04/2019 & A/Guangxi-Mashan/32/2019 & A/Colombia/5215/2015 &15.9&15.7&21.9\\\hline
2021-22& A/Cambodia/e0826360/2020 & A/Bangladesh/1003/2020 & A/Nigeria/4976/2020 & A/Nepal/21FL0201/2021 &16.3&7.4&7.5\\\hline
2022-23& A/Darwin/9/2021 & A/Saarland/1/2022 & A/Nepal/21FL2439/2021 & A/Darwin/7/2021 &8.2&7.1&6.5\\\hline
2023-24& A/Darwin/9/2021 & A/Michigan/ UOM10045676784/2022 &A/Norway/34811/2022& A/Texas/78/2022 &-1&-1&-1\\\hline
\end{tabular}

\end{table}
\begin{table}[!ht]
\captionN{H3N2 southern hemisphere}\label{tabrec3}
\sffamily\fontsize{6}{8}\selectfont
\begin{tabular}{L{0.35in}|L{1.1in}|L{1.1in}|L{1.1in}|L{1.1in}|C{0.27in}|C{0.27in}|C{0.27in}}\hline
\rowcolor{lightgray!50}Season & WHO  Recommendation & Enet  Recommendation,  Cluster 1 & Enet  Recommendation, Cluster 2 & Enet  Recommendation,  Single  Cluster & WHO  Error & Enet  Error & Enet  Error  Single \\\hline
2003&A/Moscow/10/99& A/Canterbury/09/2002 & A/Hong\_Kong/ CUHK24167/2002 & A/Canterbury/57/2002 &26.0&6.6&17.9\\\hline
2004& A/Fujian/411/2002 & A/Wellington/3/2003 & A/Hong\_Kong/ CUHK6377/2003 & A/Queensland/40/2003 &9.3&7.3&7.3\\\hline
2005& A/Wellington/1/2004 & A/Canterbury/11/2004 & A/TayNguyen/TN905/2003 & A/Canterbury/11/2004 &6.7&4.0&4.1\\\hline
2006& A/California/7/2004 & A/Hong  Kong/ CUHK7711/2005 & A/New  Jersey/ NHRC0001/2005 & A/Canterbury/11/2004 &13.6&2.9&4.6\\\hline
2007& A/Wisconsin/67/2005 & A/Mexico/DIF2601/2005 & A/Illinois/NHRC0002/2006 & A/Mexico/DIF2601/2005 &12.0&8.2&8.3\\\hline
2008& A/Brisbane/10/2007 & A/Washington/UR06-0225/2007 & A/Australia/NHRC0013/2005 & A/Washington/UR06-0225/2007 &4.9&5.8&5.8\\\hline
2009& A/Brisbane/10/2007 & A/Kentucky/UR07-0124/2008 & A/Argentina/405/2007 & A/Vietnam/214/2008 &7.9&5.9&5.9\\\hline
2010& A/Perth/16/2009 & A/Hong\_Kong/H090-755-V2 & A/Sydney/3/2009 & A/California/VRDL255/2009 &8.9&8.8&8.0\\\hline
2011& A/Perth/16/2009 & A/Bangladesh/483/2009 & A/Singapore/C2010.310/2010 & A/Singapore/C2009.784/2009 &10.0&5.5&6.6\\\hline
2012& A/Perth/16/2009 & A/Uganda/UVRI/Kisenyi/ 005/2010-09-28 & A/Romania/55656/2011 & A/California/VRDL384/2009 &17.5&12.4&14.3\\\hline
2013& A/Victoria/361/2011 & A/Alborz/1095/2012 & A/West  Virginia/06/2011 & A/Kentucky/21/2009 &8.1&8.9&12.8\\\hline
2014& A/Texas/50/2012 & A/Niakhar/7504/2012 & A/Houston/JMM\_77/2012 & A/New\_York/05/2013 &8.8&6.4&6.7\\\hline
2015& A/Switzerland/9715293/2013 & A/Trinidad/3558/2013 & A/Singapore/G2-6.1/2013 & A/Boston/DOA2-162/2012 &10.7&11.4&9.9\\\hline
2016& A/Hong  Kong/4801/2014 & A/Scotland/146/2015\_(H3N2) &A/Barbados/2879/2015& A/Scotland/146/2015\_(H3N2) &6.9&5.3&7.6\\\hline
2017& A/Hong  Kong/4801/2014 & A/Taiwan/1098/2015 & A/Mexico/2117/2015 &A/VICTORIA/5070/2014&7.6&8.2&8.5\\\hline
2018& A/Singapore/INFIMH-16-0019/2016 &A/FUKUSHIMA/122/2016& A/Alberta/RV0043/2016 & A/Tanzania/2220/2016 &9.6&6.4&8.8\\\hline
2019& A/Switzerland/8060/2017 & A/England/7400/2018 & A/Singapore/SGH0650/2017 & A/Heilongjiang-Xiangyang/11347/2015 &15.9&10.2&14.0\\\hline
2020& A/South  Australia/34/2019 & A/Alicante/ 19\_2105\_20190326 & A/Jiangsu-Rugao/326/2018 & A/Singapore/KK1149/2017 &12.0&14.2&10.9\\\hline
2021& A/Hong  Kong/2671/2019 & A/Cambodia/e0826360/2020 & A/Bangladesh/8001/2020 & A/Indonesia/ NIHRDLPG854/2020 &18.3&9.1&12.9\\\hline
2022& A/Darwin/9/2021 & A/India/PUN-NIV301718/2021 & A/Marseille/0486/2021\_aug & A/Bangladesh/0002/2020 &7.1&5.3&6.7\\\hline
2023& A/Darwin/9/2021 & A/Bulgaria/1666/2022 & A/Michigan/ UOM10045976087/2022 & A/Bulgaria/1666/2022 &-1&-1&-1\\\hline
\end{tabular}

\end{table}

\begin{table}[!ht]\centering
\captionN{Out-performance of \enet recommendations over WHO 
for Influenza A vaccine composition}\label{enetimprovement}\centering
\sffamily\fontsize{7}{8}\selectfont
\begin{tabular}{C{.45in}|C{.55in}|C{.35in}|C{0.65in}|C{0.65in}|C{0.7in}|C{0.65in}|C{0.65in}|C{0.7in}}
\multicolumn{3}{c}{}&\multicolumn{3}{c}{Two decades}&\multicolumn{3}{|c}{One decade}\\\hline
\rowcolor{lightgray!50}Subtype & Hemisphere & Cluster & WHO Error & Enet Error & Improvement \% & WHO Error & Enet Error & Improvement \%\\\hline
H1N1& North &1&13.15* (33.80)&9.42* (30.26)&28.40* (11.73)&13.78&8.24&67.27\\\hline
H1N1& North &2&13.15* (33.80)&8.46* (8.15)&35.71* (314.64)&13.78&7.13&93.29\\\hline
H1N1& South &1&12.76* (32.22)&8.57* (28.25)&32.81* (14.07)&13.47&8.35&61.31\\\hline
H1N1& South &2&12.76* (32.22)&7.24* (26.93)&43.29* (19.64)&13.47&7.26&85.59\\\hline
H3N2& North &1&11.18&7.86&42.13&10.94&9.53&14.72\\\hline
H3N2& North &2&11.18&6.54&70.86&10.94&7.29&49.93\\\hline
H3N2& South &1&11.08&9.07&22.11&10.49&9.88&6.22\\\hline
H3N2& South &2&11.08&7.65&44.94&10.49&8.55&22.73\\\hline
\end{tabular}

\fontsize{7}{8}\selectfont
\flushleft
\vskip 0.5em
$^\star$ Average omits the 2009-10 H1N1 pandemic season. Complete average in parentheses.
\end{table}

\begin{table}[!ht]\centering
\captionN{\enet against Huddleston et al.~\cite{huddleston2020integrating} predictions; reported improvement over WHO distance to future}\label{huddlestoncomparison}\centering
\sffamily\fontsize{7}{8}\selectfont
\begin{tabular}{L{0.55in}|L{1.6in}|L{0.5in}|L{0.5in}|L{0.5in}|L{0.5in}|L{0.5in}|L{0.5in}}\hline
 \rowcolor{lightgray!50}Timepoint & WHO  Recommendation & Emergenet  Distance  to  Future & Emergenet Improvement & LBI + Mutational Load  Distance  to  Future & LBI + Mutational Load Improvement & HI + Mutational Load Distance  to  Future & HI + Mutational Load Improvement \\\hline
2003-10-01& A/Fujian/411/2002 &7.28&2.05&8.44&1.01&6.50&2.95\\\hline
2004-10-01& A/Wellington/1/2004 &4.06&2.63&4.38&1.70&4.38&1.70\\\hline
2005-04-01& A/California/7/2004 &5.18&7.82&4.60&2.89&4.60&2.89\\\hline
2006-04-01& A/Wisconsin/67/2005 &8.54&2.36&5.36&4.94&5.37&4.94\\\hline
2007-10-01& A/Brisbane/10/2007 &5.78&-0.88&3.78&2.00&3.78&2.00\\\hline
2009-10-01& A/Perth/16/2009 &8.03&0.83&6.95&1.00&7.93&0.02\\\hline
2012-04-01& A/Victoria/361/2011 &4.00&4.79&3.72&2.54&7.02&-0.75\\\hline
2013-10-01& A/Texas/50/2012 &6.66&2.10&6.54&2.11&6.89&1.76\\\hline
2014-10-01& A/Switzerland/9715293/2013 &9.85&0.87&3.88&6.94&3.88&6.94\\\hline
2015-10-01& A/Hong  Kong/4801/2014 &7.64&-0.74&7.33&-1.24&6.09&0.00\\\hline
2017-10-01& A/Singapore/INFIMH-16-0019/2016 &8.79&0.78&7.47&1.15&7.81&0.80\\\hline
2018-10-01& A/Switzerland/8060/2017 &13.97&1.94&17.08&-1.91&10.42&4.75\\\hline
&Average&7.48&2.05&6.63&1.93&6.22&2.33\\\hline
\end{tabular}

\flushleft
\fontsize{7}{8}\selectfont
$^\star$ Our strain population is slightly different, so our WHO distance to the future is also slightly different. You can retrieve our WHO distance to the future by adding the columns "``Emergenet Distance to Future" and ``Emergenet Improvement". Likewise, you can retrieve their distance to the future in a similar way.
\end{table}

\begin{table}[!ht]\centering
\captionN{Out-performance of \enet single-cluster recommendations over randomly selected strains}\label{enetimprovementrandom}\centering
\sffamily\fontsize{7}{8}\selectfont
\begin{tabular}{C{.45in}|C{.55in}|C{0.65in}|C{0.65in}|C{0.7in}|C{0.65in}|C{0.65in}|C{0.7in}}
\multicolumn{2}{c}{}&\multicolumn{3}{c}{Two decades}&\multicolumn{3}{|c}{One decade}\\\hline
\rowcolor{lightgray!50}Subtype & Hemisphere & Random Error & Enet Error & Improvement \% & Random Error & Enet Error & Improvement \%\\\hline
H1N1&North&13.40* (34.03)&9.42* (30.26)&29.69* (12.49)&14.77&8.24&79.36\\\hline
H1N1&South&15.88* (35.20)&8.57* (28.25)&46.03* (24.61)&14.50&8.35&73.54\\\hline
H3N2&North&10.80&7.86&37.37&12.45&9.53&30.61\\\hline
H3N2&South&12.30&9.07&35.54&12.87&9.88&30.26\\\hline
\end{tabular}

\fontsize{7}{8}\selectfont
\flushleft
\vskip 0.5em
$^\star$ Average omits the 2009-10 H1N1 pandemic season. Complete average in parentheses.
\end{table}

\clearpage

\ifFIGS
\begin{table}[!ht]\centering
\captionN{Influenza A strains evaluated by IRAT and corresponding \enet predicted risk scores}\label{irattab}
\sffamily\fontsize{6}{8}\selectfont
\begin{tabular}{L{1.5in}|L{0.26in}|L{0.5in}|L{0.6in}|L{0.3in}|L{0.3in}|L{0.3in}|L{0.3in}|L{0.3in}|L{0.3in}|L{0.3in}}\hline
 \rowcolor{lightgray!50}Influenza  Virus & Virus  Type & Date  of  Risk  Assessment & Risk  Score  Category & Emergence  Score & Mean  Low  Acceptable  Emergence & Mean  High  Acceptable  Emergence & Geom  Mean  Risk & Predicted  Emergence & Predicted  Emergence  Low & Predicted  Emergence  High \\\hline
 A/Hong  Kong/125/2017 &H7N9&2017-05-01& Moderate-High &6.5&5.65&7.51&0.00001&7.73&6.19&9.18\\\hline
 A/Shanghai/02/2013 &H7N9&2016-04-01& Moderate-High &6.4&5.52&7.43&0.00001&7.73&6.19&9.18\\\hline
 A/California/62/2018 &H1N2&2019-07-01& Moderate &5.8&4.22&7.16&0.00001&7.73&6.19&9.18\\\hline
 A/Indiana/08/2011 &H3N2&2012-12-01& Moderate &6.0&-1&-1&0.00005&7.66&6.13&9.1\\\hline
 A/Sichuan/06681/2021 &H5N6&2021-10-01& Moderate &5.3&3.88&6.45&0.00128&6.58&5.2&7.92\\\hline
 A/Anhui-Lujiang/39/2018 &H9N2&2019-07-01& Moderate &6.2&4.76&7.57&0.00163&6.42&5.06&7.74\\\hline
 A/Ohio/13/2017 &H3N2&2019-07-01& Moderate &6.6&5.01&7.59&0.00200&6.27&4.93&7.58\\\hline
 A/mink/Spain/3691-8\_22VIR10586-10/2022 &H5N1&2023-04-01& Moderate &5.1&3.96&6.27&0.00264&6.07&4.75&7.35\\\hline
 A/swine/Shandong/1207/2016 &H1N1&2020-07-01& Moderate &7.5&6.33&8.65&0.00315&5.93&4.64&7.2\\\hline
 A/Vietnam/1203/2004 &H5N1&2011-11-01& Moderate &5.2&-1&-1&0.00505&5.54&4.31&6.78\\\hline
 A/American  wigeon/South  Carolina/AH0195145/2021 &H5N1&2022-03-01& Moderate &4.4&3.28&5.51&0.00692&5.28&4.08&6.49\\\hline
 A/Northern  pintail/Washington/40964/2014 &H5N2&2015-03-01& Low-Moderate &3.8&2.6&5.0&0.01381&4.68&3.56&5.83\\\hline
 A/American  green-winged  teal/Washington/1957050/2014 &H5N1&2015-03-01& Low-Moderate &3.6&2.4&4.6&0.01484&4.61&3.5&5.76\\\hline
 A/canine/Illinois/12191/2015 &H3N2&2016-06-01&Low&3.7&2.81&4.9&0.01750&4.46&3.38&5.6\\\hline
 A/Bangladesh/0994/2011 &H9N2&2014-02-01& Moderate &5.6&4.49&6.74&0.03224&3.92&2.91&4.99\\\hline
 A/Yunnan/14564/2015 &H5N6&2016-04-01& Moderate &5.0&4.07&6.18&0.04202&3.68&2.7&4.73\\\hline
 A/gyrfalcon/Washington/41088/2014 &H5N8&2015-03-01& Low-Moderate &4.2&2.9&5.3&0.05259&3.47&2.53&4.51\\\hline
 A/chicken/Tennessee/17-007431-3/2017 &H7N9&2017-10-01&Low&3.1&2.2&3.94&0.05422&3.45&2.5&4.48\\\hline
 A/Jiangxi-Donghu/346/2013 &H10N8&2014-02-01& Moderate &4.3&3.37&5.96&0.06794&3.24&2.32&4.25\\\hline
 A/turkey/Indiana/1573-2/2016 &H7N8&2017-07-01&Low&3.4&2.4&4.26&0.07830&3.11&2.21&4.11\\\hline
 A/chicken/Tennessee/17-007147-2/2017 &H7N9&2017-10-01&Low&2.8&2.01&3.71&0.07995&3.09&2.2&4.09\\\hline
 A/Astrakhan/3212/2020 &H5N8&2021-03-01& Moderate &4.6&3.64&5.82&0.13983&2.58&1.76&3.53\\\hline
 A/Netherlands/219/2003 &H7N7&2012-06-01& Moderate &4.6&3.22&4.39&0.14384&2.56&1.74&3.50\\\hline
\end{tabular}

\flushleft
\fontsize{7}{8}\selectfont
$^\star$ First seven columns available from IRAT\cite{Influenz24:online}. Last four columns predicted by \enet.\\
$^{\star\star}$ The strain A/duck/New York/1996, analyzed by IRAT to have emergence score 2.3 on 2011-11-01, is omitted because its HA segment is unavailable.
\end{table}

\vskip 3em

\ifFIGS
\begin{table}[!ht]\centering
\captionN{\textcolor{black}{Count of identified animal strains above estimated emergence risk threshold}}\label{riskytab}
\sffamily\fontsize{8}{8}\selectfont
\begin{tabular}{L{.6in}|L{.9in}|L{.9in}|L{.9in}|L{.9in}}\hline
\rowcolor{lightgray!50}Subtype& Score  $>$ 6.5 & Score  $>$ 7.0 & Score  $>$ 7.5 & Score  $>$ 7.7 \\\hline
H1N1& 7  (8.75\%) & 3  (25.0\%) & 2  (33.33\%) & 1  (25.0\%) \\\hline
H3N2& 14  (17.5\%) & 3  (25.0\%) & 2  (33.33\%) & 1  (25.0\%) \\\hline
H7N9& 1  (1.25\%) & 1  (8.33\%) & 1  (16.67\%) & 1  (25.0\%) \\\hline
H9N2& 0  (0.0\%) & 0  (0.0\%) & 0  (0.0\%) & 0  (0.0\%) \\\hline
H5N1& 49  (61.25\%) & 3  (25.0\%) & 0  (0.0\%) & 0  (0.0\%) \\\hline
H1N2& 8  (10.0\%) & 2  (16.67\%) & 1  (16.67\%) & 1  (25.0\%) \\\hline
H3N3& 1  (1.25\%) & 0  (0.0\%) & 0  (0.0\%) & 0  (0.0\%) \\
\hline\end{tabular}

\end{table}
\else
\refstepcounter{table}\label{riskytab}
\fi

\clearpage

\ifFIGS
\begin{table}[!ht]\centering
\captionN{Influenza A strains evaluated by IRAT and corresponding \enet predicted risk scores, sampling 75\% of strains from each subtype in the current human population (averaged over 20 random seeds)}\label{irattabsample}
\sffamily\fontsize{6}{8}\selectfont
\begin{tabular}{L{1.65in}|L{0.26in}|L{0.5in}|L{0.6in}|L{0.3in}|L{0.3in}|L{0.3in}|L{0.3in}|L{0.3in}|L{0.45in}}\hline
 \rowcolor{lightgray!50} Influenza  Virus & Virus  Type & Date  of  Risk  Assessment & Risk  Score  Category & Emergence  Score & Geom  Mean  Risk & Geom  Mean  Risk  SEM & Predicted  Emergence & Predicted  Emergence  SEM & Predicted  Emergence  W/O  Sampling \\\hline
 A/Hong  Kong/125/2017 &H7N9&2017-05-01& Moderate-High &6.5&0.00001&0.00000&7.73&0.00000&7.73\\\hline
 A/California/62/2018 &H1N2&2019-07-01& Moderate &5.8&0.00001&0.00000&7.73&0.00000&7.73\\\hline
 A/Shanghai/02/2013 &H7N9&2016-04-01& Moderate-High &6.4&0.00002&0.00001&7.72&0.01700&7.73\\\hline
 A/Indiana/08/2011 &H3N2&2012-12-01& Moderate &6.0&0.00028&0.00007&7.4&0.06600&7.67\\\hline
 A/Ohio/13/2017 &H3N2&2019-07-01& Moderate &6.6&0.00227&0.00010&6.19&0.02900&6.27\\\hline
 A/Anhui-Lujiang/39/2018 &H9N2&2019-07-01& Moderate &6.2&0.00327&0.00029&5.95&0.06900&6.42\\\hline
 A/swine/Shandong/1207/2016 &H1N1&2020-07-01& Moderate &7.5&0.00366&0.00047&5.9&0.09000&5.93\\\hline
 A/mink/Spain/3691-8\_22VIR10586-10/2022 &H5N1&2023-04-01& Moderate &5.1&0.00371&0.00059&5.99&0.14100&6.06\\\hline
 A/canine/Illinois/12191/2015 &H3N2&2016-06-01&Low&3.7&0.00854&0.00101&5.18&0.08200&4.46\\\hline
 A/Northern  pintail/Washington/40964/2014 &H5N2&2015-03-01& Low-Moderate &3.8&0.00920&0.00103&5.11&0.08500&4.67\\\hline
 A/Sichuan/06681/2021 &H5N6&2021-10-01& Moderate &5.3&0.01089&0.00640&5.9&0.25200&6.58\\\hline
 A/American  green-winged  teal/Washington/1957050/2014 &H5N1&2015-03-01& Low-Moderate &3.6&0.02009&0.00108&4.36&0.04400&4.61\\\hline
 A/Bangladesh/0994/2011 &H9N2&2014-02-01& Moderate &5.6&0.02494&0.00260&4.21&0.06900&3.92\\\hline
 A/Vietnam/1203/2004 &H5N1&2011-11-01& Moderate &5.2&0.03444&0.00184&3.88&0.04600&5.54\\\hline
 A/Yunnan/14564/2015 &H5N6&2016-04-01& Moderate &5.0&0.03817&0.00250&3.8&0.05900&3.68\\\hline
 A/gyrfalcon/Washington/41088/2014 &H5N8&2015-03-01& Low-Moderate &4.2&0.03849&0.00224&3.79&0.06200&3.47\\\hline
 A/Jiangxi-Donghu/346/2013 &H10N8&2014-02-01& Moderate &4.3&0.05588&0.00553&3.5&0.08700&3.24\\\hline
 A/turkey/Indiana/1573-2/2016 &H7N8&2017-07-01&Low&3.4&0.05800&0.00512&3.45&0.08500&3.11\\\hline
 A/chicken/Tennessee/17-007431-3/2017 &H7N9&2017-10-01&Low&3.1&0.05819&0.00524&3.47&0.10400&3.45\\\hline
 A/American  wigeon/South  Carolina/AH0195145/2021 &H5N1&2022-03-01& Moderate &4.4&0.06305&0.00707&3.46&0.13100&5.28\\\hline
 A/chicken/Tennessee/17-007147-2/2017 &H7N9&2017-10-01&Low&2.8&0.07537&0.00629&3.24&0.10900&3.09\\\hline
 A/Netherlands/219/2003 &H7N7&2012-06-01& Moderate &4.6&0.09285&0.00821&3.03&0.08700&2.56\\\hline
 A/Astrakhan/3212/2020 &H5N8&2021-03-01& Moderate &4.6&0.15732&0.00509&2.48&0.02800&2.58\\\hline
\end{tabular}

\flushleft
\fontsize{7}{8}\selectfont
$^\star$ First five columns available from IRAT\cite{Influenz24:online}. Last five columns predicted by \enet.\\
$^{\star\star}$ The strain A/duck/New York/1996, analyzed by IRAT to have emergence score 2.3 on 2011-11-01, is omitted because its HA segment is unavailable.
\end{table}

\vskip 3em

\ifFIGS
\begin{table}[!ht]\centering
\vskip .2em
\captionN{\textcolor{black}{Baseline performance with naive sequence match: correlation with 95\% confidence bounds of edit distance between animal strains and high fitness human strains (H1N1). Sequence match with edit distance metric was limited to  the RBD, and  specific biologically important residue locations within the RBD}}\label{baselinetab}
\sffamily\fontsize{8}{8}\selectfont
\begin{tabular}{L{1in}|L{1in}|L{1in}|L{.8in}|L{1in}|L{.7in}}\hline
\rowcolor{lightgray!50} %
  & HA  \qdist  (2020-2022  sequences) & Geometric  mean  of  HA  and  NA  \qdist  (2020-2022  sequences) & HA  \qdist  (IRAT  sequences) & Geometric  mean  of  HA  and  NA  \qdist  (IRAT  sequences) & CDC computed IRAT emergence  score \\\hline
 On  RBD  (63-286) &$0.671\pm0.005$&$0.773\pm0.007$&$0.214\pm0.061$&$0.376\pm0.071$&-0.43\\\hline
 On  seleted  residues  on  RBD$^\star$ &$-0.21\pm0.007$&$-0.203\pm0.01$&$-0.545\pm0.02$&$-0.234\pm0.03$&-0.16\\\hline
 On  220  loop  (210-240) &$-0.052\pm0.00$&$-0.131\pm0.00$&$-0.523\pm0.09$&$-0.172\pm0.07$&-0.15\\\hline
\end{tabular}

\fontsize{7}{8}\selectfont
\flushleft $^\star$ Residue positions 130,131,132,133,134,135,150,151,152,153,154,155,156,157,158,159,160,190,191,192,193,194,195,220,221,222,223,224,225,226, comprising the 190-helix, the 150-loop, the 220-loop and the 130-loop, which are known to play key role in HA function~\cite{wu2020influenza}
\end{table}
\else
\refstepcounter{table}\label{baselinetab}
\fi

\clearpage

\ifFIGS
\begin{figure*}[!ht]
  \centering
  \tikzexternalenable
    \tikzsetnextfilename{sequence}
\vspace{-5pt}
  \iftikzX
  \input{Figures/sequence_}  
  \vspace{0pt}   
  
  \else
  \includegraphics[width=0.95\textwidth]{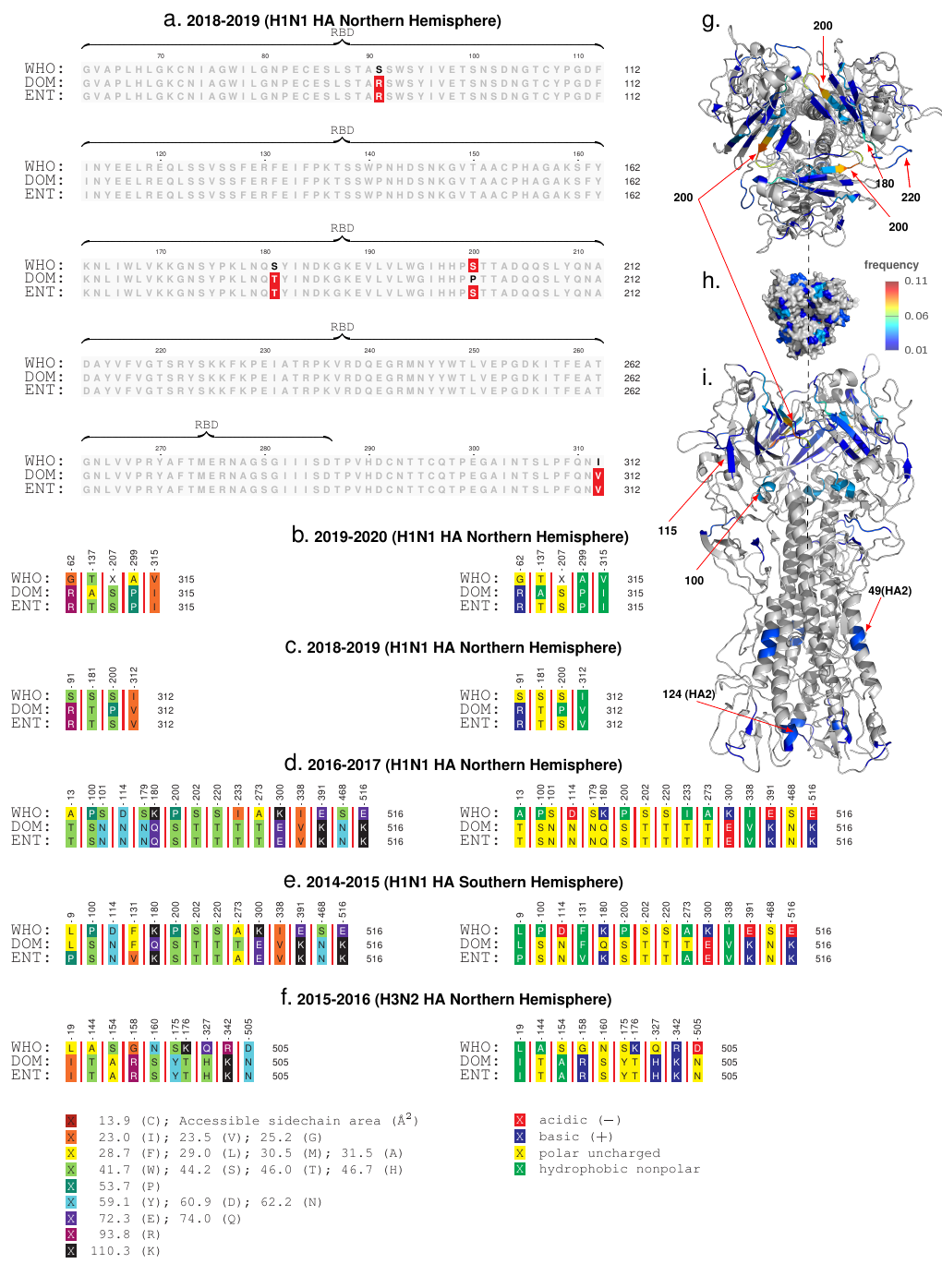}  \vspace{-5pt}   

  \fi
  
\vspace{0pt}

\captionN{\textbf{Sequence comparisons.} \textbf{Panel a} Comparing the \enet  (ENT) and the WHO recommendation (WHO), and the observed dominant strain (DOM), we note that the correct \enet  predictions tend to be within the RBD, both for H1N1 and H3N2 for HA. \textbf{Panels b-f} Additionally, by comparing the type, side chain area, and the accessible side chain area, we note that DOM and ENT are often close in important chemical properties, while WHO deviations are  not. \textbf{Panels g-i} show the localization of the deviations in the molecular structure of HA, where we note that the changes are most frequent in the HA1 sub-unit (the globular head), and around residues and structures that have been commonly implicated in receptor binding interactions $e.g$ the $\approx 200$ loop, the $\approx 220$ loop and the $\approx 180$-helix~\cite{tzarum2015structure,lazniewski2018structural,garcia2015dynamic}.}\label{figseq}
\end{figure*}
\else
\refstepcounter{figure}\label{figseq}
\fi

\ifFIGS

\begin{table}[b]\centering
  \captionN{\textcolor{black}{Top ranked risky strains amongst $6,354$ animal strains collected post-2020 unique upto 15 edits in HA sequence}}\label{highrisktab}

  \vspace{-8pt}

\bf\sffamily\fontsize{7}{5}\selectfont

\begin{tabular}{L{2.75in}|L{.45in}|L{.60in}|L{.6in}|C{1in}}\hline
Strain&Subtype& HA  accession & NA  accession & Predicted  IRAT  emergence \\
\rowcolor{colh7n9x!40}A/Shanghai/02/2013&H7N9&EPI448936&EPI448938&7.7327\\
\rowcolor{colh3n2x!30}A/swine/North\_Carolina/A02751333/2022&H3N2&EPI2396893&EPI2396894&7.7327\\
\rowcolor{colh1n2x!30}A/swine/Iowa/CEIRS-1495/2023&H1N2&EPI2908633&EPI2908631&7.7327\\
\rowcolor{colh1n1x!40}A/swine/Tver\_region/RII-81-1S/2023&H1N1&EPI2965146&EPI2965145&7.7327\\
\rowcolor{colh3n2x!30}A/Indiana/08/2011&H3N2&EPI344405&EPI344404&7.6364\\
\rowcolor{colh1n1x!40}A/swine/North\_Carolina/A02479173/2020&H1N1&EPI1780425&EPI1780426&7.5991\\
\rowcolor{colh5n1x!20}A/Common\_Buzzard/England/125155/2023&H5N1&EPI2574053&EPI2574052&7.4234\\
\rowcolor{colh5n1x!20}A/backyard\_chicken/Uruguay/UDELAR-144-M3/2023&H5N1&EPI2758928&EPI2758926&7.4152\\
\rowcolor{colh5n1x!20}A/Tufted\_duck/Netherlands/1/2023&H5N1&EPI2904803&EPI2904802&7.4125\\
\rowcolor{colh1n2x!30}A/swine/Iowa/A02479084/2020&H1N2&EPI1779472&EPI1779473&7.3386\\
\rowcolor{colh3n2x!30}A/swine/Minnesota/A02245579/2020&H3N2&EPI1775813&EPI1775814&7.2054\\
\rowcolor{colh1n1x!40}A/swine/ISU-Illinois/A02861867/2023&H1N1&EPI2908017&EPI2908016&7.1952\\
\rowcolor{colh1n1x!40}A/swine/Iowa/A02751092/2022&H1N1&EPI2381136&EPI2381137&6.9698\\
\rowcolor{colh5n1x!20}A/peregrine\_falcon/Netherlands/22000191-001/2022&H5N1&EPI1980868&EPI1980867&6.9508\\
\rowcolor{colh5n1x!20}A/otter/Finland/2860\_21VIR9619-5/2021&H5N1&EPI2197232&EPI2197231&6.9362\\
\rowcolor{colh1n2x!30}A/swine/North\_Carolina/A02751482/2023&H1N2&EPI2456343&EPI2456344&6.9019\\
\rowcolor{colh5n1x!20}A/gull/Estonia/TA2113284-4\_21VIR7512-8/2021&H5N1&EPI1945389&EPI1945388&6.8990\\
\rowcolor{colh5n1x!20}A/lynx/Finland/Vi209-22\_22VIR3127-2/2022&H5N1&EPI2197224&EPI2197223&6.8905\\
\rowcolor{colh5n1x!20}A/domestic\_duck/England/076401/2021&H5N1&EPI2071928&EPI2071927&6.8736\\
\rowcolor{colh3n2x!30}A/swine/Illinois/A02525253/2021&H3N2&EPI1910375&EPI1910376&6.8375\\
\rowcolor{colh1n2x!30}A/swine/South\_Korea/GNJJ/2020&H1N2&EPI2258341&EPI2258343&6.8270\\
\rowcolor{colh5n1x!20}A/Great\_skua/Iceland/2023AI04253/2022&H5N1&EPI2635214&EPI2635213&6.8097\\
\rowcolor{colh5n1x!20}A/Gyrfalcon/Iceland/2023AI04255/2022&H5N1&EPI2635279&EPI2635277&6.8093\\
\rowcolor{colh5n1x!20}A/herring\_gull/England/248639/2022&H5N1&EPI2158600&EPI2158599&6.8089\\
\rowcolor{colh5n1x!20}A/chicken/England/085598/2022&H5N1&EPI2089022&EPI2089021&6.8089\\
\rowcolor{colh5n1x!20}A/lesser\_black-backed\_gull/Netherlands/22012469-002/2022&H5N1&EPI2137792&EPI2137790&6.8089\\
\rowcolor{colh5n1x!20}A/Anser\_anser/Spain/1035-5\_22VIR6312-8/2022&H5N1&EPI2102649&EPI2102648&6.8089\\
\rowcolor{colh5n1x!20}A/Mallard/Netherlands/18/2022&H5N1&EPI2197857&EPI2197856&6.8089\\
\rowcolor{colh5n1x!20}A/turkey/Italy/21VIR10251/2021&H5N1&EPI1944389&EPI1944388&6.8089\\
\rowcolor{colh5n1x!20}A/grey\_heron/Netherlands/21038941-001/2021&H5N1&EPI1941418&EPI1941417&6.8089\\
\rowcolor{colh5n1x!20}A/duck/France/21343/2021&H5N1&EPI2536840&EPI2536875&6.8089\\
\rowcolor{colh5n1x!20}A/duck/Spain/2095-2\_22VIR8632-5/2022&H5N1&EPI2191149&EPI2191148&6.8085\\
\rowcolor{colh5n1x!20}A/goose/France/22P004055/2022&H5N1&EPI2780410&EPI2780409&6.8072\\
\rowcolor{colh5n1x!20}A/turkey/Poland/H1923\_21RS3290-1/2021&H5N1&EPI2782408&EPI2782407&6.8070\\
\rowcolor{colh5n1x!20}A/laying\_hen/Italy/21VIR11168/2021&H5N1&EPI2140973&EPI2140972&6.8066\\
\rowcolor{colh5n1x!20}A/turkey/Italy/21VIR8585-1/2021&H5N1&EPI1923192&EPI1923194&6.8066\\
\rowcolor{colh5n1x!20}A/chicken/Czech\_Republic/61-1/2022&H5N1&EPI1999361&EPI1999360&6.8066\\
\rowcolor{colh5n1x!20}A/turkey/Italy/21VIR10001/2021&H5N1&EPI2139999&EPI2139998&6.8066\\
\rowcolor{colh5n1x!20}A/chicken/Italy/21VIR9940-5/2021&H5N1&EPI2140039&EPI2140038&6.8066\\
\rowcolor{colh5n1x!20}A/chicken/Italy/21VIR10094/2021&H5N1&EPI2140083&EPI2140082&6.8053\\
\rowcolor{colh5n1x!20}A/chicken/Italy/21VIR10263/2021&H5N1&EPI2140199&EPI2140198&6.8043\\
\rowcolor{colh5n1x!20}A/goose/Czech\_Republic/25322-179/2021&H5N1&EPI2021894&EPI2021896&6.8035\\
\rowcolor{colh5n1x!20}A/broiler/Italy/21VIR11886-1/2021&H5N1&EPI2141602&EPI2141601&6.8025\\
\rowcolor{colh1n2x!30}A/swine/France/59-200284/2020&H1N2&EPI1976498&EPI1976500&6.8016\\
\rowcolor{colh5n1x!20}A/swan/Slovenia/2049\_22VIR777-3/2021&H5N1&EPI1995084&EPI1995083&6.8003\\
\rowcolor{colh5n1x!20}A/turkey/Italy/21VIR10852/2021&H5N1&EPI2140553&EPI2140552&6.7956\\
\rowcolor{colh5n1x!20}A/turkey/Wales/065047/2021&H5N1&EPI2071520&EPI2071519&6.7830\\
\rowcolor{colh5n1x!20}A/chicken/Luxembourg/23023602/2023&H5N1&EPI2364554&EPI2364572&6.7826\\
\rowcolor{colh5n1x!20}A/Common\_Teal/Netherlands/1/2022&H5N1&EPI2185431&EPI2185430&6.7823\\
\rowcolor{colh5n1x!20}A/Common\_Tern/Netherlands/12/2022&H5N1&EPI2119243&EPI2119242&6.7804\\
\rowcolor{colh5n1x!20}A/turkey/Spain/140-38\_22VIR2142-19/2022&H5N1&EPI1998202&EPI1998201&6.7697\\
\rowcolor{colh3n2x!30}A/swine/Minnesota/A02524797/2020&H3N2&EPI1907265&EPI1907266&6.7679\\
\rowcolor{colh3n2x!30}A/canine/Pennsylvania/CVM-985419/2023&H3N2&EPI2906952&EPI2906950&6.7569\\
\rowcolor{colh5n1x!20}A/Herring\_gull/France/22P019328/2022&H5N1&EPI2780738&EPI2780737&6.7308\\
\rowcolor{colh5n1x!20}A/otter/Netherlands/22001014-005/2022&H5N1&EPI1965229&EPI1965228&6.7203\\
\rowcolor{colh5n1x!20}A/Harbour\_Seal/Scotland/162919/2022&H5N1&EPI2398403&EPI2398402&6.7092\\
\rowcolor{colh1n2x!30}A/swine/Denmark/S20996-3/2021&H1N2&EPI1980810&EPI1980812&6.7069\\
\rowcolor{colh1n2x!30}A/swine/Germany/2021AI04886/2021\_(H1pdmN2&H1N2&EPI2551752&EPI2551751&6.6937\\
\rowcolor{colh5n1x!20}A/American\_Crow/BC/AIVPHL-1468/2023&H5N1&EPI2856554&EPI2856544&6.6846\\
\rowcolor{colh5n1x!20}A/chicken/Italy/21VIR10812-11/2021&H5N1&EPI2140660&EPI2140658&6.6807\\
\rowcolor{colh5n1x!20}A/great\_crested\_grebe/Netherlands/22001219-002/2022&H5N1&EPI1990440&EPI1990439&6.6777\\
\rowcolor{colh5n1x!20}A/fox/England/015850/2022&H5N1&EPI2437468&EPI2437467&6.6711\\
\rowcolor{colh1n1x!40}A/swine/Germany/2022AI00363/2022\_(H1avN1)&H1N1&EPI2551584&EPI2551583&6.6577\\
\rowcolor{colh5n1x!20}A/barnacle\_goose/Netherlands/22007405-004/2022&H5N1&EPI2028769&EPI2028768&6.6518\\
\rowcolor{colh1n2x!30}A/swine/England/129502/2023&H1N2&EPI2818798&EPI2818800&6.6469\\
\rowcolor{colh3n2x!30}A/swine/Iowa/ISU-A02862194/2023&H3N2&EPI2971632&EPI2971631&6.6454\\
\rowcolor{colh1n1x!40}A/swine/Iowa/A02636162/2021&H1N1&EPI1932974&EPI1932975&6.6453\\
\rowcolor{colh3n2x!30}A/swine/Iowa/ISU-A02861853/2023&H3N2&EPI2874299&EPI2874298&6.6389\\
\hline\end{tabular}

\end{table}
\else
\refstepcounter{table}\label{highrisktab}
\fi

\clearpage

\ifFIGS
  
\begin{figure}[t]
  \tikzexternalenable
  \tikzsetnextfilename{riskyseq_h3n2}
  \centering
  \iftikzX  
  \begin{tikzpicture}[font=\bf\sffamily\fontsize{8}{8}\selectfont]
  \def\SEQA{Figures/plotdata/risky4.fasta}
  \def\SEQB{Figures/plotdata/risky4_h3n2.fasta}
  \def\SEQC{Figures/plotdata/risky4_domtop.fasta}
  \def\SEQC{Figures/plotdata/risky4_h3n2.fasta}
  \def\LENA{550}
  \def\LENB{63}
  \def\LENC{286}
  \def\LENE{1}
  \def\LEND{550}
  \def\COLM{jet}
  \def\rndfileA{rndfile1.png}
  \def\rndfileB{rndfile2.png}
  \def\rndfileC{rndfile3.png} 

  \newcommand{\panelX}[2] {
    \begin{tikzpicture}[font=\bf\sffamily\fontsize{7}{7}\selectfont]
      \node[ ] (A) at (0,0) {
        \mnp{3.2in}{\begin{texshade}{#1}
            \shadingmode[accessible area]{functional}
            \hideallmatchpositions
            \rulersteps{1}
            \setfont{residues}{sf}{up}{bf}{tiny} 
            \setfont{numbering}{sf}{up}{bf}{tiny} 
            \setfont{names}{tt}{up}{bf}{small}
            \setfont{legend}{tt}{up}{bf}{scriptsize}
            \threshold[80]{50}
            \setends{1}{1..\LENA}
            \showruler{1}{top}
            \hideconsensus
            \shadeallresidues
            #2
          \end{texshade}}};
\node[] (B) at (A.north east) {  \mnp{3.5in}{      
          \begin{texshade}{#1}
            \shadingmode[hydropathy]{functional}
            \hideallmatchpositions
            \rulersteps{1}
            \setfont{residues}{sf}{up}{bf}{tiny} 
            \setfont{numbering}{sf}{up}{bf}{tiny} 
            \setfont{names}{tt}{up}{bf}{small}
            \setfont{legend}{tt}{up}{bf}{scriptsize}
            \threshold[80]{50}
            \setends{1}{1..\LENA}
            \showruler{1}{top}
            \hideconsensus
            \shadeallresidues
            #2
          \end{texshade}}};
    \end{tikzpicture}
    }

  \node[%
  ] (T1) at (0,0){  
    \begin{tikzpicture}
      \node [
      ]
      (A) at (0,0.0) {
        \mnp{\textwidth}{
          \begin{texshade}{\SEQC}
           \residuesperline*{70}
           \shadingmode[allmatchspecial]{identical}
            \shadingcolors{grays}
            \conservedresidues{Red1!40}{lightgray!10}{upper}{bf}
            \allmatchresidues{Red1!50}{lightgray!10}{upper}{bf}
            \nomatchresidues{white}{Red1!70}{upper}{bf}
            \setfont{residues}{sf}{up}{bf}{tiny} 
            \setfont{numbering}{sf}{up}{bf}{tiny} 
            \setfont{names}{tt}{up}{bf}{small}
            \setfont{legend}{tt}{up}{bf}{scriptsize}
            \setfont{features}{tt}{up}{bf}{scriptsize}
            \feature{top}{1}{\LENB..\LENC}{brace[black]}{RBD}
            \setends{1}{\LENE..\LEND}
            \showruler{1}{top}
            \hideconsensus
             \hidelegend
          \end{texshade}
        }};
    \end{tikzpicture}};

\end{tikzpicture}
 \else
  \includegraphics[width=.975\textwidth]{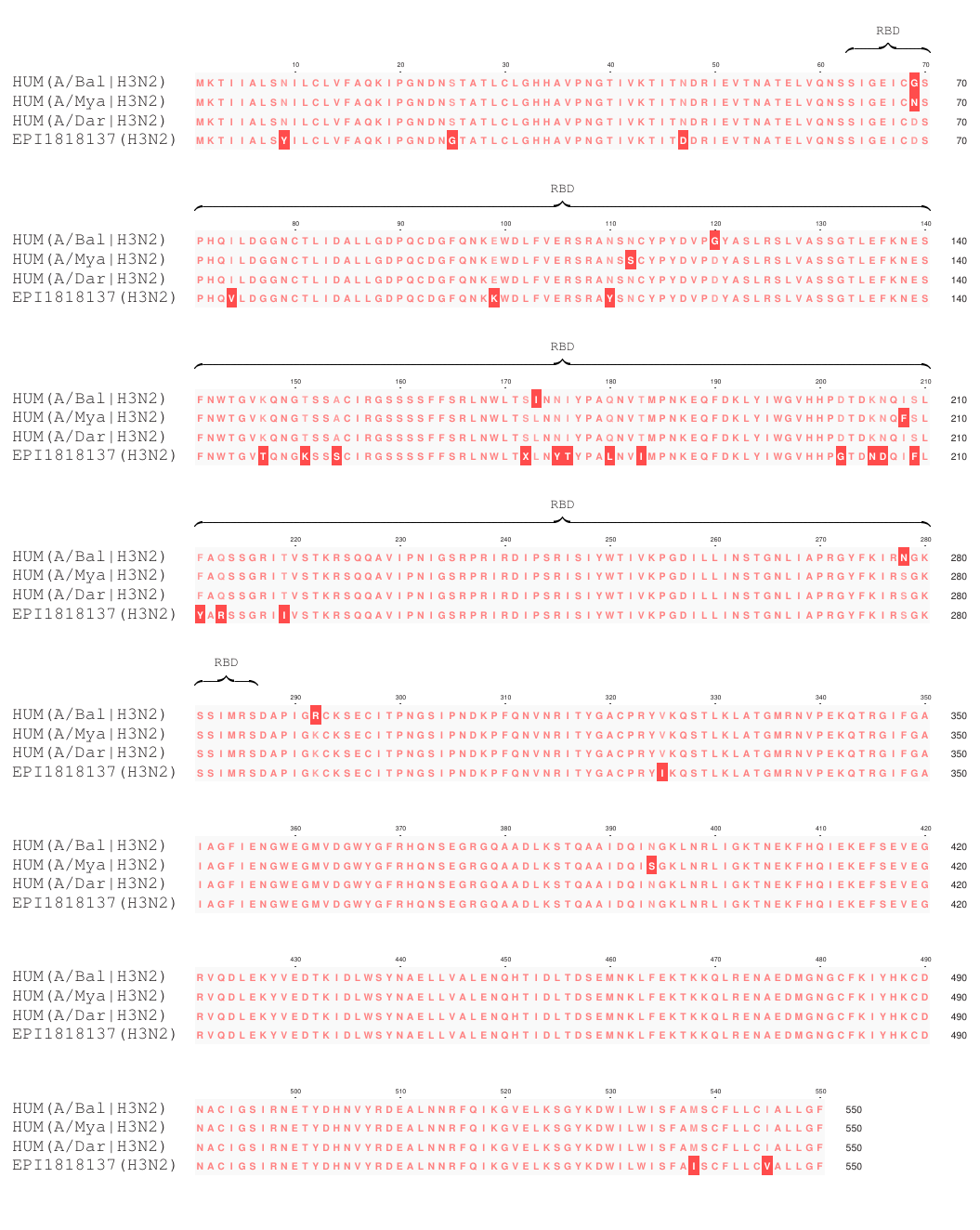}
  \fi 
  \vspace{-18pt}
  
\captionN{\textcolor{black}{\textbf{High-risk animal strain comparison.} HA sequence comparison  with 2020-2021 dominat frequency human strains (A/Baltimore/JH/001/2021, A/Myanmar/I026/2021, A/Darwin/12/2021)  with \enet estimated top H3N2 risky strain EPI1818137 (emergence score $> 7.26$, 2020-2022 April)  showing  differences in and out of the RBD.}}\label{figriskyseq}
\end{figure}
\else
\refstepcounter{figure}\label{figriskyseq}
\fi

\clearpage

\section*{Supplementary Figures \& Tables}

\setcounter{figure}{0}
\renewcommand{\figurename}{S-Fig.}
\setcounter{table}{0}
\renewcommand{\tablename}{S-Tab.}

\begin{table}[H]
\mnp{\textwidth}{  
    \centering
    \captionN{Number of Influenza sequences collected from public databases}\label{tabseq}
    \sffamily\fontsize{8}{8}\selectfont
    \begin{tabular}{L{0.5in}|L{0.5in}|L{0.7in}|L{0.7in}|L{0.5in}}\hline
\rowcolor{lightgray!50} Database & Influenza  Subtype & No.  HA  Sequences & No.  NA  Sequences & Total \\\hline
GISAID&H1N1&73,905&73,920&147,825\\\hline
NCBI&H1N1&18,577&16,913&35,490\\\hline
GISAID&H3N2&108,829&108,860&217,689\\\hline
NCBI&H3N2&18,840&15,249&34,089\\\hline
GISAID&H1N2&1,340&1,340&2,680\\\hline
GISAID&H1N7&18&18&36\\\hline
GISAID&H3N8&406&405&811\\\hline
GISAID&H4N6&68&68&136\\\hline
GISAID&H5N1&8,245&8,145&16,390\\\hline
GISAID&H5N2&35&35&70\\\hline
GISAID&H5N3&44&43&87\\\hline
GISAID&H5N5&74&74&148\\\hline
GISAID&H5N6&282&282&564\\\hline
GISAID&H5N8&1,561&1,513&3,074\\\hline
GISAID&H6N1&31&31&62\\\hline
GISAID&H6N2&47&41&88\\\hline
GISAID&H6N6&62&62&124\\\hline
GISAID&H7N3&122&120&242\\\hline
GISAID&H7N9&1,274&1,273&2,547\\\hline
GISAID&H9N2&451&451&902\\\hline
GISAID&H10N3&43&43&86\\\hline
GISAID&H10N7&15&15&30\\\hline
GISAID&H11N9&33&33&66\\\hline
GISAID&H13N6&15&15&30\\\hline
Both&Total&234,317&228,949&463,266\\\hline
\end{tabular}

}
\end{table}

\vskip 1em

\begin{table}[H]
    \def\ACOL{teal!30}
    \def\BCOL{Red1!30}
    \def\CCOL{gray!30}
    \captionN{Examples: \enet induced distance varying for fixed sequence pair when background population changes (rows 1-5), sequences with small edit distance and large \qdist, and the converse (rows 6-9)}\label{tabex}
    \sffamily\fontsize{8}{8}\selectfont
    \centering
    \begin{tabular}{L{.5in}|L{1.6in}|L{1.6in}|L{.6in}|L{.4in}|L{.4in}}\hline
\rowcolor{lightgray!50}Hamming Distance & Sequence A & Sequence B & \enet E-Distance & Year A$^\star$ & Year B$^\star$\\\hline
18 & A/Singapore/23J/2007 & A/Tennessee/UR06-0294/2007 & 0.0111 & 2007 & 2007\\\hline
18 & A/Singapore/23J/2007 & A/Tennessee/UR06-0294/2007 & 0.0094 & 2008 & 2008\\\hline
18 & A/Singapore/23J/2007 & A/Tennessee/UR06-0294/2007 & 0.0027 & 2009 & 2009\\\hline
18 & A/Singapore/23J/2007 & A/Tennessee/UR06-0294/2007 & 0.0025 & 2010 & 2010\\\hline
18 & A/Singapore/23J/2007 & A/Tennessee/UR06-0294/2007 & 0.6163 & 2007 & 2010\\\hline
\rowcolor{\ACOL}11 & A/Naypyitaw/M783/2008 & A/Singapore/201/2008 &      0.8852 & 2008 & 2008\\\hline
\rowcolor{\ACOL}15 & A/Cambodia/W0908339/2012 & A/Singapore/DMS1233/2012&0.2737 & 2012 & 2012\\\hline
\rowcolor{\BCOL}126 & A/South Dakota/03/2008 & A/Singapore/10/2008 &     0.3034 & 2008 & 2008\\\hline
\rowcolor{\BCOL}141 & A/Jodhpur/3248/2012 & A/Cambodia/W0908339/2012 &   0.2405 & 2012 & 2012\\\hline
\end{tabular}

    \flushleft
    \fontsize{7}{8}\selectfont
    $^\star$ Year A and year B correspond to the assumed collection years for sequences A and B respectively for the purpose of this example. Sequence A in row 1 is collected in 2007, but is assumed to be from different years in rows 2-4 to demonstrate the change in \qdist from sequence B, arising only from a change in the background population.
\end{table}

\vskip 1em

\begin{table}[H]
\mnp{\textwidth}{
    \centering
    \captionN{Correlation between \qdist and edit distance between sequence pairs}\label{tabcor}
    \sffamily\fontsize{8}{8}\selectfont
    \begin{tabular}{L{1.35in}|L{.65in}}\hline
        \rowcolor{lightgray!50}Phenotypes & Correlation \\\hline
        Influenza H1N1 HA  &0.76\\\hline
        Influenza H1N1 NA &0.74\\\hline
        Influenza H3N2 HA &0.85\\\hline
        Influenza H3N2 NA &0.79\\\hline
    \end{tabular}
}
\end{table}

\ifFIGS
  
\begin{figure}[t]
  \tikzexternalenable
  \tikzsetnextfilename{riskyseq_h3n2_lr}
  \centering
\tikzXtrue
  \iftikzX  
  \begin{tikzpicture}[font=\bf\sffamily\fontsize{8}{8}\selectfont]
  \def\SEQA{Figures/plotdata/risky4.fasta}
  \def\SEQB{Figures/plotdata/risky4_h3n2.fasta}
  \def\SEQC{Figures/plotdata/risky4_domtop.fasta}
  \def\SEQC{Figures/plotdata/risky4_h3n2_lowrisk.fasta}
  \def\LENA{550}
  \def\LENB{63}
  \def\LENC{286}
  \def\LENE{1}
  \def\LEND{550}
  \def\COLM{jet}
  \def\rndfileA{rndfile1.png}
  \def\rndfileB{rndfile2.png}
  \def\rndfileC{rndfile3.png} 

  \newcommand{\panelX}[2] {
    \begin{tikzpicture}[font=\bf\sffamily\fontsize{7}{7}\selectfont]
      \node[ ] (A) at (0,0) {
        \mnp{3.2in}{\begin{texshade}{#1}
            \shadingmode[accessible area]{functional}
            \hideallmatchpositions
            \rulersteps{1}
            \setfont{residues}{sf}{up}{bf}{tiny} 
            \setfont{numbering}{sf}{up}{bf}{tiny} 
            \setfont{names}{tt}{up}{bf}{small}
            \setfont{legend}{tt}{up}{bf}{scriptsize}
            \threshold[80]{50}
            \setends{1}{1..\LENA}
            \showruler{1}{top}
            \hideconsensus
            \shadeallresidues
            #2
          \end{texshade}}};
\node[] (B) at (A.north east) {  \mnp{3.5in}{      
          \begin{texshade}{#1}
            \shadingmode[hydropathy]{functional}
            \hideallmatchpositions
            \rulersteps{1}
            \setfont{residues}{sf}{up}{bf}{tiny} 
            \setfont{numbering}{sf}{up}{bf}{tiny} 
            \setfont{names}{tt}{up}{bf}{small}
            \setfont{legend}{tt}{up}{bf}{scriptsize}
            \threshold[80]{50}
            \setends{1}{1..\LENA}
            \showruler{1}{top}
            \hideconsensus
            \shadeallresidues
            #2
          \end{texshade}}};
    \end{tikzpicture}
    }

  \node[%
  ] (T1) at (0,0){  
    \begin{tikzpicture}
      \node [
      ]
      (A) at (0,0.0) {
        \mnp{\textwidth}{
          \begin{texshade}{\SEQC}
           \residuesperline*{70}
           \shadingmode[allmatchspecial]{identical}
            \shadingcolors{grays}
            \conservedresidues{Red1!40}{lightgray!10}{upper}{bf}
            \allmatchresidues{Red1!50}{lightgray!10}{upper}{bf}
            \nomatchresidues{white}{Red1!70}{upper}{bf}
            \setfont{residues}{sf}{up}{bf}{tiny} 
            \setfont{numbering}{sf}{up}{bf}{tiny} 
            \setfont{names}{tt}{up}{bf}{small}
            \setfont{legend}{tt}{up}{bf}{scriptsize}
            \setfont{features}{tt}{up}{bf}{scriptsize}
            \feature{top}{1}{\LENB..\LENC}{brace[black]}{RBD}
            \setends{1}{\LENE..\LEND}
            \showruler{1}{top}
            \hideconsensus
             \hidelegend
          \end{texshade}
        }};
    \end{tikzpicture}};

\end{tikzpicture}
 \else
  \includegraphics[width=.975\textwidth]{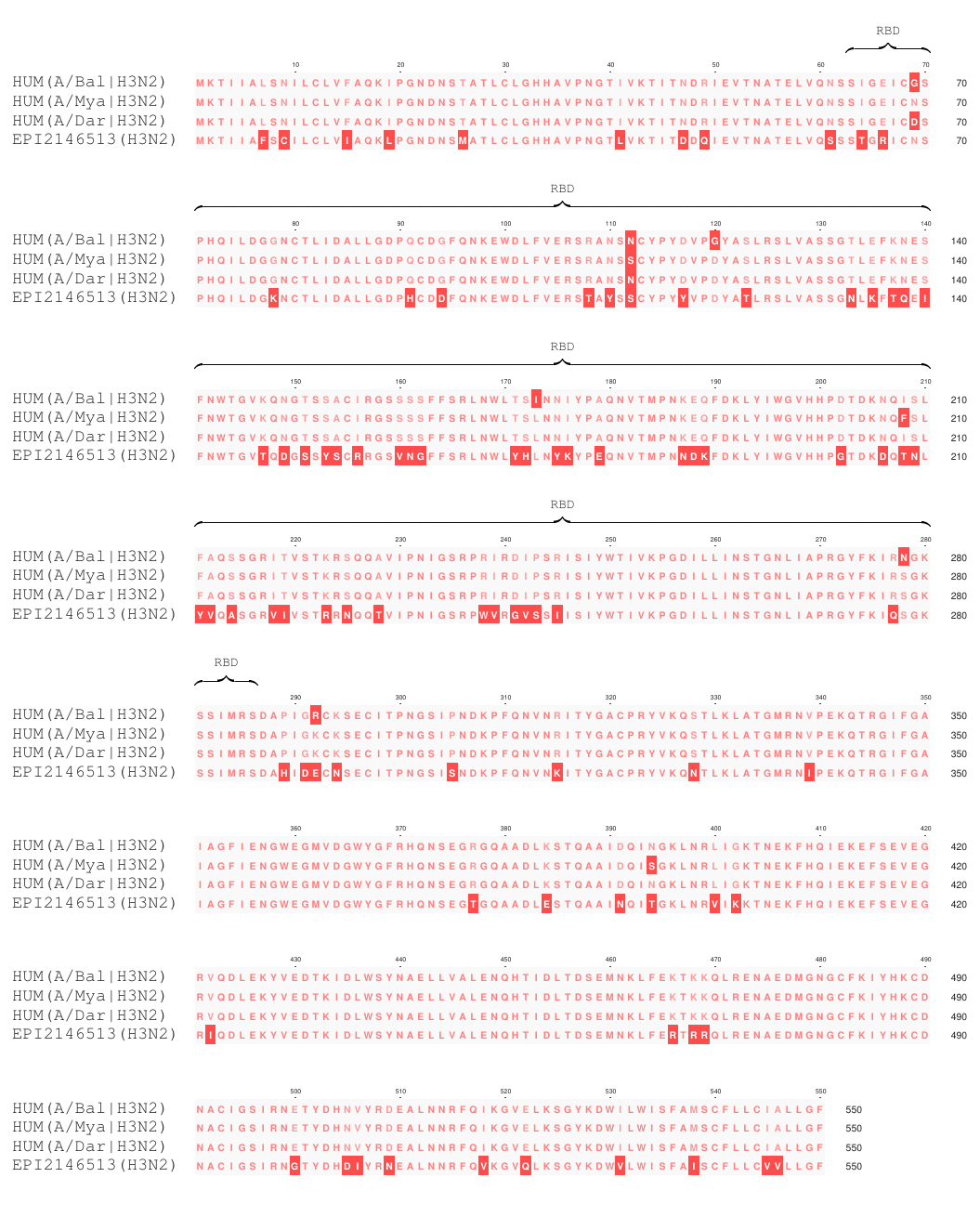}
  \fi 
  \vspace{-20pt}
  
\captionN{\textcolor{black}{\textbf{Low-risk animal strain comparison.} HA sequence comparison  with 2020-2021 dominant frequency human strains (A/Baltimore/JH/001/2021, A/Myanmar/I026/2021, A/Darwin/12/2021)  with \enet estimated relatively medium-risk H3N2 EPI2146849  (emergence risk score 6.5) showing  substantially more  differences compared to high-risk strain comaprison shown in \EXTENDED Fig.~\ref{figriskyseqlr}.}}\label{figriskyseqlr}
\end{figure}
\else
\refstepcounter{figure}\label{figriskyseqlr}
\fi

\ifFIGS
  
\begin{figure}[t]
  \tikzexternalenable
  \tikzsetnextfilename{riskyseq_h1n1}
  \centering
  \iftikzX  
  \begin{tikzpicture}[font=\bf\sffamily\fontsize{8}{8}\selectfont]
  \def\SEQA{Figures/plotdata/risky4.fasta}
  \def\SEQB{Figures/plotdata/risky4_h3n2.fasta}
  \def\SEQC{Figures/plotdata/risky4_domtop.fasta}
  \def\SEQC{Figures/plotdata/risky4_h1n1.fasta}
  \def\LENA{550}
  \def\LENB{63}
  \def\LENC{286}
  \def\LENE{1}
  \def\LEND{550}
  \def\COLM{jet}
  \def\rndfileA{rndfile1.png}
  \def\rndfileB{rndfile2.png}
  \def\rndfileC{rndfile3.png} 

  \newcommand{\panelX}[2] {
    \begin{tikzpicture}[font=\bf\sffamily\fontsize{7}{7}\selectfont]
      \node[ ] (A) at (0,0) {
        \mnp{3.2in}{\begin{texshade}{#1}
            \shadingmode[accessible area]{functional}
            \hideallmatchpositions
            \rulersteps{1}
            \setfont{residues}{sf}{up}{bf}{tiny} 
            \setfont{numbering}{sf}{up}{bf}{tiny} 
            \setfont{names}{tt}{up}{bf}{small}
            \setfont{legend}{tt}{up}{bf}{scriptsize}
            \threshold[80]{50}
            \setends{1}{1..\LENA}
            \showruler{1}{top}
            \hideconsensus
            \shadeallresidues
            #2
          \end{texshade}}};
\node[] (B) at (A.north east) {  \mnp{3.5in}{      
          \begin{texshade}{#1}
            \shadingmode[hydropathy]{functional}
            \hideallmatchpositions
            \rulersteps{1}
            \setfont{residues}{sf}{up}{bf}{tiny} 
            \setfont{numbering}{sf}{up}{bf}{tiny} 
            \setfont{names}{tt}{up}{bf}{small}
            \setfont{legend}{tt}{up}{bf}{scriptsize}
            \threshold[80]{50}
            \setends{1}{1..\LENA}
            \showruler{1}{top}
            \hideconsensus
            \shadeallresidues
            #2
          \end{texshade}}};
    \end{tikzpicture}
    }

  \node[%
  ] (T1) at (0,0){  
    \begin{tikzpicture}
      \node [
      ]
      (A) at (0,0.0) {
        \mnp{\textwidth}{
          \begin{texshade}{\SEQC}
           \residuesperline*{70}
           \shadingmode[allmatchspecial]{identical}
            \shadingcolors{grays}
            \conservedresidues{Red1!40}{lightgray!10}{upper}{bf}
            \allmatchresidues{Red1!50}{lightgray!10}{upper}{bf}
            \nomatchresidues{white}{Red1!70}{upper}{bf}
            \setfont{residues}{sf}{up}{bf}{tiny} 
            \setfont{numbering}{sf}{up}{bf}{tiny} 
            \setfont{names}{tt}{up}{bf}{small}
            \setfont{legend}{tt}{up}{bf}{scriptsize}
            \setfont{features}{tt}{up}{bf}{scriptsize}
            \feature{top}{1}{\LENB..\LENC}{brace[black]}{RBD}
            \setends{1}{\LENE..\LEND}
            \showruler{1}{top}
            \hideconsensus
             \hidelegend
          \end{texshade}
        }};
    \end{tikzpicture}};

\end{tikzpicture}
 \else
  \includegraphics[width=.975\textwidth]{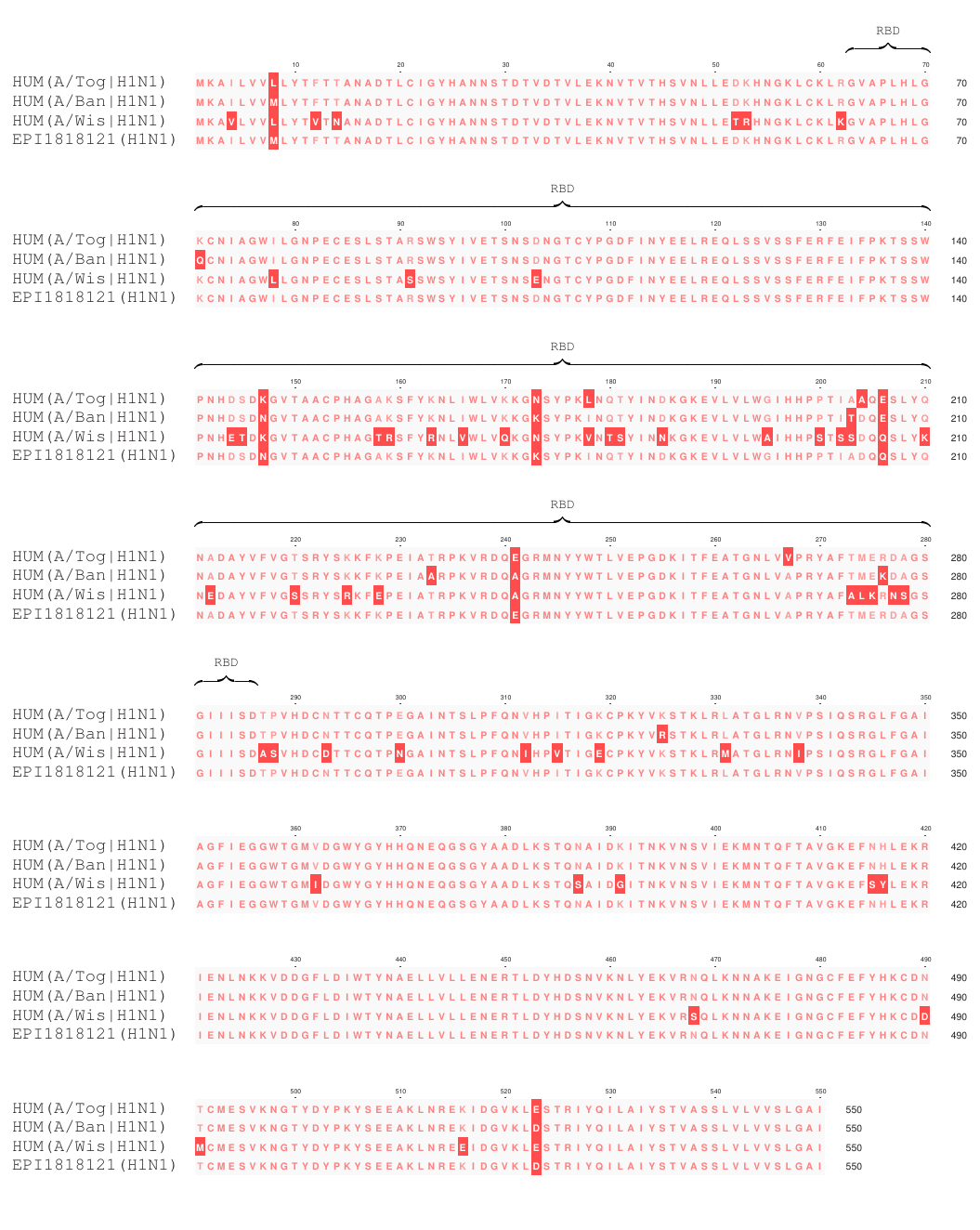}
  \fi 
  \vspace{-18pt}
  
\captionN{\textcolor{black}{\textbf{High-risk comparison.} HA sequence comparison  with dominant frequency human strains (A/Togo/0172/2021, A/Bangladesh/9004/2021,  A/Wisconsin/04/2021)  with \enet estimated top H1N1 risky strain EPI1818121  (2020-2022 April, emergence score 7.69) showing  differences  both in and out of the RBD.}}\label{figriskyseq4}
\end{figure}
\else
\refstepcounter{figure}\label{figriskyseq4}
\fi

\ifFIGS
  
\begin{figure}[t]
  \tikzexternalenable
  \tikzsetnextfilename{riskyseq_h1n1_lr}
  \centering
  \iftikzX  
  \begin{tikzpicture}[font=\bf\sffamily\fontsize{8}{8}\selectfont]
  \def\SEQA{Figures/plotdata/risky4.fasta}
  \def\SEQB{Figures/plotdata/risky4_h3n2.fasta}
  \def\SEQC{Figures/plotdata/risky4_domtop.fasta}
  \def\SEQC{Figures/plotdata/risky4_h1n1_lowrisk.fasta}
  \def\LENA{550}
  \def\LENB{63}
  \def\LENC{286}
  \def\LENE{1}
  \def\LEND{550}
  \def\COLM{jet}
  \def\rndfileA{rndfile1.png}
  \def\rndfileB{rndfile2.png}
  \def\rndfileC{rndfile3.png} 

  \newcommand{\panelX}[2] {
    \begin{tikzpicture}[font=\bf\sffamily\fontsize{7}{7}\selectfont]
      \node[ ] (A) at (0,0) {
        \mnp{3.2in}{\begin{texshade}{#1}
            \shadingmode[accessible area]{functional}
            \hideallmatchpositions
            \rulersteps{1}
            \setfont{residues}{sf}{up}{bf}{tiny} 
            \setfont{numbering}{sf}{up}{bf}{tiny} 
            \setfont{names}{tt}{up}{bf}{small}
            \setfont{legend}{tt}{up}{bf}{scriptsize}
            \threshold[80]{50}
            \setends{1}{1..\LENA}
            \showruler{1}{top}
            \hideconsensus
            \shadeallresidues
            #2
          \end{texshade}}};
\node[] (B) at (A.north east) {  \mnp{3.5in}{      
          \begin{texshade}{#1}
            \shadingmode[hydropathy]{functional}
            \hideallmatchpositions
            \rulersteps{1}
            \setfont{residues}{sf}{up}{bf}{tiny} 
            \setfont{numbering}{sf}{up}{bf}{tiny} 
            \setfont{names}{tt}{up}{bf}{small}
            \setfont{legend}{tt}{up}{bf}{scriptsize}
            \threshold[80]{50}
            \setends{1}{1..\LENA}
            \showruler{1}{top}
            \hideconsensus
            \shadeallresidues
            #2
          \end{texshade}}};
    \end{tikzpicture}
    }

  \node[%
  ] (T1) at (0,0){  
    \begin{tikzpicture}
      \node [
      ]
      (A) at (0,0) {
        \mnp{\textwidth}{
          \begin{texshade}{\SEQC}
           \residuesperline*{70}
           \shadingmode[allmatchspecial]{identical}
            \shadingcolors{grays}
            \conservedresidues{Red1!40}{lightgray!10}{upper}{bf}
            \allmatchresidues{Red1!50}{lightgray!10}{upper}{bf}
            \nomatchresidues{white}{Red1!70}{upper}{bf}
            \setfont{residues}{sf}{up}{bf}{tiny} 
            \setfont{numbering}{sf}{up}{bf}{tiny} 
            \setfont{names}{tt}{up}{bf}{small}
            \setfont{legend}{tt}{up}{bf}{scriptsize}
            \setfont{features}{tt}{up}{bf}{scriptsize}
            \feature{top}{1}{\LENB..\LENC}{brace[black]}{RBD}
            \setends{1}{\LENE..\LEND}
            \showruler{1}{top}
            \hideconsensus
             \hidelegend
          \end{texshade}
        }};
    \end{tikzpicture}};

\end{tikzpicture}
 \else
  \includegraphics[width=.975\textwidth]{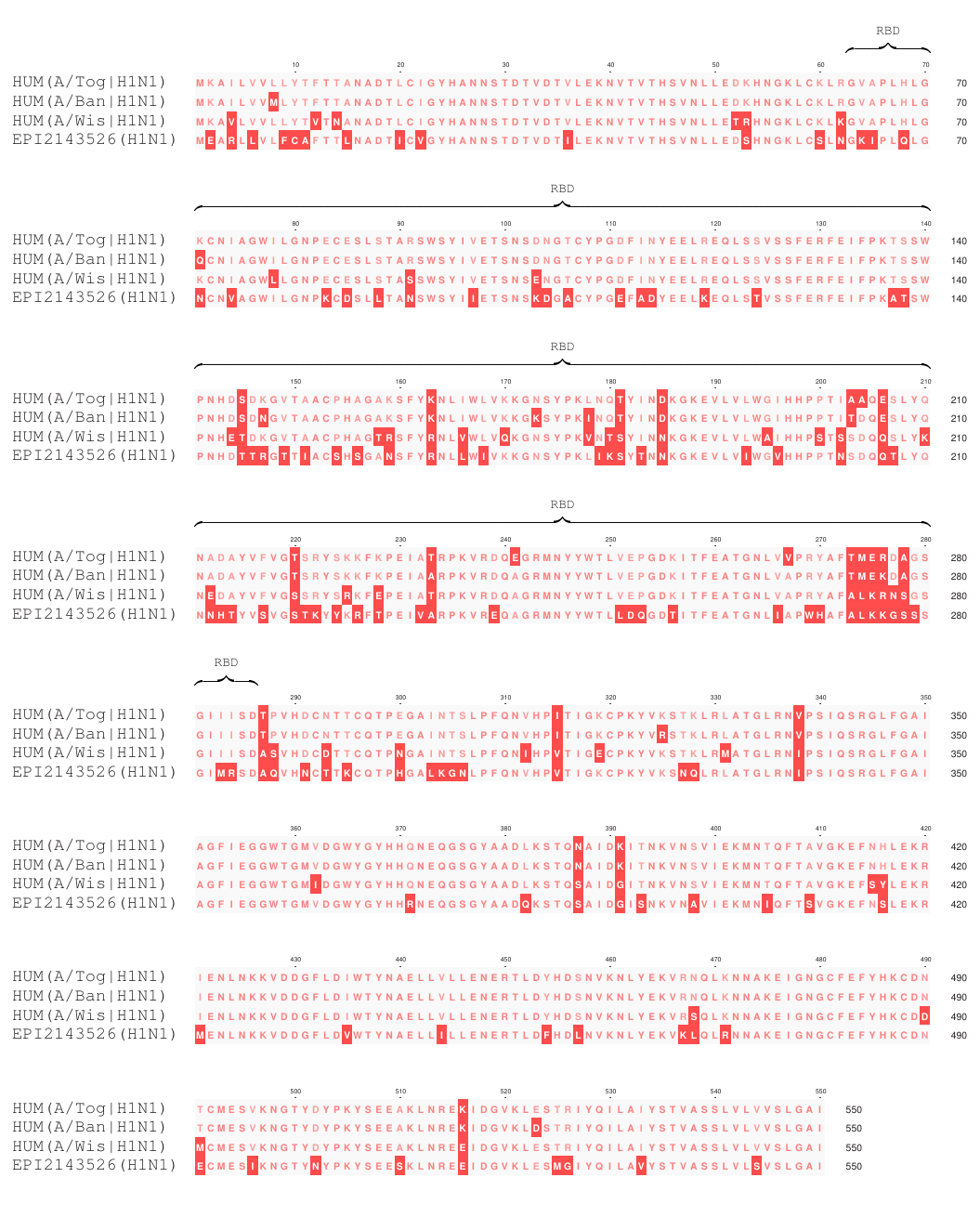}
  \fi 
  \vspace{-18pt}
  
\captionN{\textcolor{black}{\textbf{Low-risk comparison.} HA sequence comparison  with dominant frequency human strains (A/Togo/0172/2021, A/Bangladesh/9004/2021,  A/Wisconsin/04/2021)  with \enet estimated  H1N1 low-risk strain EPI2143526 (emergence score 5.66) showing  substantially more differences compared to \EXTENDED Fig.~\ref{figriskyseq4}.}}\label{figriskyseq4lr}
\end{figure}
\else
\refstepcounter{figure}\label{figriskyseq4lr}
\fi

\ifFIGS
\begin{figure}[t]
    \centering    
    \tikzexternalenable  
    \tikzsetnextfilename{blastvalid}
    \iftikzX 
    \input{Figures/sonet}   
    \vspace{-5pt}    
    \else 
    \includegraphics[width=0.9\textwidth]{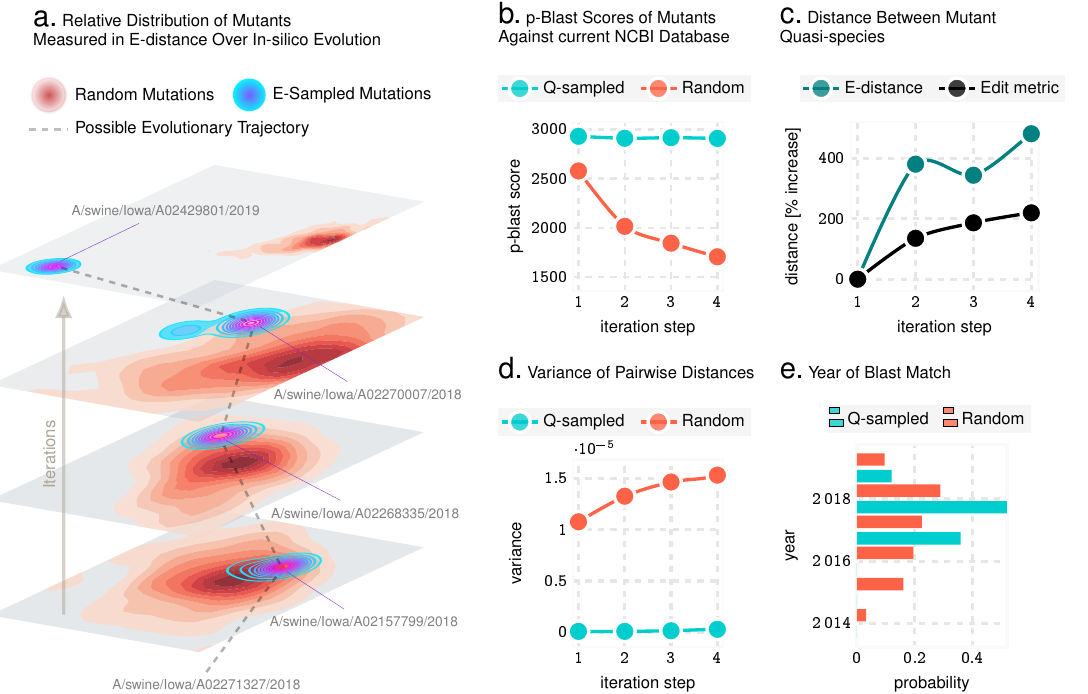}
    \fi
    \captionN{\textbf{\qdist validation in-silico using Influenza A sequences from NCBI database. Panel a} illustrates that the \enet induced modeling of evolutionary trajectories initiated from known haemagglutinnin (HA)  sequences  are distinct from random paths in the strain space. In particular, random trajectories have more variance, and more importantly, diverge to different regions of the landscape compared to \enet predictions. \textbf{Panels b-e} show that unconstrained Q-sampling  produces sequences maintain a higher degree of similarity to known sequences, as verified by blasting against known HA sequences, have a smaller rate of growth of variance, and produce matches in closer time frames to the initial sequence. \textbf{Panel c} shows that this is not due to simply restricting the  mutational variations, which increases rapidly in both the \enet and the classical metric.}\label{figsoa}
\end{figure}
\else
\refstepcounter{figure}\label{figsoa}
\fi

\vskip 5em

\begin{figure}[t]
  \centering    
  \includegraphics[width=\textwidth]{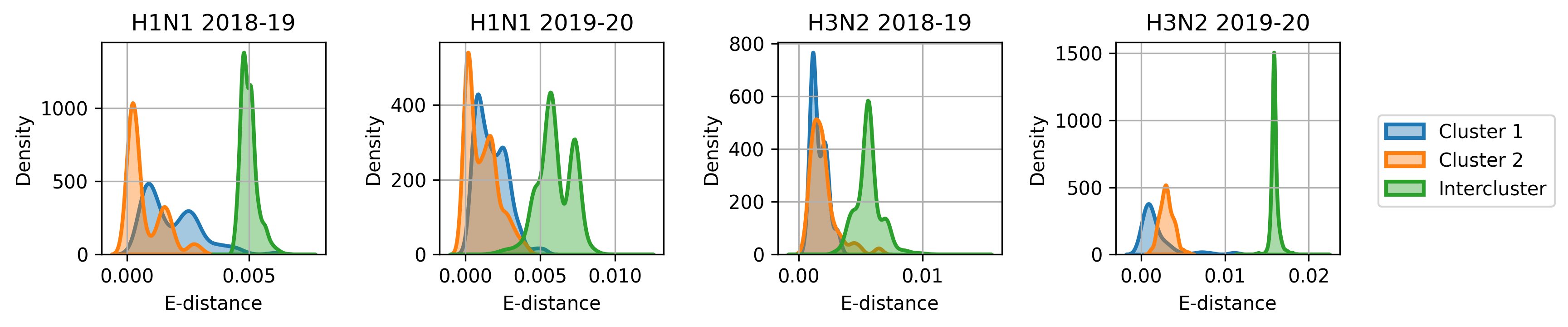}
  \captionN{\textbf{Sampling effects on \qdist.} We trained 100 \qnet models the HA segment with different random samples of 3000 strains. We sample a pair of random strains from each of the two largest clusters and compute their \qdist under each \qnet model. There is little to no overlap between the same cluster E-distances and the inter-cluster distances.}\label{randomsampling}
\end{figure}

\clearpage

\begin{table}[H]
    \centering
    \captionN{IRAT predictions that broadly corroborated with outbreaks}\label{iratoutbreak}\centering
    \sffamily\fontsize{8}{8}\selectfont
    \begin{tabular}{L{1.5in}|L{.45in}|L{.6in}|L{.55in}|L{2.7in}}\hline
        \rowcolor{lightgray!50}Influenza Virus & Subtype & IRAT Emergence Score & IRAT Impact Score & Description \\\hline
        A/swine/Shandong/1207/2016 & H1N1 & 7.5 & 6.9 & Cases of ILI (Influenza-like-illness) saw a substantial increase in weekly cases in the Shandong province of China during the years 2016-2017 when compared to previous years~\cite{yin2023}. This strain was also chosen as a representative strain in a study by Sun, et al. which found that its potential for human infectivity ``greatly enhances the opportunity for virus adaptation in humans and raises concerns for the possible generation of pandemic viruses"~\cite{sun2020}. \\\hline
        A/Ohio/13/2017 & H3N2 & 6.6 & 5.8 & The CDC reports that flu activity in the United States during the 2017–2018 season started to increase in November and was dominated by Influenza A (H3N2) viruses through February~\cite{influenza2017-28}.\\\hline
        A/Hong Kong/125/2017 & H7N9 & 6.5 & 7.5 & The CDC reports that "during March 31, 2013–August 7, 2017, a total of 1,557 human infections with Asian H7N9 viruses were reported; at least 605 (39\%) of these infections resulted in death." 759 of these infections were reported during the fifth epidemic (October 1, 2016–August 7, 2017)~\cite{cdc2017}. \\\hline
    \end{tabular}
\end{table}

\vskip 5em

\ifFIGS
\begin{figure}[H]
    \centering
    \tikzexternalenable
    \tikzsetnextfilename{dom}
    \iftikzX
    \begin{tikzpicture}[font=\bf\sffamily\fontsize{8}{8}\selectfont]
\def\DATAA{Figures/plotdata/h1n1humanHA_dom_ldist.dat}
\def\DATAB{Figures/plotdata/h1n1humanNA_dom_ldist.dat}
\def\DATAC{Figures/plotdata/h3n2humanHA_dom_ldist.dat}
\def\DATAD{Figures/plotdata/h3n2humanNA_dom_ldist.dat}

\node[] (A) at (0,0) {
\begin{tikzpicture}
\begin{groupplot}[group style={
        group size=2 by 2,
        xlabels at=edge bottom,
        xticklabels at=edge bottom,
        vertical sep=0.2in, horizontal sep=.35in,
    },height=2in,width=2in]

\nextgroupplot[ylabel=probability,ylabel style={xshift=-.8in,yshift=-.1in}]
\addplot [
    hist={
        bins=50,
        data min=0,
        data max=30,
density=true
    }  , fill=black!50
] table [y index=0] {\DATAA};
\addlegendentry{HA H1N1}

\nextgroupplot[]
\addplot [
    hist={
        bins=50,
        data min=0,
        data max=30,
density=true
    }    , fill=black!50
] table [y index=0] {\DATAB};
\addlegendentry{NA H1N1}

\nextgroupplot[xlabel=Edit Distance from Dominant Strain,xlabel style={xshift=.8in}]
\addplot [
    hist={
        bins=50,
        data min=0,
        data max=30,
density=true
    }    , fill=black!50
] table [y index=0] {\DATAC};
\addlegendentry{HA H3N2}

\nextgroupplot[]
\addplot [
    hist={
        bins=50,
        data min=0,
        data max=30,
density=true
    }    , fill=black!50
] table [y index=0] {\DATAD};
\addlegendentry{NA H3N2}

\end{groupplot}
\end{tikzpicture}
};

\node [anchor=south west] (L1) at (A.north west) {{\large a.} Distribution around dominant strain};

\end{tikzpicture}  
    \else
    \includegraphics[width=.6\textwidth]{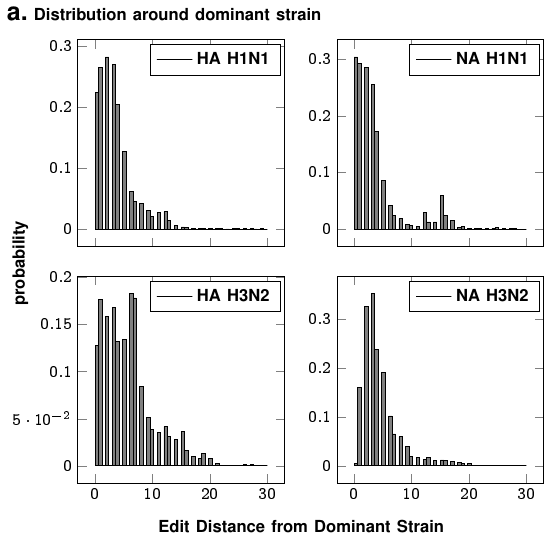}
    \fi
    \captionN{\textbf{No. of mutations from a seasonal dominant (maximal frequency) strain over the years} The quasispecies that circulates each season for each sub-type is tightly distributed around the maximal strains on average.}\label{figdom}
\end{figure}
\else
\refstepcounter{figure}\label{figdom}
\fi

\clearpage

\ifFIGS
\begin{figure}[H]
  \tikzexternalenable
  \tikzsetnextfilename{straindist}
  \centering
 \iftikzX
 \begin{tikzpicture}
   \node[] (A) at (0,0) {   \begin{tikzpicture}

\definecolor{darkgray176}{RGB}{176,176,176}
\definecolor{darkolivegreen85850}{RGB}{85,85,0}
\def\OPC{.4}

\begin{axis}[
tick align=outside,
tick pos=left,
x grid style={darkgray176},
xlabel={emergence score},
xmin=2.73363774372318, xmax=8.2507791550608,
xtick style={color=black},
y grid style={darkgray176},
ylabel={probability density},
ymin=0, ymax=1.8, xmax=8,
ytick style={color=black}
]
\draw[draw=white,fill=darkolivegreen85850,opacity=\OPC] (axis cs:2.98441689878398,0) rectangle (axis cs:3.12008345555513,0.0290014689244097);
\draw[draw=white,fill=darkolivegreen85850,opacity=\OPC] (axis cs:3.12008345555513,0) rectangle (axis cs:3.25575001232628,0.107885464398804);
\draw[draw=white,fill=darkolivegreen85850,opacity=\OPC] (axis cs:3.25575001232628,0) rectangle (axis cs:3.39141656909742,0.04176211525115);
\draw[draw=white,fill=darkolivegreen85850,opacity=\OPC] (axis cs:3.39141656909742,0) rectangle (axis cs:3.52708312586857,0.105565346884851);
\draw[draw=white,fill=darkolivegreen85850,opacity=\OPC] (axis cs:3.52708312586857,0) rectangle (axis cs:3.66274968263972,0.0974449355860166);
\draw[draw=white,fill=darkolivegreen85850,opacity=\OPC] (axis cs:3.66274968263972,0) rectangle (axis cs:3.79841623941087,0.0127606463267403);
\draw[draw=white,fill=darkolivegreen85850,opacity=\OPC] (axis cs:3.79841623941087,0) rectangle (axis cs:3.93408279618202,0.00812041129883469);
\draw[draw=white,fill=darkolivegreen85850,opacity=\OPC] (axis cs:3.93408279618202,0) rectangle (axis cs:4.06974935295317,0.00696035254185835);
\draw[draw=white,fill=darkolivegreen85850,opacity=\OPC] (axis cs:4.06974935295317,0) rectangle (axis cs:4.20541590972432,0.0719236429325361);
\draw[draw=white,fill=darkolivegreen85850,opacity=\OPC] (axis cs:4.20541590972432,0) rectangle (axis cs:4.34108246649547,0.0533627028209135);
\draw[draw=white,fill=darkolivegreen85850,opacity=\OPC] (axis cs:4.34108246649547,0) rectangle (axis cs:4.47674902326662,0.131086639538332);
\draw[draw=white,fill=darkolivegreen85850,opacity=\OPC] (axis cs:4.47674902326662,0) rectangle (axis cs:4.61241558003776,0.0962848768290403);
\draw[draw=white,fill=darkolivegreen85850,opacity=\OPC] (axis cs:4.61241558003776,0) rectangle (axis cs:4.74808213680891,0.153127755920883);
\draw[draw=white,fill=darkolivegreen85850,opacity=\OPC] (axis cs:4.74808213680891,0) rectangle (axis cs:4.88374869358006,0.416461093754524);
\draw[draw=white,fill=darkolivegreen85850,opacity=\OPC] (axis cs:4.88374869358006,0) rectangle (axis cs:5.01941525035121,0.048722467793008);
\draw[draw=white,fill=darkolivegreen85850,opacity=\OPC] (axis cs:5.01941525035121,0) rectangle (axis cs:5.15508180712236,0.0533627028209139);
\draw[draw=white,fill=darkolivegreen85850,opacity=\OPC] (axis cs:5.15508180712236,0) rectangle (axis cs:5.29074836389351,0.121806169482521);
\draw[draw=white,fill=darkolivegreen85850,opacity=\OPC] (axis cs:5.29074836389351,0) rectangle (axis cs:5.42641492066466,0.141527168351119);
\draw[draw=white,fill=darkolivegreen85850,opacity=\OPC] (axis cs:5.42641492066466,0) rectangle (axis cs:5.56208147743581,0.235491927666207);
\draw[draw=white,fill=darkolivegreen85850,opacity=\OPC] (axis cs:5.56208147743581,0) rectangle (axis cs:5.69774803420696,0.116005875697639);
\draw[draw=white,fill=darkolivegreen85850,opacity=\OPC] (axis cs:5.69774803420696,0) rectangle (axis cs:5.8334145909781,0.0858443480162528);
\draw[draw=white,fill=darkolivegreen85850,opacity=\OPC] (axis cs:5.8334145909781,0) rectangle (axis cs:5.96908114774925,0.124126286996474);
\draw[draw=white,fill=darkolivegreen85850,opacity=\OPC] (axis cs:5.96908114774925,0) rectangle (axis cs:6.1047477045204,0.183289283602269);
\draw[draw=white,fill=darkolivegreen85850,opacity=\OPC] (axis cs:6.1047477045204,0) rectangle (axis cs:6.24041426129155,0.712276076783503);
\draw[draw=white,fill=darkolivegreen85850,opacity=\OPC] (axis cs:6.24041426129155,0) rectangle (axis cs:6.3760808180627,0.520866381882399);
\draw[draw=white,fill=darkolivegreen85850,opacity=\OPC] (axis cs:6.3760808180627,0) rectangle (axis cs:6.51174737483385,0.390939801101043);
\draw[draw=white,fill=darkolivegreen85850,opacity=\OPC] (axis cs:6.51174737483385,0) rectangle (axis cs:6.647413931605,1.06957417393222);
\draw[draw=white,fill=darkolivegreen85850,opacity=\OPC] (axis cs:6.647413931605,0) rectangle (axis cs:6.78308048837615,0.786519837229991);
\draw[draw=white,fill=darkolivegreen85850,opacity=\OPC] (axis cs:6.78308048837615,0) rectangle (axis cs:6.9187470451473,0.5846696135161);
\draw[draw=white,fill=darkolivegreen85850,opacity=\OPC] (axis cs:6.9187470451473,0) rectangle (axis cs:7.05441360191845,0.417621152511497);
\draw[draw=white,fill=darkolivegreen85850,opacity=\OPC] (axis cs:7.05441360191845,0) rectangle (axis cs:7.19008015868959,0.142687227108096);
\draw[draw=white,fill=darkolivegreen85850,opacity=\OPC] (axis cs:7.19008015868959,0) rectangle (axis cs:7.32574671546074,0.0881644655302055);
\draw[draw=white,fill=darkolivegreen85850,opacity=\OPC] (axis cs:7.32574671546074,0) rectangle (axis cs:7.46141327223189,0.100925111856946);
\draw[draw=white,fill=darkolivegreen85850,opacity=\OPC] (axis cs:7.46141327223189,0) rectangle (axis cs:7.59707982900304,0.0220411163825514);
\draw[draw=white,fill=darkolivegreen85850,opacity=\OPC] (axis cs:7.59707982900304,0) rectangle (axis cs:7.73274638577419,0.0928047005581117);
\addplot [ultra thick, blue]
table {%
3 1.64282096522451e-96
3.05050505050505 2.13350411808735e-92
3.1010101010101 1.93728372410951e-88
3.15151515151515 1.24829966781793e-84
3.2020202020202 5.78826832604809e-81
3.25252525252525 1.95719403794843e-77
3.3030303030303 4.88675787582173e-74
3.35353535353535 9.11740099094306e-71
3.4040404040404 1.28552385932322e-67
3.45454545454545 1.38450583039176e-64
3.50505050505051 1.15061152129626e-61
3.55555555555556 7.45036072034177e-59
3.60606060606061 3.79342139888932e-56
3.65656565656566 1.53210437455164e-53
3.70707070707071 4.9495608801192e-51
3.75757575757576 1.28918043010318e-48
3.80808080808081 2.72782927858905e-46
3.85858585858586 4.72295935063286e-44
3.90909090909091 6.73751669880276e-42
3.95959595959596 7.97141193207241e-40
4.01010101010101 7.87145578540114e-38
4.06060606060606 6.52641164025328e-36
4.11111111111111 4.569784231891e-34
4.16161616161616 2.71715330348933e-32
4.21212121212121 1.37918981770429e-30
4.26262626262626 6.00653234623986e-29
4.31313131313131 2.2553870818693e-27
4.36363636363636 7.33562853810856e-26
4.41414141414141 2.07592805602813e-24
4.46464646464646 5.13343574595906e-23
4.51515151515152 1.11381544040339e-21
4.56565656565657 2.12885422607899e-20
4.61616161616162 3.59797258070855e-19
4.66666666666667 5.39681111755655e-18
4.71717171717172 7.20961106662641e-17
4.76767676767677 8.60700328790825e-16
4.81818181818182 9.21240536577935e-15
4.86868686868687 8.86825516890493e-14
4.91919191919192 7.70123962543382e-13
4.96969696969697 6.05070581110215e-12
5.02020202020202 4.3131498569394e-11
5.07070707070707 2.79706716925316e-10
5.12121212121212 1.65449820423863e-09
5.17171717171717 8.94911910313761e-09
5.22222222222222 4.43713273894592e-08
5.27272727272727 2.02140776260153e-07
5.32323232323232 8.48049616152464e-07
5.37373737373737 3.28365528498611e-06
5.42424242424242 1.17593697626404e-05
5.47474747474747 3.90293540148479e-05
5.52525252525253 0.000120293128544033
5.57575757575758 0.000344957505106087
5.62626262626263 0.000922086976098236
5.67676767676768 0.0023016536514041
5.72727272727273 0.00537434149537648
5.77777777777778 0.0117586986667198
5.82828282828283 0.0241462968361667
5.87878787878788 0.0466106942029483
5.92929292929293 0.0847090172596881
5.97979797979798 0.145153536434966
6.03030303030303 0.234857667448328
6.08080808080808 0.359308387954052
6.13131313131313 0.520479487288689
6.18181818181818 0.714799395089476
6.23232323232323 0.931886042652547
6.28282828282828 1.15472091893308
6.33333333333333 1.36159739994157
6.38383838383838 1.52962369454664
6.43434343434343 1.63899222320363
6.48484848484848 1.6768867376563
6.53535353535354 1.63994659336175
6.58585858585859 1.5346393478967
6.63636363636364 1.37553766532227
6.68686868686869 1.18210177546558
6.73737373737374 0.974918684166032
6.78787878787879 0.772354353666207
6.83838383838384 0.588290767678703
6.88888888888889 0.431197881010828
6.93939393939394 0.304398147747119
6.98989898989899 0.207133295097946
7.04040404040404 0.135972930000807
7.09090909090909 0.0861772140575356
7.14141414141414 0.0527720204722758
7.19191919191919 0.0312472207072691
7.24242424242424 0.0179032619138057
7.29292929292929 0.00993286589046611
7.34343434343434 0.00533997536086925
7.39393939393939 0.00278367926314787
7.44444444444444 0.00140799128045752
7.49494949494949 0.000691448679172268
7.54545454545454 0.000329891080314846
7.5959595959596 0.00015300178580084
7.64646464646465 6.90233215071097e-05
7.6969696969697 3.03054272413675e-05
7.74747474747475 1.29573379449716e-05
7.7979797979798 5.3978600549771e-06
7.84848484848485 2.19215753380118e-06
7.8989898989899 8.68349228741695e-07
7.94949494949495 3.35669873767083e-07
8 1.2669020596184e-07
};
\addplot [ultra thick, red, dashed]
table {%
7.1050425349515 0
7.1050425349515 1.8
} node [rotate=90, anchor=north, xshift=-1in,align=center] {95\% significance: 7.1,\\p-value: $7\times 10^{-3}$};
\end{axis}

\end{tikzpicture}};
   \node[anchor=south west] (B) at ([xshift=.20in]A.south east) {\begin{tikzpicture}

\definecolor{darkgray176}{RGB}{176,176,176}
\definecolor{darkolivegreen85850}{RGB}{85,85,0}
\def\OPC{.4}

\begin{axis}[
tick align=outside,
tick pos=left,
x grid style={darkgray176},
xlabel={emergence score},
xmin=3.8, xmax=8.2,
xtick style={color=black},
y grid style={darkgray176},
ylabel={probability density},
ymin=0, ymax=2.7,
xmin=6,xmax=8,
ytick style={color=black}
]
\draw[draw=white,fill=darkolivegreen85850,opacity=\OPC] (axis cs:6.25864263504708,0) rectangle (axis cs:6.33234782258344,0.737367720014954);
\draw[draw=white,fill=darkolivegreen85850,opacity=\OPC] (axis cs:6.33234782258344,0) rectangle (axis cs:6.4060530101198,0.147473544002989);
\draw[draw=white,fill=darkolivegreen85850,opacity=\OPC] (axis cs:6.40605301011979,0) rectangle (axis cs:6.47975819765615,0.589894176011963);
\draw[draw=white,fill=darkolivegreen85850,opacity=\OPC] (axis cs:6.47975819765615,0) rectangle (axis cs:6.55346338519251,1.03231480802092);
\draw[draw=white,fill=darkolivegreen85850,opacity=\OPC] (axis cs:6.55346338519251,0) rectangle (axis cs:6.62716857272886,0.884841264017944);
\draw[draw=white,fill=darkolivegreen85850,opacity=\OPC] (axis cs:6.62716857272886,0) rectangle (axis cs:6.70087376026522,1.76968252803587);
\draw[draw=white,fill=darkolivegreen85850,opacity=\OPC] (axis cs:6.70087376026522,0) rectangle (axis cs:6.77457894780157,1.03231480802094);
\draw[draw=white,fill=darkolivegreen85850,opacity=\OPC] (axis cs:6.77457894780157,0) rectangle (axis cs:6.84828413533793,4.57167986409266);
\draw[draw=white,fill=darkolivegreen85850,opacity=\OPC] (axis cs:6.84828413533793,0) rectangle (axis cs:6.92198932287428,0.589894176011963);
\draw[draw=white,fill=darkolivegreen85850,opacity=\OPC] (axis cs:6.92198932287428,0) rectangle (axis cs:6.99569451041064,0.442420632008967);
\draw[draw=white,fill=darkolivegreen85850,opacity=\OPC] (axis cs:6.99569451041064,0) rectangle (axis cs:7.06939969794699,0);
\draw[draw=white,fill=darkolivegreen85850,opacity=\OPC] (axis cs:7.06939969794699,0) rectangle (axis cs:7.14310488548335,0);
\draw[draw=white,fill=darkolivegreen85850,opacity=\OPC] (axis cs:7.14310488548335,0) rectangle (axis cs:7.2168100730197,0.294947088005978);
\draw[draw=white,fill=darkolivegreen85850,opacity=\OPC] (axis cs:7.2168100730197,0) rectangle (axis cs:7.29051526055606,0);
\draw[draw=white,fill=darkolivegreen85850,opacity=\OPC] (axis cs:7.29051526055606,0) rectangle (axis cs:7.36422044809241,0.147473544002989);
\draw[draw=white,fill=darkolivegreen85850,opacity=\OPC] (axis cs:7.36422044809241,0) rectangle (axis cs:7.43792563562877,0.442420632008972);
\draw[draw=white,fill=darkolivegreen85850,opacity=\OPC] (axis cs:7.43792563562877,0) rectangle (axis cs:7.51163082316512,0);
\draw[draw=white,fill=darkolivegreen85850,opacity=\OPC] (axis cs:7.51163082316512,0) rectangle (axis cs:7.58533601070148,0);
\draw[draw=white,fill=darkolivegreen85850,opacity=\OPC] (axis cs:7.58533601070148,0) rectangle (axis cs:7.65904119823783,0.294947088005978);
\draw[draw=white,fill=darkolivegreen85850,opacity=\OPC] (axis cs:7.65904119823783,0) rectangle (axis cs:7.73274638577419,0.589894176011963);
\addplot [ultra thick, blue]
table {%
4 2.48461515655749e-58
4.04040404040404 1.35950110429041e-56
4.08080808080808 6.99930189466505e-55
4.12121212121212 3.39066289707433e-53
4.16161616161616 1.5455007803732e-51
4.2020202020202 6.62839788639855e-50
4.24242424242424 2.67486974388844e-48
4.28282828282828 1.01566721343991e-46
4.32323232323232 3.62873250473985e-45
4.36363636363636 1.2198690237974e-43
4.4040404040404 3.85856817991897e-42
4.44444444444444 1.14840184196259e-40
4.48484848484848 3.2160019915496e-39
4.52525252525253 8.47409765284221e-38
4.56565656565657 2.10099658477941e-36
4.60606060606061 4.9013075561697e-35
4.64646464646465 1.07585388090113e-33
4.68686868686869 2.22202710913801e-32
4.72727272727273 4.31817403945319e-31
4.76767676767677 7.89597305964355e-30
4.80808080808081 1.35851957439073e-28
4.84848484848485 2.19928174603079e-27
4.88888888888889 3.35004441968198e-26
4.92929292929293 4.80147849717136e-25
4.96969696969697 6.47521417399091e-24
5.01010101010101 8.21652183814397e-23
5.05050505050505 9.81017188517657e-22
5.09090909090909 1.1020973202597e-20
5.13131313131313 1.16497874777579e-19
5.17171717171717 1.1586994101632e-18
5.21212121212121 1.08437200191155e-17
5.25252525252525 9.54861828385559e-17
5.29292929292929 7.91147564764875e-16
5.33333333333333 6.16778442097801e-15
5.37373737373737 4.52434383508725e-14
5.41414141414141 3.12274661282383e-13
5.45454545454546 2.02802178656815e-12
5.49494949494949 1.23926222898549e-11
5.53535353535354 7.12538830976495e-11
5.57575757575758 3.85485979866049e-10
5.61616161616162 1.96229078578296e-09
5.65656565656566 9.39880996113549e-09
5.6969696969697 4.23581612350895e-08
5.73737373737374 1.79620564341908e-07
5.77777777777778 7.16687352839029e-07
5.81818181818182 2.6906558433686e-06
5.85858585858586 9.50476293171839e-06
5.8989898989899 3.15921478690733e-05
5.93939393939394 9.88033669799871e-05
5.97979797979798 0.00029074955744319
6.02020202020202 0.000805046760761499
6.06060606060606 0.00209738364997321
6.1010101010101 0.0051414945072321
6.14141414141414 0.0118592056693428
6.18181818181818 0.0257381040591505
6.22222222222222 0.0525596216881004
6.26262626262626 0.100990991398713
6.3030303030303 0.182586123107264
6.34343434343434 0.310604415765306
6.38383838383838 0.497166904595274
6.42424242424242 0.748775322239817
6.46464646464647 1.06109815971892
6.50505050505051 1.41486279529048
6.54545454545455 1.7751204153991
6.58585858585859 2.09554039797798
6.62626262626263 2.32765699708683
6.66666666666667 2.43274537443365
6.70707070707071 2.39237389309643
6.74747474747475 2.21368682475444
6.78787878787879 1.92733865947262
6.82828282828283 1.57889995374277
6.86868686868687 1.21704298652965
6.90909090909091 0.882697610008487
6.94949494949495 0.602383004303315
6.98989898989899 0.38680154161365
7.03030303030303 0.233699836861674
7.07070707070707 0.132856665732593
7.11111111111111 0.0710661871494866
7.15151515151515 0.0357682266120685
7.19191919191919 0.0169389516436896
7.23232323232323 0.00754797364022816
7.27272727272727 0.00316467377483218
7.31313131313131 0.00124848189334709
7.35353535353535 0.000463436542030452
7.39393939393939 0.000161865028665661
7.43434343434343 5.31949627080015e-05
7.47474747474747 1.64491217753463e-05
7.51515151515152 4.78596644608046e-06
7.55555555555556 1.31024137629966e-06
7.5959595959596 3.37510817713263e-07
7.63636363636364 8.18048045018801e-08
7.67676767676768 1.86562628425096e-08
7.71717171717172 4.00336528761438e-09
7.75757575757576 8.08314767690916e-10
7.7979797979798 1.53564395386765e-10
7.83838383838384 2.74508174556529e-11
7.87878787878788 4.61715820944807e-12
7.91919191919192 7.30716572598043e-13
7.95959595959596 1.08812273647639e-13
8 1.52461974384175e-14
};
\addplot [ultra thick, red, dashed]
table {%
7.1 0
7.1 3
}node [rotate=90, anchor=north, xshift=-1in,align=center] {95\% significance: 7.1,\\p-value: $6\times 10^{-3}$};
\end{axis}

\end{tikzpicture}};

   \node[anchor=south west,align=left] at (A.north west) {{\Large a.} All animal strains \\ (normal dist. fitted to values $> 6.25$)};
   \node[anchor=south west] at (B.north west) {{\Large b.} Strains $> 6.25$
unique up to 15 edits};
   \end{tikzpicture}
 \else
  \includegraphics[width=\textwidth]{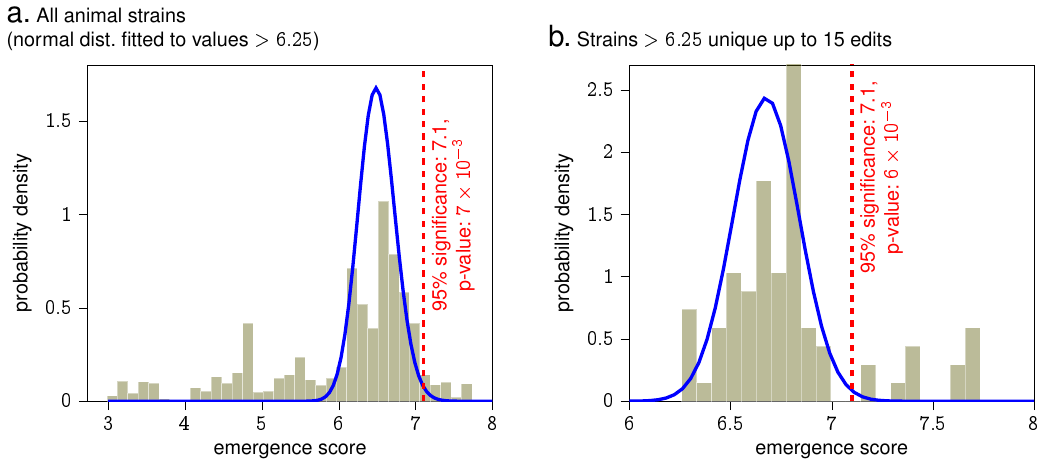}
  \fi 
  \vspace{-15pt}
  
    \captionN{\textcolor{black}{\textbf{Distribution of \enet predicted emergence scores for animal strains.} \textbf{Panel a} We plot the computed emergence risk score distribution of $6,354$ animal strains collected post-2020. \textbf{Panel b} We  show the distribution of computed emergence scores only for strains which are unique up to 15 edits in their HA sequence. Note that the right tail, where the high risk strains lie, are extreme occurrences in a statistically significant sense with p-values lower that 0.01.  %
      }}\label{animalfrequency}
\end{figure}
\else
\refstepcounter{figure}\label{animalfrequency}
\fi

\vskip 5em

\begin{figure}[H]
  \tikzexternalenable
  \tikzsetnextfilename{varianttime}
  \centering
 \tikzXtrue
  \begin{tikzpicture}

\definecolor{black17}{RGB}{17,17,17}
\definecolor{darkcyan}{RGB}{0,139,139}
\definecolor{darkgoldenrod}{RGB}{184,134,11}
\definecolor{darkorchid}{RGB}{153,50,204}
\definecolor{gold}{RGB}{25,215,200}
\definecolor{lightgray203}{RGB}{203,203,203}
\definecolor{lightgray204}{RGB}{204,204,204}
\definecolor{limegreen}{RGB}{50,175,50}
\definecolor{orangered}{RGB}{255,69,0}
\definecolor{slateblue}{RGB}{106,90,205}
\definecolor{steelblue}{RGB}{70,130,180}
\definecolor{whitesmoke240}{RGB}{240,240,240}

\begin{axis}[
axis background/.style={fill=gray!1},
axis line style={black, semithick,},
legend cell align={left},
legend style={
  fill opacity=0.80,
  draw opacity=1,
  text opacity=1,
  at={(1.1,1)},
  anchor=north west,
  draw=lightgray204,
  fill=whitesmoke240
},
height=3in,
width=3in,
tick align=outside,
tick pos=left,
x grid style={lightgray203},
xlabel={time to collection [months]},
xmajorgrids,
xmin=-0.2, xmax=4.2,
xtick style={color=black},
xtick={-1,0,1,2,3,4,5},
xticklabels={0,12,6,3,1,0,},
y grid style={lightgray203},
ylabel={emergence score},
ymajorgrids,
ymin=5.33960029394284, ymax=7.35802255467382,
ytick style={color=black}
]
\addplot [ultra thick, black17, dashed]
table {%
0 6.30109495956234
1 6.54492133939915
2 6.57975253377974
3 6.67097358733071
4 6.73897136145912
};
\addlegendentry{mean behavior}
\addplot [very thick, steelblue]
table {%
0 5.5170892375892
1 5.80433607249119
2 5.83113578260054
3 5.96530504442993
4 5.96530504442993
};
\addlegendentry{A/Hessen/47/2020}
\addplot [very thick, limegreen]
table {%
0 6.2204869006929
1 6.2204869006929
2 6.2204869006929
3 6.2204869006929
4 6.2204869006929
};
\addlegendentry{A/Ohio/09/2015}
\addplot [very thick, gold]
table {%
0 5.85204413110674
1 6.06681921845406
2 6.06681921845406
3 6.14456031860239
4 6.58121809692394
};
\addlegendentry{A/Wisconsin/71/2016}
\addplot [very thick, slateblue]
table {%
0 7.22478447358162
1 7.22478447358162
2 7.22478447358162
3 7.22478447358162
4 7.22478447358162
};
\addlegendentry{A/Iowa/39/2015}
\addplot [very thick, orangered]
table {%
0 7.11181784477211
1 7.24126198332346
2 7.26627608827696
3 7.26627608827696
4 7.26627608827696
};
\addlegendentry{A/Ohio/35/2017}
\addplot [very thick, darkcyan]
table {%
0 6.021641707042
1 6.47735165074596
2 6.70418739072788
3 6.70418739072788
4 6.81151180543359
};
\addlegendentry{A/Iowa/32/2016}
\addplot [very thick, darkgoldenrod]
table {%
0 5.4313467603397
1 6.29478179452952
2 6.29478179452952
3 6.81263986095956
4 6.81263986095956
};
\addlegendentry{A/Hunan/42443/2015}
\addplot [very thick, darkorchid]
table {%
0 7.02954862137446
1 7.02954862137446
2 7.02954862137446
3 7.02954862137446
4 7.02954862137446
};
\addlegendentry{A/Minnesota/19/2011}
\end{axis}

\end{tikzpicture}
 \else
  \includegraphics[width=.975\textwidth]{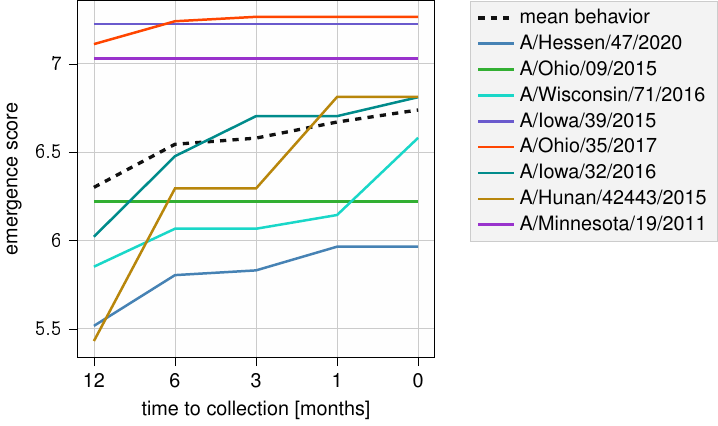}
  \fi 
    \captionN{\textcolor{black}{\textbf{\enet predicted risk scores for top 8 risky variant strains.} We evaluate  H1N1 and H1N2 animal strains that emerged in humans with \enet at twelve, six, three, one, and zero months before their collection date, showing how often risk can be observed to increase over time.}}\label{figvariantrisk}
\end{figure}

\clearpage

\begin{table}[H]
    \centering
    \captionN{Emergence score estimations for fifteen H1N1 and H1N2 variant strains (animal strains collected in humans)}\label{tabvariant}\centering
    \sffamily\fontsize{8}{8}\selectfont
   \begin{tabular}{L{2in}|L{.6in}|L{.6in}|L{.60in}|L{.6in}|L{.6in}|L{.6in}}\hline
\rowcolor{lightgray!50} Variant  strain &Subtype& 12  months  prior & 6  months  prior & 3  months  prior & 1  month  prior & At  collection \\\hline
A/Ohio/24/2017&H1N2&7.5&7.6&7.6&7.6&7.6\\\hline
A/Ohio/35/2017&H1N2&7.1&7.2&7.3&7.3&7.3\\\hline
A/Iowa/39/2015&H1N1&7.2&7.2&7.2&7.2&7.2\\\hline
A/Minnesota/19/2011&H1N2&7.0&7.0&7.0&7.0&7.0\\\hline
A/Hunan/42443/2015&H1N1&5.4&6.3&6.3&6.8&6.8\\\hline
A/Iowa/32/2016&H1N2&6.0&6.5&6.7&6.7&6.8\\\hline
A/Wisconsin/71/2016&H1N2&5.8&6.1&6.1&6.1&6.6\\\hline
A/Ohio/09/2015&H1N1&6.2&6.2&6.2&6.2&6.2\\\hline
A/Michigan/383/2018&H1N2&4.5&5.2&6.1&6.1&6.1\\\hline
A/Hessen/47/2020&H1N1&5.5&5.8&5.8&6.0&6.0\\\hline
A/Wisconsin/03/2021&H1N1&5.9&5.9&5.9&5.9&5.9\\\hline
A/California/71/2021&H1N2&5.3&5.3&5.4&5.5&5.5\\\hline
A/Netherlands/10370-1b/2020&H1N1&5.0&5.0&5.0&5.0&5.0\\\hline
A/Netherlands/3315/2016&H1N1&3.3&4.1&4.1&4.1&4.3\\\hline
A/Bretagne/24241/2021&H1N2&4.2&4.2&4.2&4.2&4.2\\
\hline\end{tabular}

\end{table}

\vskip 5em

\begin{table}[H]
    \centering
    \captionN{High risk strain percentage comparison between variant pool and animal surveillance }\label{tabvariantpc}\centering
    \sffamily\fontsize{8}{8}\selectfont
   \begin{tabular}{L{2.5in}|L{.6in}|L{.60in}|L{.6in}|L{.6in}|L{.6in}}\hline
\rowcolor{lightgray!50} Emergence  score  thresholds &6.0&6.5&6.8&7.0&7.1\\\hline
 \%  in  variants &60.0&46.7&40.0&26.7&20.0\\\hline
 \%  in  animal  surveillance &69.1&45.5&18.5&6.4&5.6\\
\hline\end{tabular}

\end{table}

\clearpage

\begin{table}[H]
    \centering
    \captionN{General linear model for evaluating effect of data diversity on \enet performance}\label{tabreg}\centering
    \sffamily\fontsize{8}{8}\selectfont
    \begin{tabular}{L{1.5in}|L{2in}}\hline
        \rowcolor{lightgray!50}  Variable Name & Description \\\hline
        enet\_complexity & Cumulative number of nodes in all predictors in the corresponding \enet \\\hline
        data\_diversity &  Number of clusters in set of input sequence where each sequence in a specific cluster is separated by at least $5$ mutations from sequences not in the cluster\\\hline
        ldistance\_WHO & Deviation of WHO predicted strain from the dominant strain\\\hline
\end{tabular}
\vskip 1em
\mnp{6.5in}{
    \fontsize{8}{8}\selectfont
\begin{verbatim}
model:dev ~ enet_complexity + data_diversity + enet_complexity * data_diversity + ldistance_WHO 
                 Generalized Linear Model Regression Results                  
==============================================================================
Dep. Variable:                    dev   No. Observations:                  235
Model:                            GLM   Df Residuals:                      230
Model Family:                Gaussian   Df Model:                            4
Link Function:               identity   Scale:                          23.214
Method:                          IRLS   Log-Likelihood:                -700.43
Date:                Thu, 11 Jun 2020   Deviance:                       5339.2
Time:                        16:45:46   Pearson chi2:                 5.34e+03
No. Iterations:                     3   Covariance Type:             nonrobust
==================================================================================================
                                     coef    std err          z      P>|z|      [0.025      0.975]
--------------------------------------------------------------------------------------------------
Intercept                         -0.1116      1.090     -0.102      0.918      -2.248       2.025
enet_complexity                    0.0005      0.000      1.075      0.282      -0.000       0.001
data_diversity                     0.3197      0.126      2.531      0.011       0.072       0.567
enet_complexity:data_diversity -6.932e-05   5.01e-05     -1.383      0.167      -0.000    2.89e-05
ldistance_WHO                     -0.0348      0.035     -1.007      0.314      -0.102       0.033
==================================================================================================
\end{verbatim}
}
\vskip 2em
\mnp{6.5in}{
    \fontsize{8}{8}\selectfont
\begin{verbatim}
model:dev ~ enet_complexity + data_diversity + ldistance_WHO
                 Generalized Linear Model Regression Results                  
==============================================================================
Dep. Variable:                    dev   No. Observations:                  235
Model:                            GLM   Df Residuals:                      231
Model Family:                Gaussian   Df Model:                            3
Link Function:               identity   Scale:                          23.306
Method:                          IRLS   Log-Likelihood:                -701.41
Date:                Thu, 11 Jun 2020   Deviance:                       5383.6
Time:                        16:45:47   Pearson chi2:                 5.38e+03
No. Iterations:                     3   Covariance Type:             nonrobust
===================================================================================
                      coef    std err          z      P>|z|      [0.025      0.975]
-----------------------------------------------------------------------------------
Intercept           1.0841      0.665      1.630      0.103      -0.219       2.387
enet_complexity  -4.12e-05      0.000     -0.156      0.876      -0.001       0.000
data_diversity      0.1788      0.075      2.392      0.017       0.032       0.325
ldistance_WHO      -0.0695      0.024     -2.930      0.003      -0.116      -0.023
===================================================================================
\end{verbatim}
}
\end{table}

\end{document}